\documentclass[lettersize,journal]{IEEEtran}

\usepackage{amsmath}
\usepackage{amsfonts}
\usepackage{amsthm}
\usepackage{amssymb}

\usepackage{url}

\usepackage{graphicx}
\graphicspath{ {figs/} }
\usepackage{subfigure}

\usepackage[ruled,linesnumbered,vlined]{algorithm2e}

\usepackage{xcolor}

\usepackage{tabularray}
\UseTblrLibrary{diagbox}
\usepackage[export]{adjustbox}

\newtheorem{definition}{Definition}
\newtheorem{lemma}{Lemma}
\newtheorem{theorem}{Theorem}

\usepackage{cite}

\begin{document}
%-------------------------------------------------------------------------------

\title{Local Differential Privacy is Not Enough: A Sample Reconstruction Attack against Federated Learning with Local Differential Privacy}

\author{Zhichao~You, 
Xuewen~Dong,~\IEEEmembership{Member,~IEEE,} 
Shujun~Li,~\IEEEmembership{Senior~Member,~IEEE,}
Ximeng~Liu,~\IEEEmembership{Senior~Member,~IEEE,}  
Siqi~Ma,~\IEEEmembership{Member,~IEEE,} 
Yulong~Shen,~\IEEEmembership{Member,~IEEE}
% <-this % stops a space
\IEEEcompsocitemizethanks{
\IEEEcompsocthanksitem This work was supported in part by the National Key R\&D Program of China (No. 2023YFB3107500), National Natural Science Foundation of China (No. 62220106004, 62232013), Technology Innovation Leading Program of Shaanxi (No. 2022KXJ-093, 2023KXJ-033), and Innovation Fund of Xidian (No. YJSJ24015). (Corresponding author: Xuewen Dong.)
\IEEEcompsocthanksitem Zhichao~You, Xuewen~Dong and  Yulong~Shen are with the School of Computer Science \& Technology, Xidian University, and are with the Shaanxi Key Laboratory of Network and System Security, Xi'an, China (email: zcyou@stu.xidian.edu.cn, xwdong@xidian.edu.cn, ylshen@mail.xidian.edu.cn).
\IEEEcompsocthanksitem  Shujun~Li is with the School of Computing and the Institute of Cyber Security for Society (iCSS), University of Kent, Canterbury, UK (e-mail: S.J.Li@kent.ac.uk).
\IEEEcompsocthanksitem  Ximeng~Liu is with the College of Computer and Data Science, Fuzhou University, Fuzhou, China (e-mail: snbnix@gmail.com).
\IEEEcompsocthanksitem Siqi Ma is with the School of System and Computing, University of New South Wales, Canberra, ACT 2612, Australia (e-mail: siqi.ma@unsw.edu.au).
}
}

\maketitle
%-------------------------------------------------------------------------------
\begin{abstract}
Reconstruction attacks against federated learning (FL) aim to reconstruct users' samples through users' uploaded gradients. Local differential privacy (LDP) is regarded as an effective defense against various attacks, including sample reconstruction in FL, where gradients are clipped and perturbed. Existing attacks are ineffective in FL with LDP since clipped and perturbed gradients obliterate most sample information for reconstruction. Besides, existing attacks embed additional sample information into gradients to improve the attack effect and cause gradient expansion, leading to a more severe gradient clipping in FL with LDP. In this paper, we propose a sample reconstruction attack against LDP-based FL with any target models to reconstruct victims' sensitive samples to illustrate that FL with LDP is not flawless. Considering gradient expansion in reconstruction attacks and noise in LDP, the core of the proposed attack is gradient compression and reconstructed sample denoising. For gradient compression, an inference structure based on sample characteristics is presented to reduce redundant gradients against LDP. For reconstructed sample denoising, we artificially introduce zero gradients to observe noise distribution and scale confidence interval to filter the noise. Theoretical proof guarantees the effectiveness of the proposed attack. Evaluations show that the proposed attack is the only attack that reconstructs victims' training samples in LDP-based FL and has little impact on the target model's accuracy. We conclude that LDP-based FL needs further improvements to defend against sample reconstruction attacks effectively.
\end{abstract}

\section{Introduction}
\label{section-introduction}
\IEEEPARstart{F}{ederated} learning (FL) is a distributed learning framework in which users train a given global model through local samples without uploading these samples to the server or the platform~\cite{Kairouz2019open,mu2024feddmc,mu2023fedproc}. The server updates the global model according to gradients or model updates uploaded by users. Since the server trains machine learning models without collecting users' training samples, FL effectively protects users' privacy and reduces communication overhead. FL has been applied to multiple platforms for user privacy preservation, e.g., word prediction of Google Gboard~\cite{GBoard2017}, automatic speech recognition of Apple's Siri~\cite{apple2021}, and credit evaluation of WeBank~\cite{FATE2021}.

\textit{Reconstruction attacks} against gradients uploaded by users cause a severe training sample leakage issue in FL (e.g., \cite{Zhu2020DLG, Yang2023compress, Yin2021see, fowl2022robbing, Boenisch2021When, Wei2020Framework}). Although users do not share private training samples, the information about the samples is implied in the gradients generated by the local samples. Combined with reconstruction attacks, adversaries can reconstruct users' original samples according to the uploaded gradients with high accuracy and quality, which is shown in Fig.~\ref{fig-introduction-robbing-no-dp}.

\begin{figure}
\centering
\subfigure[]{
    \centering
    \label{fig-introduction-robbing-no-dp}
    \includegraphics*[width=0.67in]{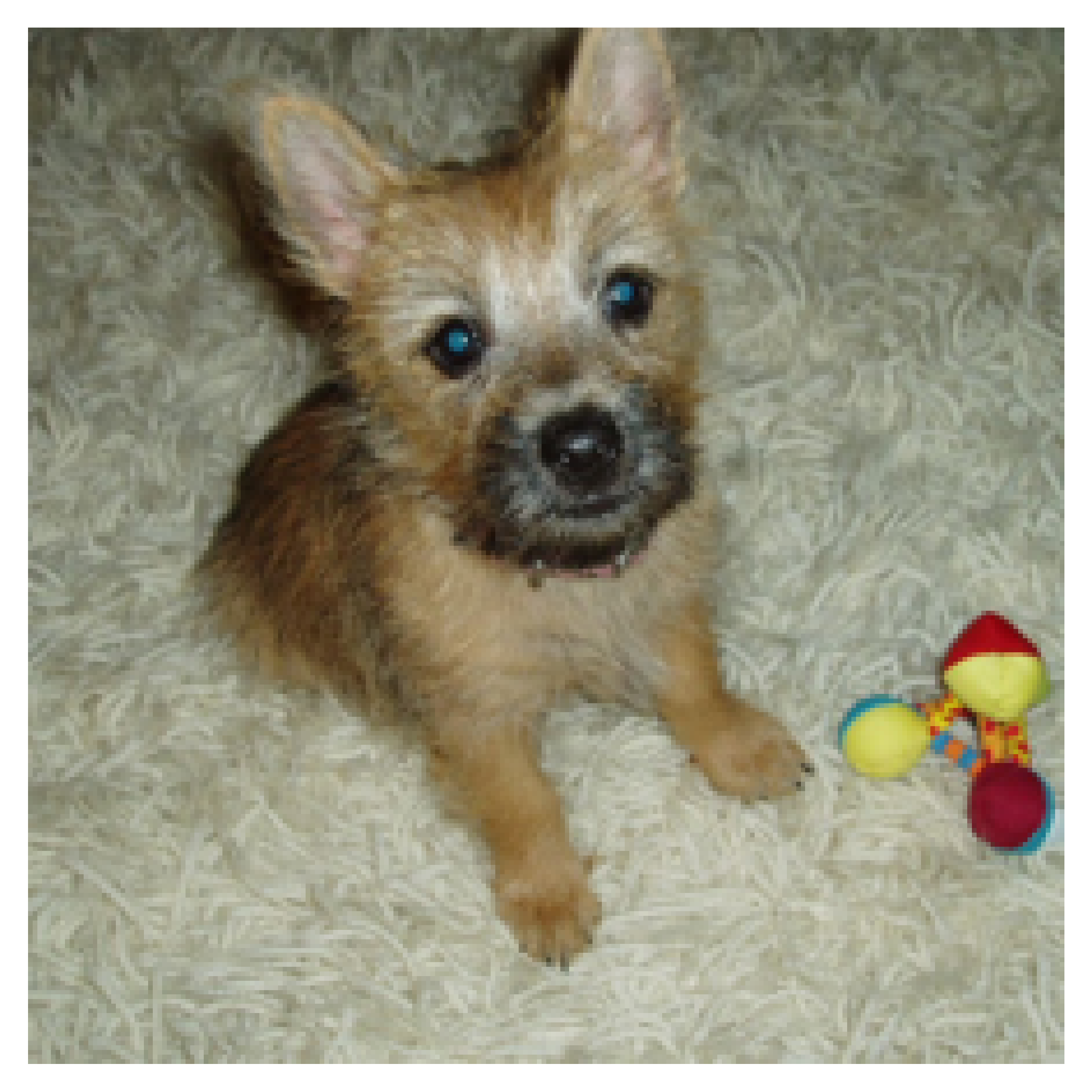}
}\hspace{-0.14in}
\subfigure[]{
    \centering
    \label{fig-introduction-robbing-dp}
    \includegraphics*[width=0.67in]{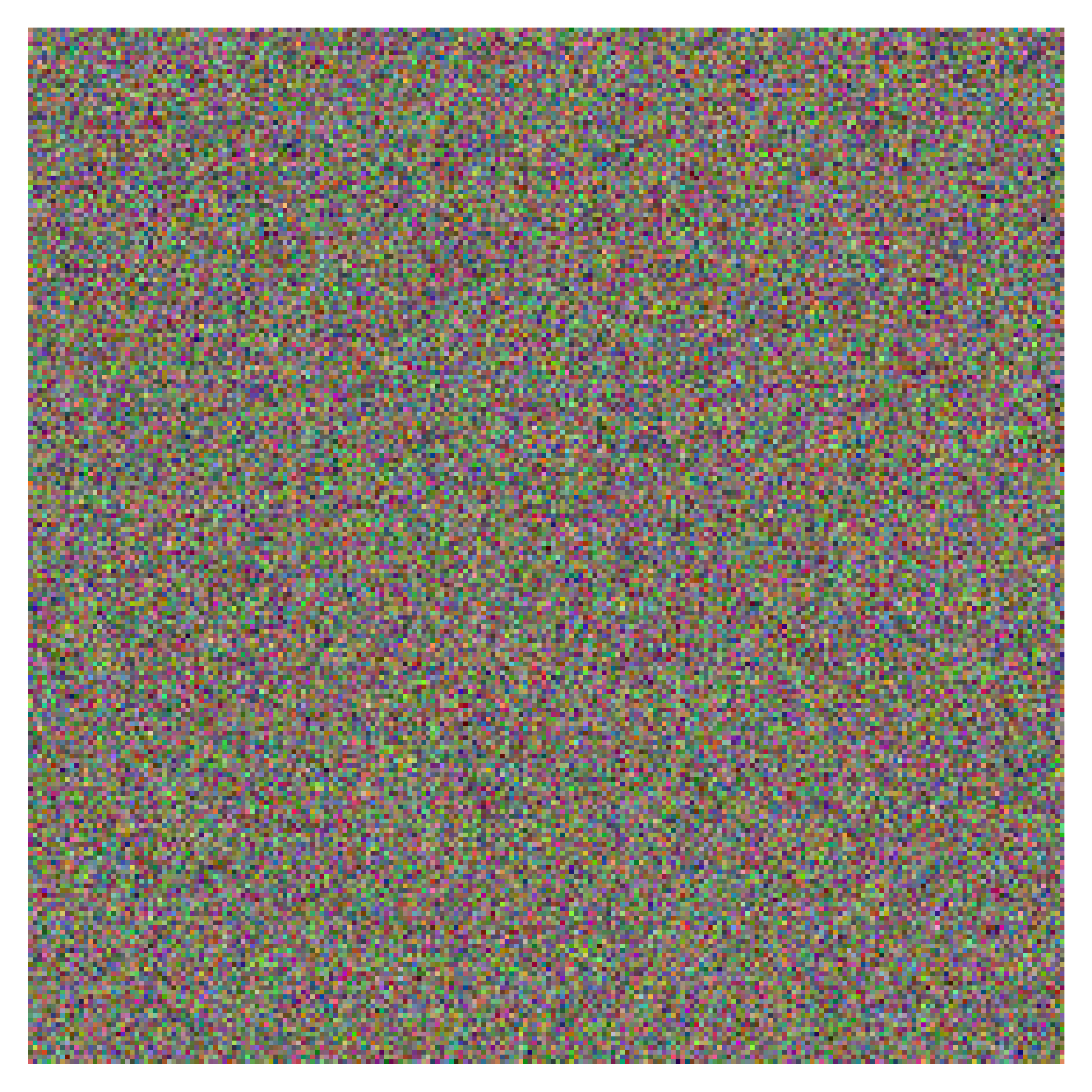}
}\hspace{-0.14in}
    \subfigure[]{
    \centering
    \label{fig-introduction-proposed-dp}
    \includegraphics*[width=0.67in]{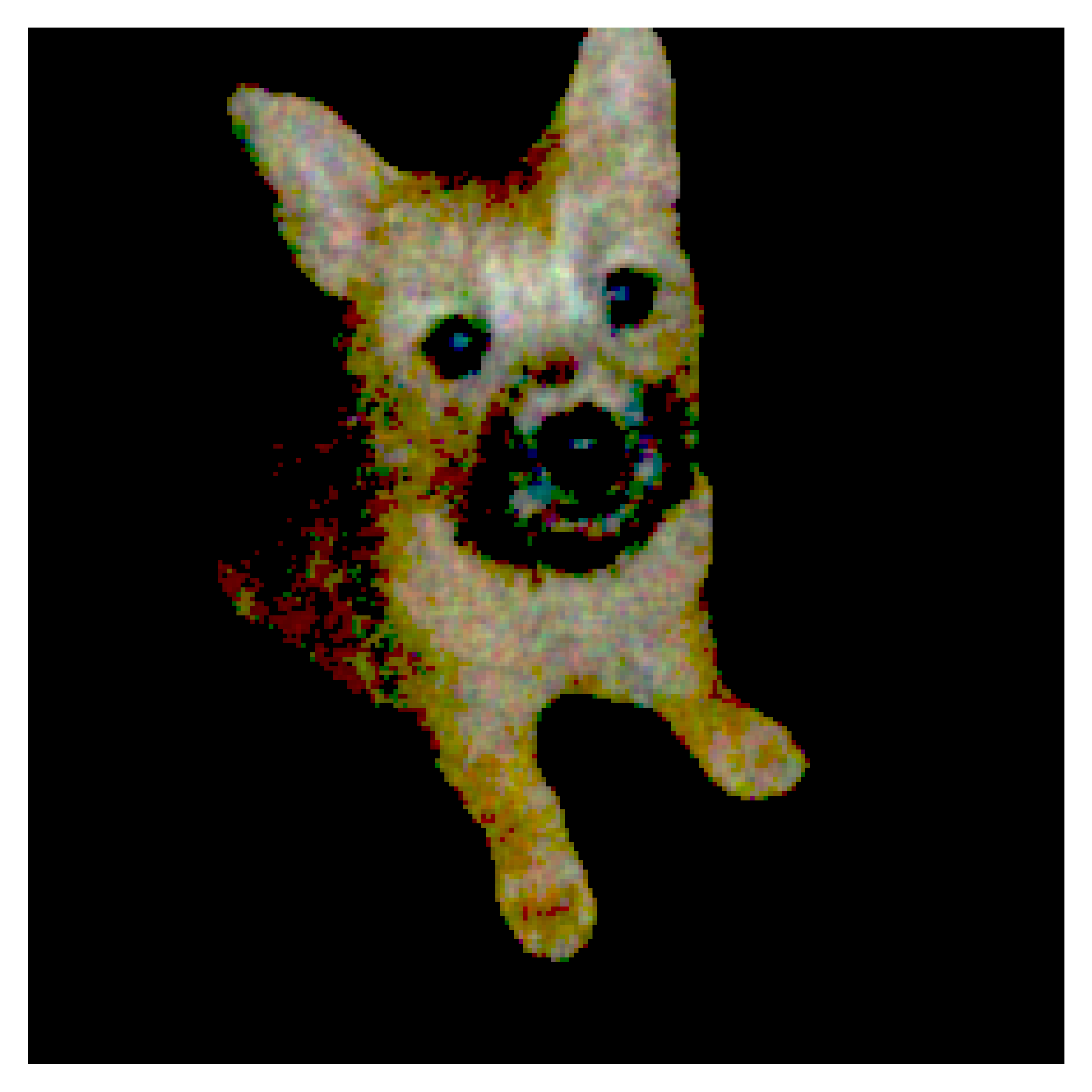}
}\hspace{-0.14in}
\subfigure[]{
    \centering
    \label{fig-introduction-overlapped}
    \includegraphics*[width=0.67in]{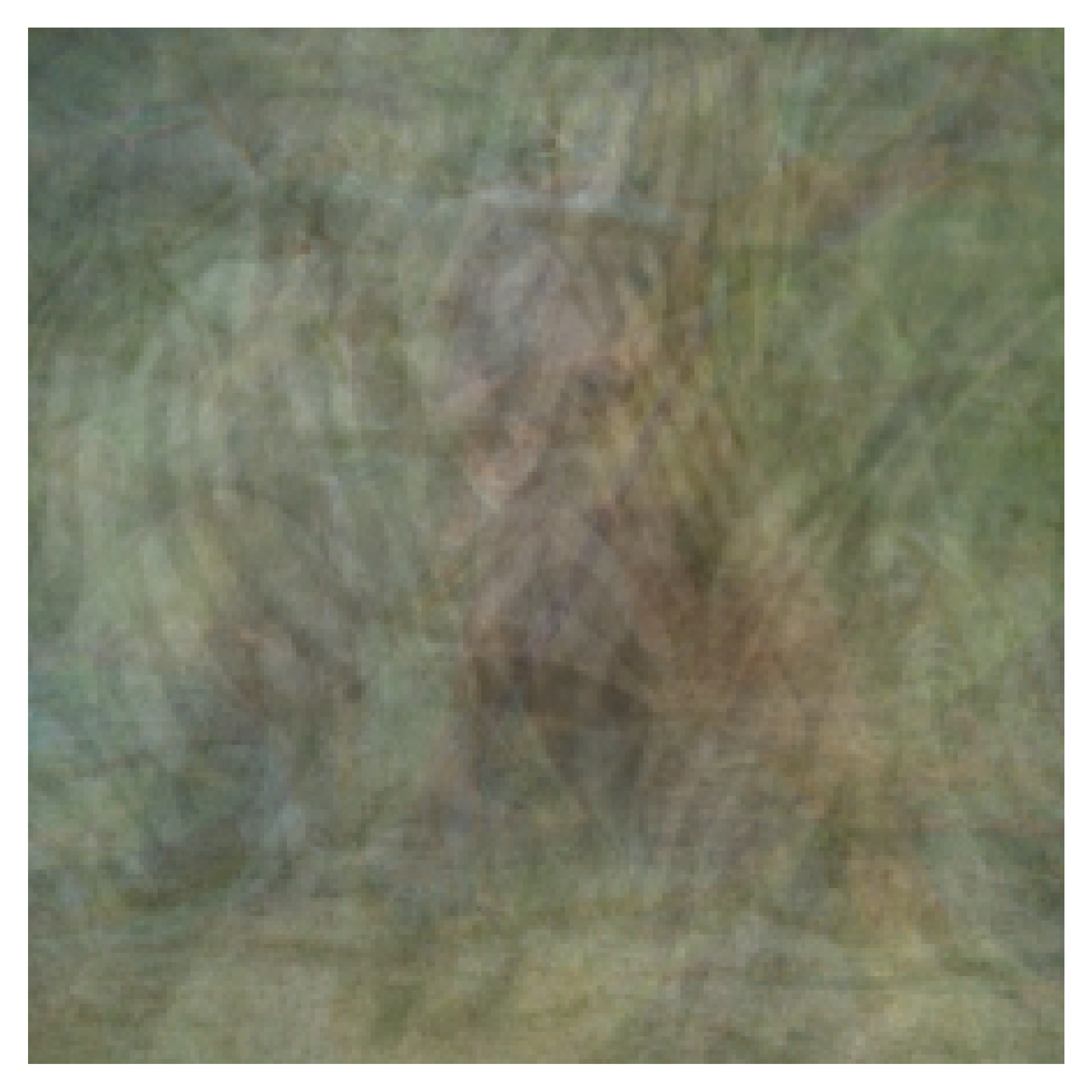}
}\hspace{-0.14in}
\subfigure[]{
    \centering
    \label{fig-introduction-raw-res}
    \includegraphics*[width=0.67in]{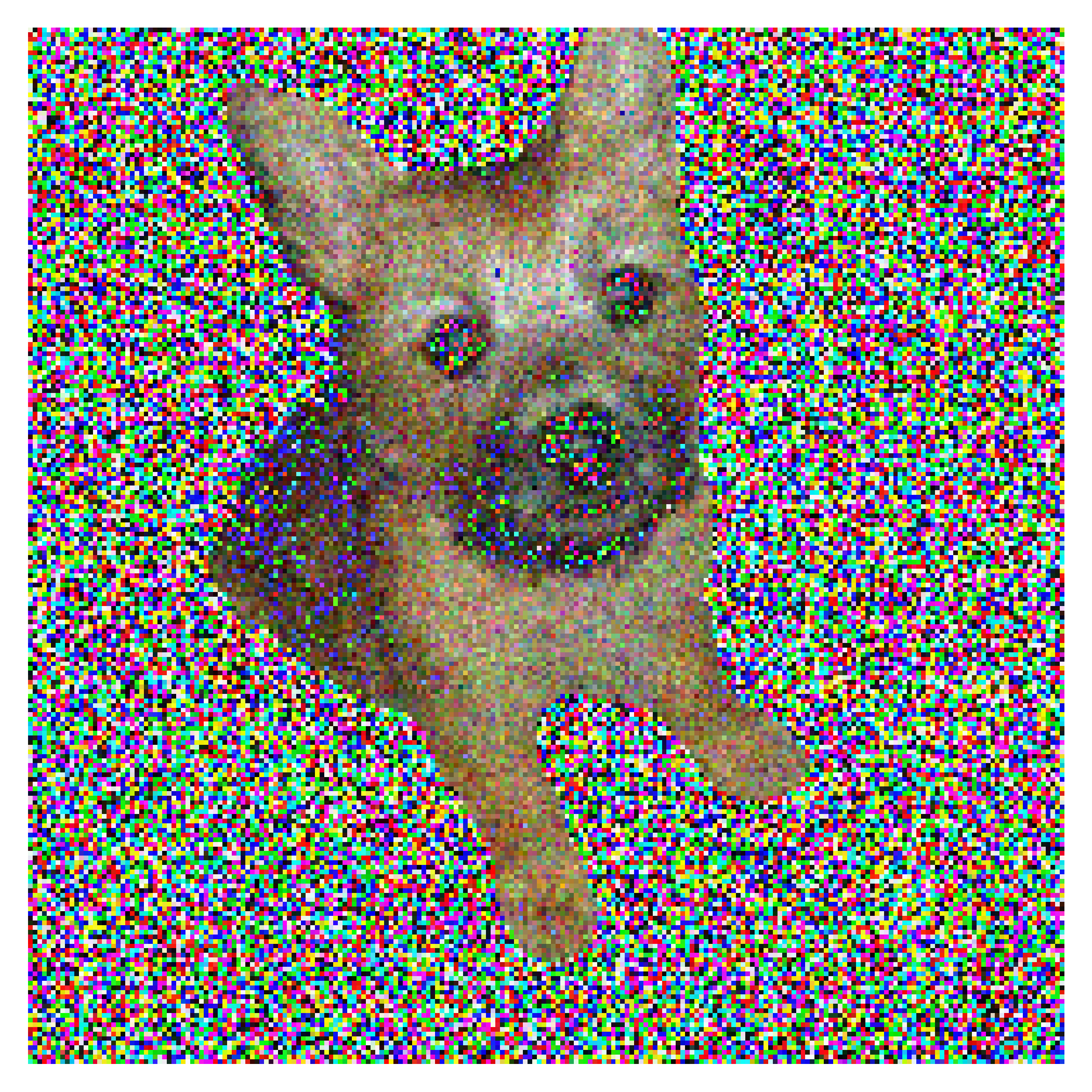}
}
\caption{Samples reconstructed by different attacks when the batch size is 16: (a) reconstructed by an existing attack in FL without LDP; (b) reconstructed by the same attack in FL with LDP; (c) reconstructed by the proposed attack in FL with LDP; (d) a linear combination of generated samples;  (e) generated samples with noise in the background.}
\end{figure}

Existing reconstruction attacks are lacking in feasibility and practicability. On the one hand, these attacks make too strong assumptions for better reconstruction attack performance (e.g., \cite{Pasquini2022secure, yin2020dreaming, Boenisch2023Reconstructing}). For example, a common assumption is that the adversary has multiple statistics of victims' training samples (e.g., \cite{yin2020dreaming, jeon2021gradient, Hatamizadeh2022GradViT}). However, victims have no motivation to calculate and upload the sample features, let alone how the adversary obtains the sample features from victims. On the other hand, existing works are difficult to extend to practical situations. Many attacks are only effective in simple cases with small batch sizes or training samples with low pixels (e.g., \cite{Zhu2020DLG, Yang2023compress, Wei2020Framework, Song2020Analyzing}). In contrast, the batch size and pixels are both larger in practice. Moreover, some attacks restrict the setting of the target model (e.g., \cite{Song2020Analyzing, jin2021cafe, zhu2021rgap}), damaging the target model's performance and making these attacks impracticable.

Besides, FL with local differential privacy (LDP) can protect users' privacy against reconstruction attacks (e.g., \cite{Kang2020FLwithDP, Zhou2022Differentially, Rui2020Personalized, Stevens2022Efficient}), in which gradients are clipped and perturbed. As shown in Fig.~\ref{fig-introduction-robbing-dp}, existing attacks hardly work in FL mechanisms with LDP. These attacks' ineffectiveness comes from two LDP measures: clipping and perturbing gradients. Some attacks modify target models so that the gradient contains a large amount of sample information, resulting in a \textit{gradient expansion}, i.e., the norm of gradients is too large. When the gradients are clipped and perturbed, the sample information in the gradients is seriously damaged, leading to the failure of the reconstruction attack. Existing attacks rely on accurate gradient values to reconstruct training samples. When a small amount of noise is added to these gradients, the reconstructed samples differ from the actual ones. We make a more detailed theoretical analysis of the failure of the existing attacks in Section~\ref{section-removing-redundate-gradients}.

This paper aims to determine \textbf{whether an adversary can feasibly and practically reconstruct the samples of victims in FL with LDP}. Table~\ref{table-prvacy-parameters} provides a comparison of the privacy parameter $\varepsilon$ settings in existing FL mechanisms with LDP, highlighting the deviations caused by noise to the loss function or model accuracy in the worst case, compared to scenarios without DP. The results indicate that when $\varepsilon=10$, the noise introduced by LDP significantly reduces the model accuracy.\footnote{Results come from experiments of the corresponding papers.} For smaller $\varepsilon$, the model performance drops so dramatically that the global model becomes unusable. However, our simulation results demonstrate that a malicious server can still effectively reconstruct the most private samples of victims in FL with LDP through the proposed attack (as shown in Fig.~\ref{fig-introduction-proposed-dp}) when $\varepsilon=10$.
\begin{table}
\caption{Setting of the privacy parameter $\varepsilon$ and deviations of different metrics in existing FL mechanisms with LDP.}
\centering
\label{table-prvacy-parameters}
\begin{tblr}{columns = {c, m, 0.7cm}, hlines={1pt}, hline{1,Z}={2pt}, column{1}={1.5cm}, vline{2,5}}
Reference & \SetCell[c=3]{c} Wei's work~\cite{Kang2020FLwithDP} & & & \SetCell[c=3]{c} Zhou's work~\cite{Zhou2022Differentially} & &\\
Metric & \SetCell[c=3]{c} Loss function & & & \SetCell[c=3]{c} Model accuracy & &\\
$\varepsilon$ & 6 & 8 & 10 & 1 & 5 & 10\\
Deviations & 120\% & 85\% & 57\% & 52\% & 27\% & 24\%\\
\end{tblr}
\end{table}

Theoretical analysis ensures the effectiveness of the proposed attack. Experimental results show that the proposed attack can effectively reconstruct sensitive information of samples from clipped and perturbed gradients protected by LDP. As shown in Fig.~\ref{fig-introduction-proposed-dp}, the reconstructed sample of the proposed attack against protected gradient exposes the primary information of the sample. In contrast, other attacks can only reconstruct meaningless noise. In addition, the proposed attack has almost no impact on the model training of non-target users, which ensures the performance of the target model. The key contributions of this paper are as follows:

\begin{itemize}
\item \textbf{Reconstruction attack against FL with LDP.} The proposed attack is the first reconstruction attack that reconstructs user samples in FL mechanisms with LDP, where user gradients are clipped and perturbed. Additionally, the attack targets only specific victims, having minimal impact on non-target users' training and the performance of global models.

\item \textbf{Attack feasibility and practicability.} The adversary does not interfere with the FL training process but targets standard FL mechanisms without requiring additional abilities or knowledge. Furthermore, the proposed attack is flexible, remaining effective across different target models, training samples, and large batch sizes. Consequently, this attack can be applied to any FL protocol and various learning scenarios while maintaining satisfactory concealment.

\item \textbf{Techniques against LDP.} We propose several techniques to counter noise introduced by LDP for sample reconstruction in FL with LDP. Firstly, we introduce a separation layer to prevent gradient expansion caused by reconstruction attacks, reducing the information loss of clipping gradients in LDP. This approach retains minimal gradient information necessary for sample reconstruction and employs an image segmentation model to extract the main subjects from samples, significantly reducing the gradients' norm. Secondly, we enhance the quality of the reconstructed samples by incorporating an imprinted structure that observes the noise distribution and scales the confidence interval to mitigate background noise. Additionally, we introduce a metric saver that imprints sample metrics onto the gradients, which are then used as an optimization objective to improve sample quality.

\item \textbf{Theoretical and experimental validity proof.} We guarantee the effectiveness of the proposed attack through theoretical analysis. Evaluation results demonstrate that the attack can effectively reconstruct users' privacy information of training samples in FL with LDP. Evaluations indicate that the proposed attack has minimal impact on FL training and target model accuracy. Additionally, we analyze several factors influencing the attack's performance, including privacy parameters and model complexity. Through evaluation results, we identify the critical conditions under which the attack is most effective in FL with LDP. Finally, we discuss the limitations of the proposed attack and suggest possible defenses.
\end{itemize}

\section{Related Work}

\begin{table*}
\centering
\caption{Comparison of existing sample reconstruction attacks in FL.}
\label{table-comparison-reference}
\begin{tblr}{rows={c,m}, hlines, hline{1,Z}={2pt}, column{1}={2.3cm}, column{3,4}={1.9cm}, column{5}={l}}
\SetCell[r=2]{c,m} Attack Type & \SetCell[r=2]{c,m} Mechanism & \SetCell[r=2]{c,m} Effectiveness in FL with LDP & \SetCell[r=2]{c,m} Sample Resolution & \SetCell[r=2]{m} Notes \\
& & & & \\
\SetCell[r=3]{c,m} Optimization-based Attack & DLG~\cite{Zhu2020DLG} & $\times$ & $ 64 \times 64$ &  The first effective reconstruction attack through gradients.\\
& Yin's work~\cite{Yin2021see}  & $\times$ & $ 224 \times 224$ & Considering multiple metrics for optimization.\\
& Pan's work~\cite{Pan2022Exploring} & $\times$ & $224 \times 224$ & Discussion about model complexity and attack effectiveness. \\
\SetCell[r=2]{c,m} Network-based Attack & mGAN-AI~\cite{Song2020Analyzing} & $\times$ & $ 64 \times 64$ & Samples must be identically distributed. \\
& GIAS~\cite{jeon2021gradient} & $\times$ & $ 64 \times 64$ & Applying the GAN to generate image prior. \\
\SetCell[r=2]{c,m} Analysis-based Attack & Fowl's work~\cite{fowl2022robbing} & $\times$ & $224 \times 224$ & Constructing bins to separate training samples.\\
&  Boenisch's work~\cite{Boenisch2021When} & $\times$ & $224 \times 224$ & Considering passive and active attacks in different scenarios.\\
Hybrid Attack & The Proposed Attack & $\surd$ & $224 \times 224$ & The sole effective reconstruction attack in FL with LDP.
\end{tblr}
\end{table*}

Sample reconstruction attacks are mainly divided into the following categories: optimization-based attacks (e.g., \cite{Zhu2020DLG, Yang2023compress, Hatamizadeh2022GradViT, Pan2022Exploring}), network-based attacks (e.g., \cite{Song2020Analyzing, jeon2021gradient, Khosravy2022Model, ganev2023inadequacy}), and analysis-based attacks (e.g., \cite{fowl2022robbing, Boenisch2021When, yuan2021beyond}). Table~\ref{table-comparison-reference} provides a brief comparison of these efforts.

\textbf{Optimization-based Attacks.}
In optimization-based attacks, the adversary regards reconstructed samples as multiple random variables and optimizes the reconstructed samples through gradients of objective functions. DLG~\cite{Zhu2020DLG} first proposed an optimization-based sample reconstruction attack, in which the objective function is the $\ell^2$ norm of the difference between victims' gradients and reconstructed samples-generated gradients. Other existing works proposed various improvements according to different requirements. For example, Yin et al.~\cite{Yin2021see} considered $\ell^2$ norm of the gradient difference, total variation of reconstructed samples, and sample statistics difference in the objective function to reconstruct samples with higher quality. Since the reconstructed samples are considered as multiple variables, the effectiveness of the optimization-based attacks depends heavily on the complexity of the sample (i.e., the number of random variables) and the complexity of the model (i.e., the number of constraints). When the batch size is large, finding the optimal solution (completely reconstructed samples) in optimization-based attacks is difficult.

\textbf{Network-based Attacks.}
The adversary in network-based reconstruction attacks trains a deep learning model to reconstruct victims' training samples. Most network-based attacks apply generative adversarial networks (GAN)~\cite{goodfellow2014GAN} to generate images similar to users' original samples. An attack mGAN-AI~\cite{Song2020Analyzing} modifies the global model with GAN in FL so that victims train the generator in the global model while training the classifier with local samples. The adversary reconstructs victims' training samples through the converged generator. GIAS~\cite{jeon2021gradient} applied GAN to transform the sample variables in the optimization problem into the input variables of the generator with lower dimensions, significantly reducing the search space and facilitating the finding of a better solution. Khosravy et al.~\cite{Khosravy2022Model} used a similar idea to successfully reconstruct victims' facial features in a face recognition system. A significant challenge for network-based attacks is to obtain enough samples with the same distribution as users' samples to train the generator.

\textbf{Analysis-based Attacks.} In analysis-based attacks, the adversary reconstructs training samples through the exact connection between gradients and training samples. Based on the theorem that gradients can calculate the input of FCL, Fowl et al.~\cite{fowl2022robbing} proposed to add a linear FCL to the global model and separate sample gradients so that most of the victims' samples can be reconstructed with high quality. Besides, based on the connection between gradients and training samples, Franziska et al.~\cite{Boenisch2021When} proposed a passive attack for reconstructing a single sample and an active attack for reconstructing multiple samples, where another idea for modifying the global model is given. The analysis-based attacks reconstruct high-quality training samples with low complexity and are still effective in the case of large batch sizes.

Some reconstruction attacks modify the global model structure to enable gradients to contain more sample information and improve the quality of reconstructed samples (e.g., \cite{fowl2022robbing, Boenisch2021When, Song2020Analyzing, Khosravy2022Model, ganev2023inadequacy}). Figure~\ref{fig:attack_type} compares sample reconstruction attacks using original and modified models. By embedding malicious structures in the global model, the gradients generated by modified models and client data carry additional sensitive information. Since the adversary reconstructs victims' samples through their gradient, attacks with modified models can achieve better attack effects than attacks with original models.

\begin{figure}
\centering
\includegraphics[width=0.95\linewidth]{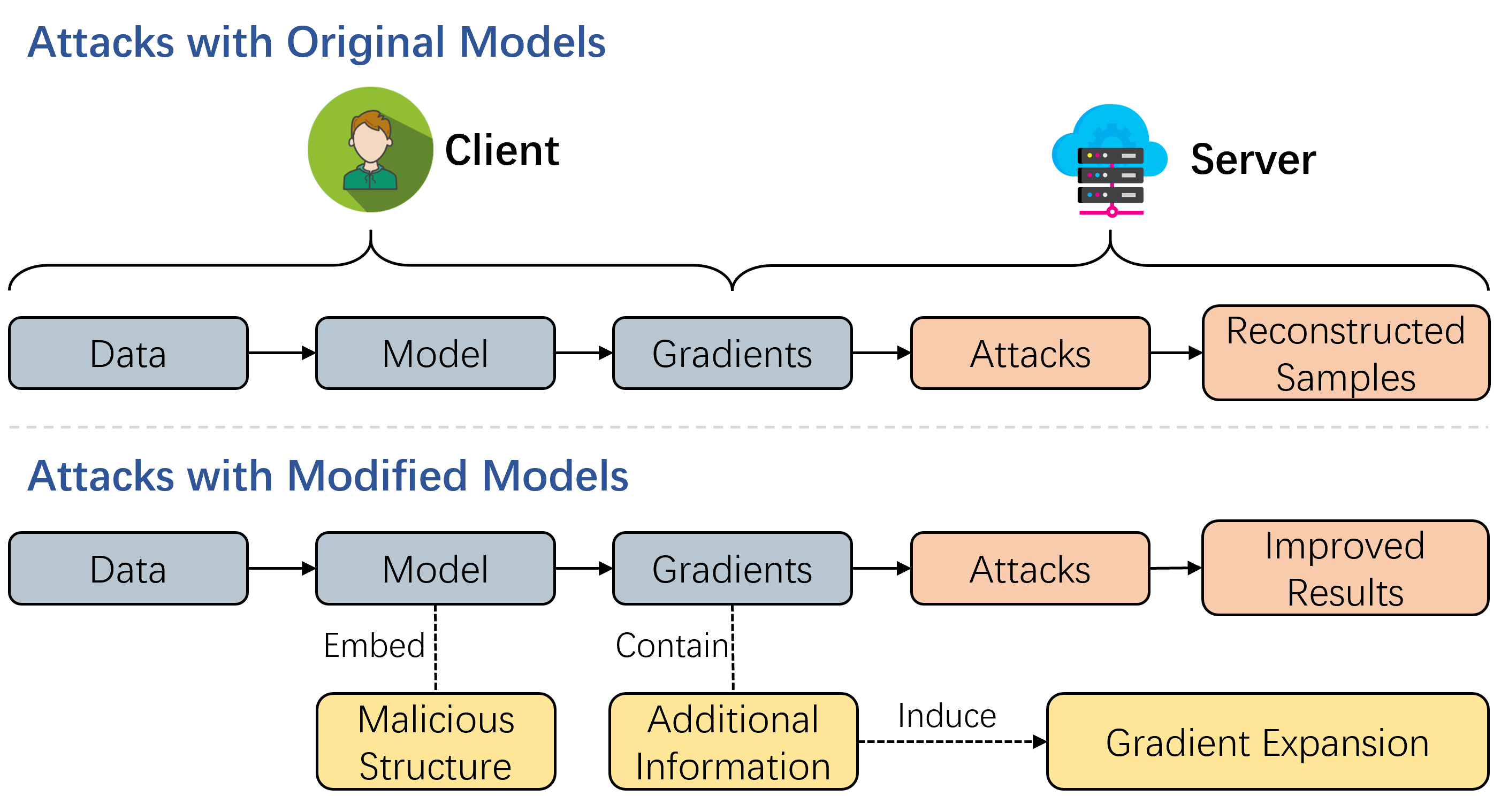}
\caption{Comparison of sample reconstruction attacks using original and modified models.}
\label{fig:attack_type}
\end{figure}

A disadvantage of attacks with modified models is causing gradient expansion, i.e., the norm of gradients increases since gradients carry additional information. Section~\ref{section-removing-redundate-gradients} provides a technical analysis for generating gradient expansion. Gradient expansion presents challenges when reconstructing samples in FL with LDP. Since gradients are clipped according to a fixed norm value and perturbed with noise, gradients with additional information are severely compressed. As a result, the signal-to-noise ratio in the protected gradient drops significantly. No existing works can effectively reconstruct samples from clipped and perturbed gradients. The proposed attack against FL with LDP modifies the global model to obtain effective performance while reducing the impact of gradient expansion by gradient compression presented in Sections~\ref{section-removing-redundate-gradients} and \ref{section-removing-background}.

In addition, existing attacks rely on accurate gradients to reconstruct samples. In FL with LDP, the noise in the perturbed gradient causes existing attacks to reconstruct meaningless noise samples. The proposed attack reduces the impact of perturbed gradients on the attack effect by filtering the noise of gradients and reconstructed samples presented in Sections~\ref{section-noise-filtering} and \ref{section-mectics-optimization}. Dealing with gradient clipping and perturbation in LDP by gradient compression and noise filtering, the proposed attack effectively reconstructs sensitive information in the sample from the protected gradients in FL with LDP.

\section{Preliminaries}

\subsection{Federated Learning}

FL is a distributed learning framework to solve the privacy concern that servers (i.e., model owners) train their models with users' sensitive and private data. In FL, servers design machine Learning target models based on predicted task requirements. Then, they distribute the global model (i.e., a target model) to users who train the global model with their local data to produce intermediate training results (gradients or new parameters). Finally, the server aggregates the intermediate training results and obtains new global model parameters. The above procedures are repeated until global models converge.

We illustrate FL framework details with a typical FL algorithm FedSGD~\cite{McMahan2017FedSGD}, the FL aggregation algorithm considered in this paper. Specifically, a server customizes the global ML model structure $f$, initialized model parameters $\omega$, hyper-parameters, and training process. The server distributes the training materials to users whose local data is $x$ and $y$. User $i$ trains the model with a random batch of its local data $x_i$ and $y_i$ and generates model gradients $\nabla_\omega L(f(x_i;\omega), y_i) = \partial L(f(x_i;\omega), y_i) / \partial \omega$, where $L$ is a loss function. The server aggregates users' uploaded gradients and updates the global model parameters distributed to users for the next round of training. The training process terminates until the global model converges or meets the termination condition.

\subsection{Federated Learning with Local Differential Privacy}
\label{section-preliminary-fl-ldp}

Existing works apply LDP~\cite{dwork2014algorithmic} to protect users' intermediate results during the FL training process (e.g., \cite{Kang2020FLwithDP, Zhou2022Differentially, Rui2020Personalized, Stevens2022Efficient, mcmahan2018learning}). LDP enables users to perturb gradients with noise before uploading gradients to the server.
\begin{definition}[Differential privacy~\cite{dwork2014algorithmic}]
A randomized algorithm $\mathcal{M}$ with domain $\mathbb{N}^{\left| \mathcal{X} \right| }$ is $\left(\varepsilon, \delta \right)$-differentially private if for all $S \subseteq \mathrm{Range}\left(\mathcal{M}\right) $ and for all $x,y \in \mathbb{N}^{\left| \mathcal{X} \right| }$ such that $\| x - y\|_1 \leq 1$:
\begin{equation}
\mathrm{Pr}\left[ \mathcal{M}\left(x\right) \in S \right] \leq \exp\left(\varepsilon\right) \mathrm{Pr}\left[\mathcal{M}\left(y\right) \in S \right] + \delta.
\end{equation}
\end{definition}
Adjusting privacy parameters $\left(\varepsilon, \delta \right)$ can meet various user privacy and data accuracy requirements. In most cases, better user privacy protection would reduce data accuracy.

As in the existing works~\cite{Kang2020FLwithDP, Zhou2022Differentially, Naseri2022local, Rui2020Personalized}, we consider that users add the Gaussian noise $\mathcal{N}\left(0, \sigma^2 \right)$ to gradients for privacy preservation. Given privacy parameters $\left(\varepsilon, \delta \right)$, the Gaussian mechanism is $\left(\varepsilon, \delta \right)$-differentially private when setting a proper scale $\sigma$~\cite{dwork2014algorithmic}. Algorithm~\ref{alg-noising-gradients} presents a general local learning process for users in FL mechanisms with LDP. Users train the global model with a batch of local samples and generate local gradients. Then, generated gradients are clipped according to the clipping bound to avoid the norm of generated gradients being too large so that the perturbation is too small to protect user privacy. Users add the Gaussian noise to clipped gradients with the scale factor decided by themselves according to different privacy requirements and upload the perturbed gradients to the server.

\begin{algorithm}
\caption{Local training with LDP}
\label{alg-noising-gradients}
\KwIn{Global model $f$, model parameters $\omega$, loss function $L$, a batch of samples $\left\{x, y\right\}$, clipping bound $C$, and scale factor $\sigma$}
\KwOut{Clipped and noisy local gradient $\nabla_{\omega} L$}
Generate local gradients $\nabla_{\omega} L$ = $\frac{\partial L\left( f\left(x;\omega \right), y \right)}{\partial \omega}$; \\
Clip gradients $\nabla_{\omega} L = \nabla_\omega L / \max \left(1 , \frac{\| \nabla_\omega L \|}{C} \right)$;\\
Perturb gradients $\nabla_{\omega} L = \nabla_\omega + \mathcal{N}\left(0, \sigma^2\right)$;\\
\Return{$\nabla_{\omega} L$}
\end{algorithm}

\subsection{Sample Reconstruction Attacks in Federated Learning}

\begin{figure}
\centering
\includegraphics[width=3.4in]{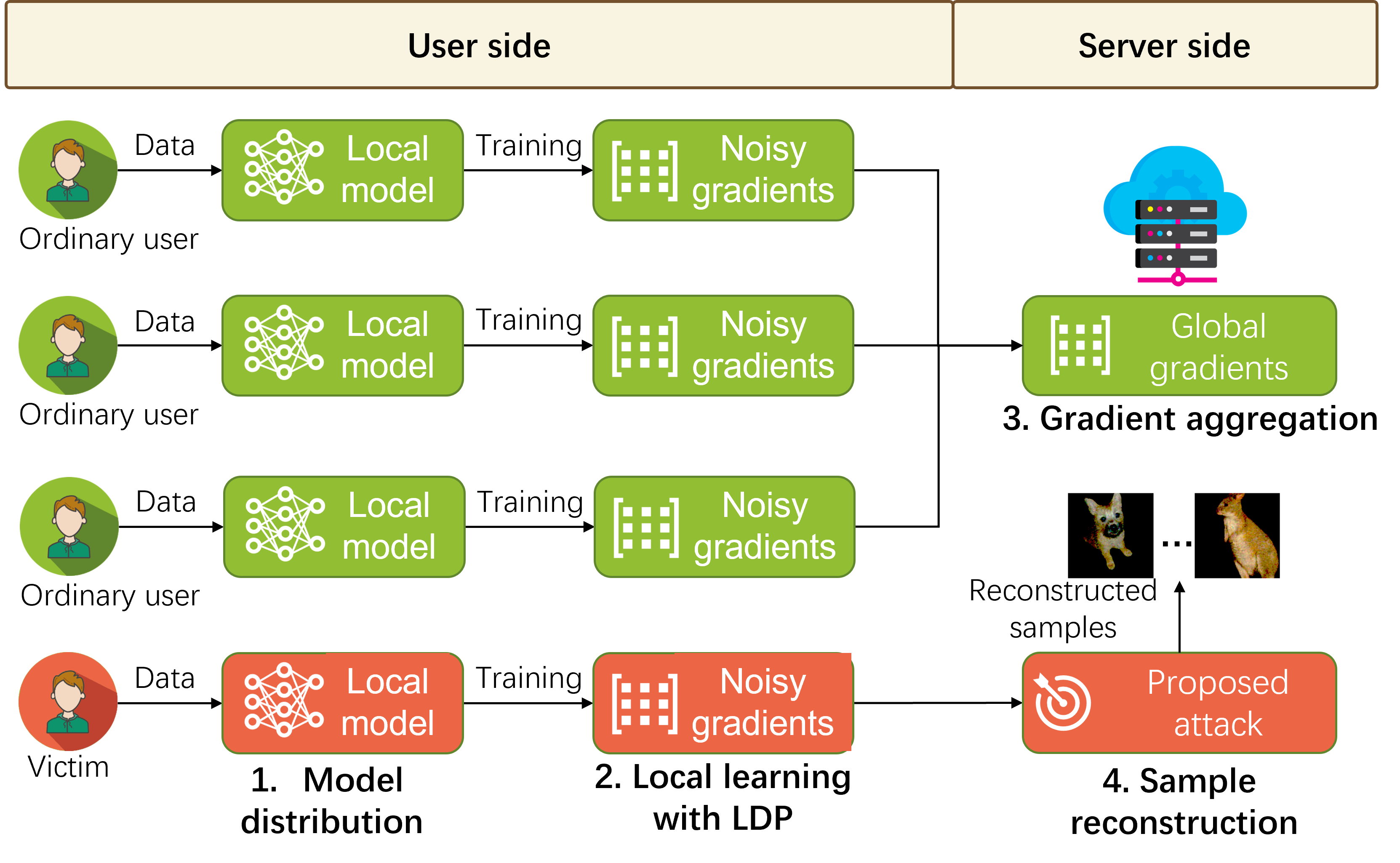}
\caption{The process of the proposed attacks in FL with LDP.}
\label{fig-reconstruction-attack-process}
\end{figure}

Figure~\ref{fig-reconstruction-attack-process} introduces the process by which the server implements the proposed reconstruction attacks and steals user samples through clipped and perturbed gradients uploaded by users with LDP. Firstly, the server distributes the global model structure and parameters to users as local models for local training in model distribution, which is a standard setting of FL. Then, users train local models with their data and generate local gradients. In FL with LDP, users clip and add noise to the generated local gradients according to custom privacy parameters before uploading the gradients to the server for gradient protection. Therefore, the server only gets clipped and perturbed gradients.

For ordinary users, i.e., non-target users, the server aggregates their gradients to update the global model to ensure model performance. For the victim, i.e., the target user, the server reconstructs its samples according to its clipped and perturbed gradients through the proposed attack (given in Section~\ref{section-attack} and Section~\ref{section-all-in-one}). It should be noted that existing reconstruction attacks cannot reconstruct the victim's samples through clipped and perturbed gradients. Experimental results in Section~\ref{section-evaluation} show that even though protecting gradients by LDP, the proposed attack is still effective in reconstructing the victim's samples.

Users upload local gradients instead of original samples to avoid privacy leaks in FL. However, sample reconstruction attacks can reconstruct users' training samples through their uploaded gradients, exposing privacy risks in FL. Meanwhile, the proposed attack shows that FL with LDP cannot entirely defend sample reconstruction attacks.

\section{Threat Model and Primary Attack}

\subsection{Threat Model}

\textbf{Attack motivations.}
The privacy protection FL provides users is that users can perform model training by leaving sensitive samples locally. However, training samples contain precious information, such as medical, financial, or personal information. When model training and performance are not affected, the server may be interested in performing an attack to secretly reconstruct training data to obtain more benefits.

\textbf{Adversary's goal.}
The adversary performs a reconstruction attack to reconstruct users' training samples from gradients and extract the sensitive information in the samples. For example, as shown in Fig.~\ref{fig-introduction-proposed-dp}, we assume that the adversary is more concerned with the subject information in the image, and its background information can be ignored. However, the adversary cannot require the user to pre-process local samples before training to reduce the difficulty of the attack. Meanwhile, whether a reconstruction attack is performed should not be distinguished from the performance of the global model on target prediction tasks for better attack concealment.

\textbf{Adversary's ability.}
The malicious server's only ability is to design a global model for training, which is a common situation in FL (e.g., \cite{Kairouz2019open, McMahan2017FedSGD, Stevens2022Efficient}). Besides, designing a model is also the basic assumption of most existing reconstruction attacks (e.g., \cite{pasquini2022eluding, Boenisch2021When, fowl2022robbing}). We do not limit target models and prediction tasks for better attack practicability, i.e., the designed global model should completely contain any given target models. This paper aims to show that LDP-based FL without global model verification cannot protect users' privacy.

\textbf{Adversary's knowledge.} 
The adversary has the complete structure and parameters since they are in charge of global design. In addition, the adversary receives victims' clipped and noisy gradients with LDP protection, which is discussed as follows. Meanwhile, the server only obtains victims' gradients once, i.e., the attack should be effective with one round of gradients, and the adversary cannot require victims to upload multiple rounds of gradients to reduce attack difficulty.

\textbf{Gradient protection.} We considers gradient protection with LDP based on the setting in Section~\ref{section-preliminary-fl-ldp}, which contains clipping and perturbation of the gradient as shown in Algorithm~\ref{alg-noising-gradients}. Users determine the privacy parameters and clipping bounds in LDP, and the above values are not exposed to the server.

\subsection{Primary Attack}
\label{section-primacy-attack}

The primary attack considers a simple case on a fully connected layer (FCL) where the batch size is one. More complicated cases with larger batch sizes and complex models are discussed in Section~\ref{section-attack}.

\begin{lemma}
The adversary can reconstruct the input of any FCL through its gradients when the batch size is one~\cite{Geiping2020}.
\label{lemma-attack-1-one-sample}
\end{lemma}

\begin{proof}
Suppose the parameters of an FCL are $[w, b]^\intercal$, for any single sample $x$, the output of the FCL is $y = w^\intercal x + b$.
According to the chain rule, we have
\begin{equation}
\nabla_{b} L = \frac{\partial L }{\partial b} = \frac{\partial L }{\partial y} \cdot \frac{\partial y}{\partial b} = \frac{\partial L}{ \partial y},
\end{equation}
and
\begin{equation}
\nabla_{w} L = \frac{\partial L }{\partial w} = \frac{\partial L }{\partial y} \cdot \frac{\partial y}{\partial w} = \frac{\partial L}{ \partial y} \cdot x = \nabla_{b} L \cdot x.
\end{equation}
As a result, the adversary can reconstruct the sample $x$ by $ x = \nabla_{w} L \oslash \nabla_{b} L$ where $\oslash$ is the entry-wise division.
\end{proof}

Lemma~\ref{lemma-attack-1-one-sample} provides a straightforward way of reconstructing samples from gradients: embedding an FCL to the model so that the first layer is an FCL and reconstructing samples by Lemma~\ref{lemma-attack-1-one-sample}, which we refer to as the \textit{primary reconstruction attack}. However, the batch size is hardly set to one in a real scenario, and the primary reconstruction attack is limited when the batch size is larger due to the following theorem.

\begin{theorem}
The output of the primary attack is a linear combination of training samples.
\label{theorem-attack-2-multiple-samples}
\end{theorem}

\begin{proof}
Suppose that the batch size is $B$, and training samples are $\left\{ x^{\left(1\right)}, x^{\left(2\right)}, \cdots, x^{\left(B\right)} \right\}$, according to the back-propagation, we have 
\begin{equation}
\nabla_{b} L = \frac{1}{B} \sum_{i=1}^{B} \frac{\partial L }{\partial b} = \frac{1}{B} \sum_{i=1}^{B} \frac{\partial L }{\partial y} \cdot \frac{\partial y}{\partial b} =  \frac{\partial L}{ \partial y},
\end{equation}
and
\begin{equation}
\nabla_{w} L = \frac{1}{B} \sum_{i=1}^{B} \frac{\partial L }{\partial w} = \frac{1}{B} \sum_{i=1}^{B} \frac{\partial L}{ \partial y} \cdot x^{\left( i \right)} = \frac{\nabla_{b} L}{B} \sum_{i=1}^{B} x^{\left( i \right)}.
\end{equation}
As a result,
\begin{equation}
\nabla_{w} L \oslash \nabla_{b} L =  \frac{1}{B} \sum_{i=1}^{B} x^{(i)},
\end{equation}
i.e., a linear combination of all training samples.
\end{proof}

Theorem~\ref{theorem-attack-2-multiple-samples} shows that when the batch size exceeds $1$, the primary attack only gets a linear combination of training samples. When the batch size is large, such a linear combination cannot expose too much information about the training samples. For example, Fig.~\ref{fig-attack-example-multiples-sum} shows the training samples reconstructed by the primary attack when the batch size is 2, 4, 8, and 16. As the batch size increases, it is harder to distinguish sample information from the reconstructed samples.

\begin{figure}
\centering
\includegraphics*[width=0.79in]{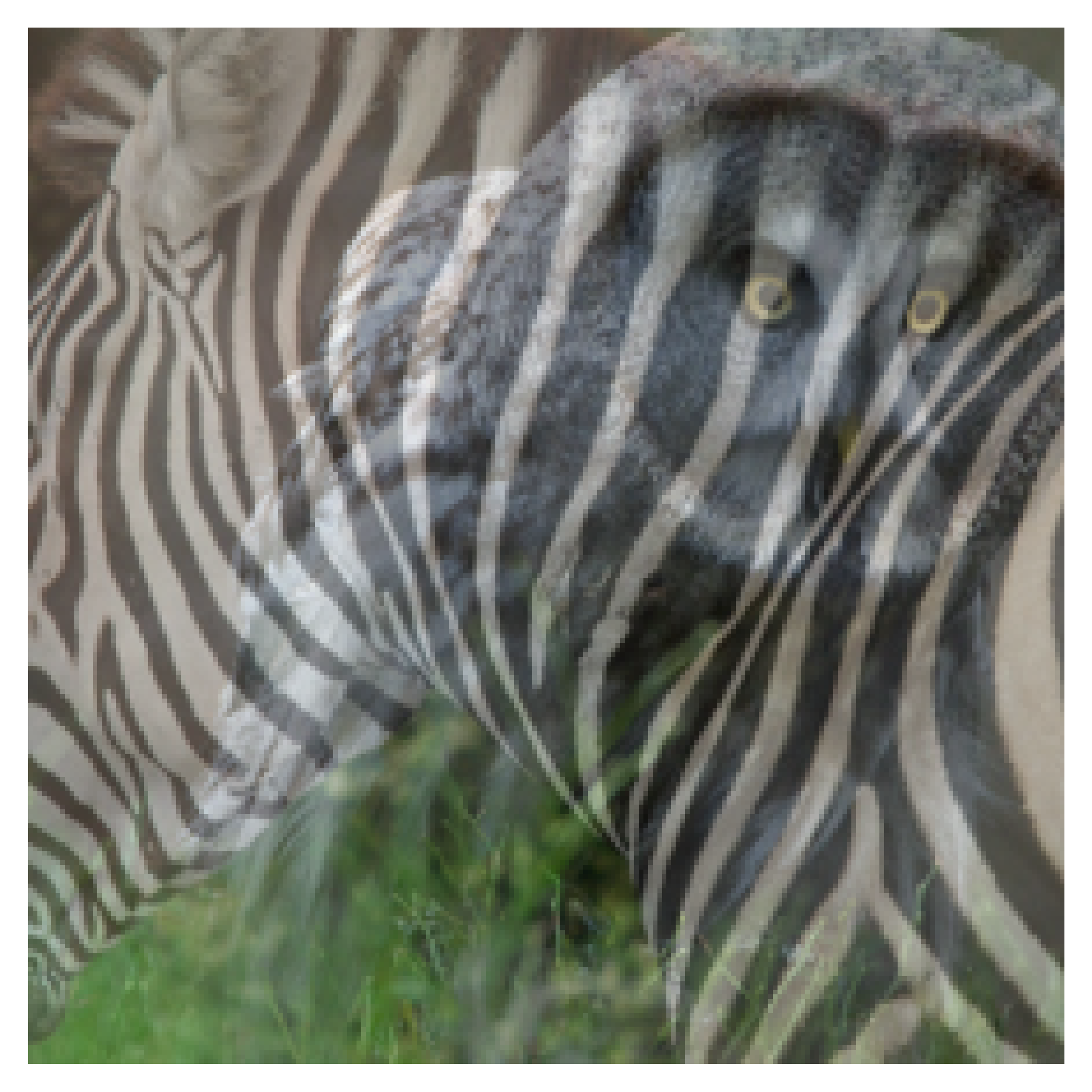}
\includegraphics*[width=0.79in]{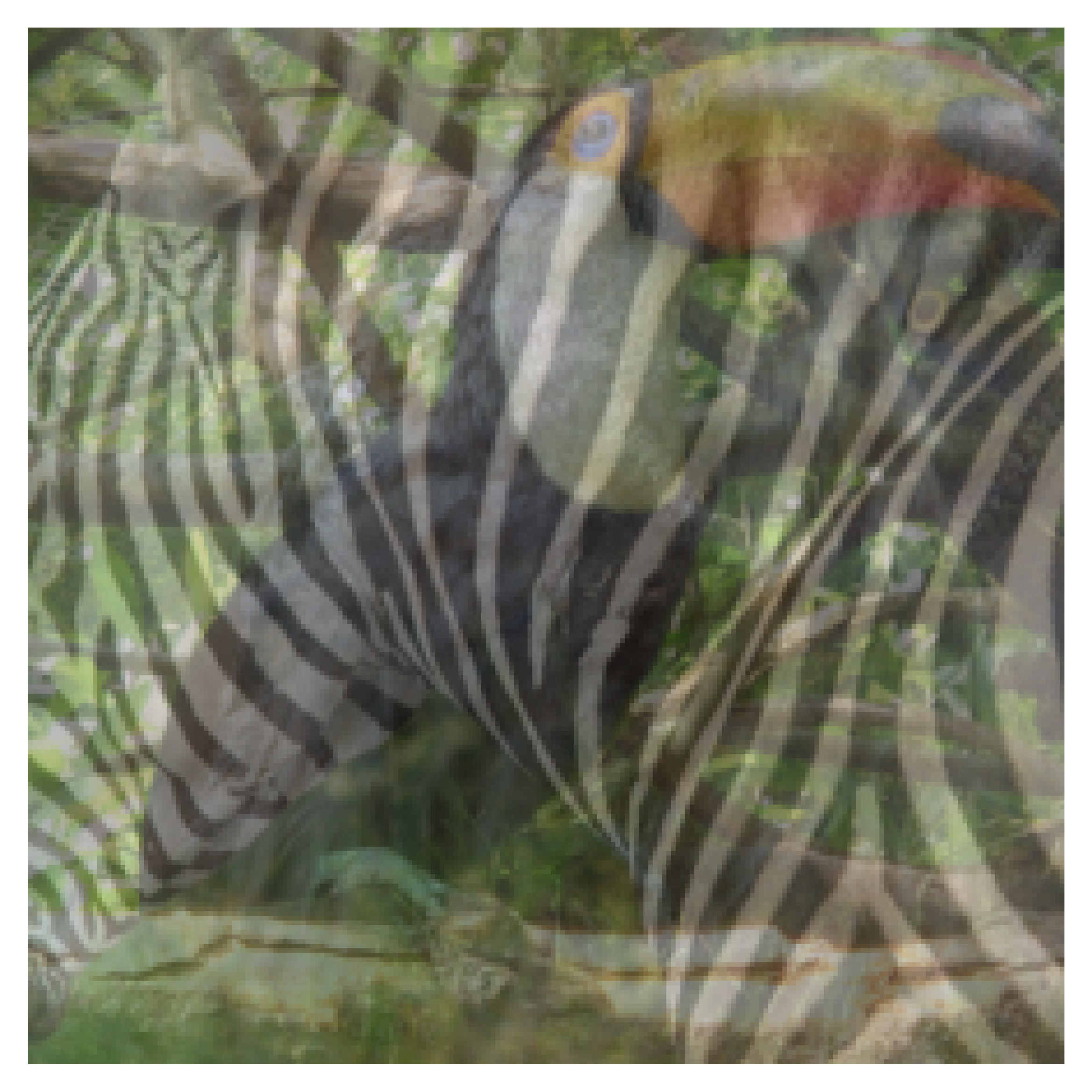}
\includegraphics*[width=0.79in]{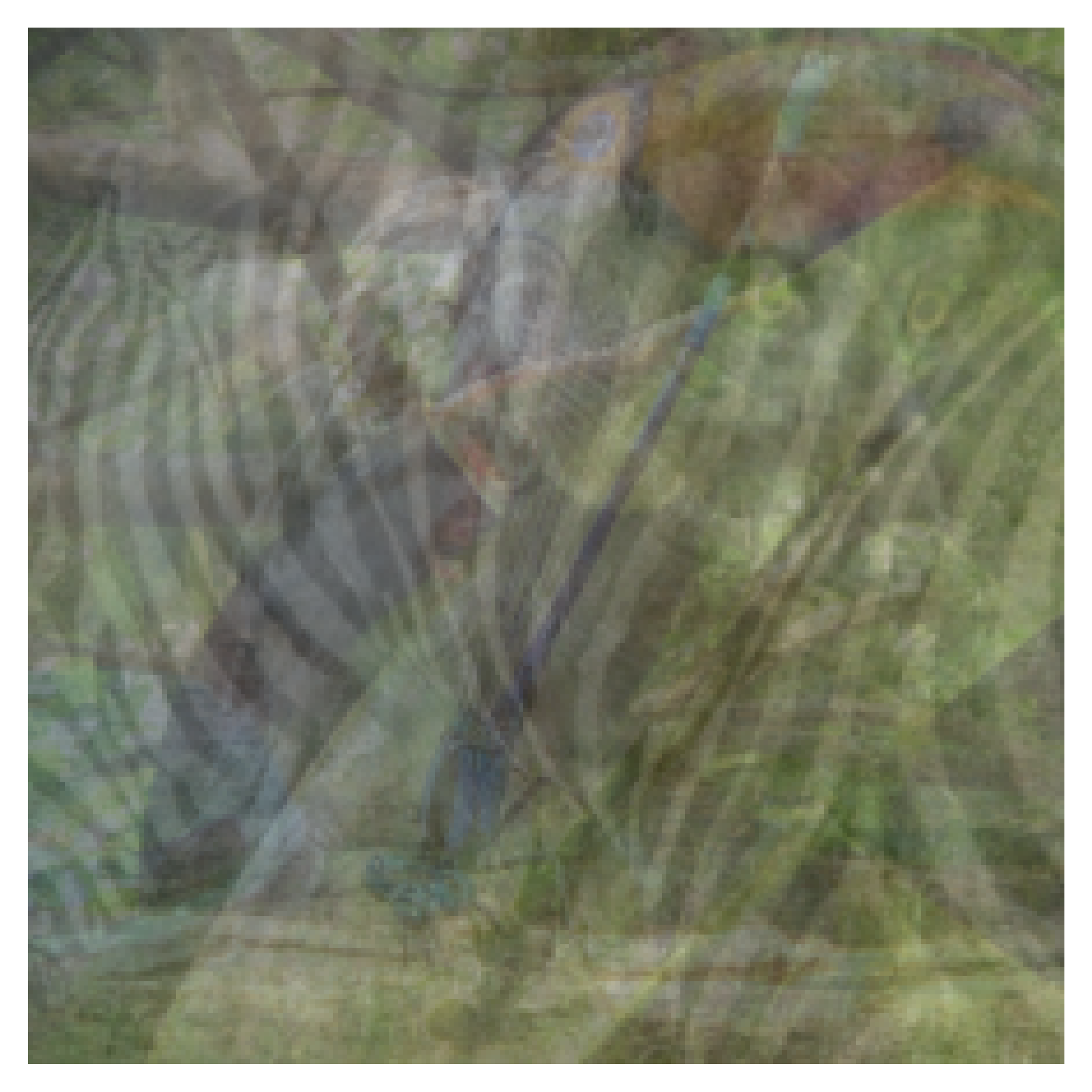}
\includegraphics*[width=0.79in]{sum_of_16_images.png}
\caption{Reconstructed training samples generated by the primary attack when the batch size is 2, 4, 8, and 16.}
\label{fig-attack-example-multiples-sum}
\end{figure}

A straightforward solution is to let each neural unit in the FCL only contain one sample's gradients, and the adversary can reconstruct separated training samples through the primary attack. Specifically, suppose that there are $K$ units in the FCL, and unit $k$ only contains the gradients generated by sample $x^{\left(i\right)}$, we have
\begin{equation}
\nabla_{b_k} L = \frac{1}{B} \cdot \frac{\partial L }{\partial b_k} = \frac{1}{B} \cdot \frac{\partial L }{\partial y} \cdot \frac{\partial y}{\partial b_k} = \frac{1}{B} \cdot \frac{\partial L}{ \partial y_k},
\end{equation}
and 
\begin{equation}
\nabla_{w_k} L = \frac{1}{B} \cdot \frac{\partial L }{\partial w_k} = \frac{1}{B} \cdot \frac{\partial L}{\partial y_k} \cdot x^{\left( i \right)} = \nabla_{b_k} L \cdot x^{\left( i \right)}.
\end{equation}
When unit $k$ only contains the gradients generated by $x^{\left(i\right)}$, the adversary can reconstruct $x^{\left( i \right)}$ by $\nabla_{w_k} L \oslash \nabla_{b_k} L$ as discussed in Lemma~\ref{lemma-attack-1-one-sample}, where $\oslash$ is the entry-wise division. For convenience, we refer to $\nabla_w L$ and $\nabla_b L$ as weight and bias gradients, respectively. According to the primary attack, the adversary can reconstruct any sample with its weight and bias gradients. We refer to the above process of separating gradients of each sample into different units of FCL as \textit{gradient separation}.

Since gradients are clipped and perturbed, the primary attack cannot obtain effective reconstructed results in FL with LDP. The primary attack cannot distinguish the samples’ corresponding gradients (weight and bias gradients) and noise from FCL gradients. It further leads to the failure of gradient separation. Besides, $x^{\left( i \right)}$ is reconstructed by perturbed $\nabla_w L$ and $\nabla_b L$, which leads to the primary attack getting some meaningless noise samples since gradients are noisy.

\section{Reconstruction Attack against FL with LDP}
\label{section-attack}
\subsection{Gradient Separation without Expansion}
\label{section-removing-redundate-gradients}

We first briefly analyze the reason for gradient expansion. As discussed in Section~\ref{section-primacy-attack}, the primary attack reconstructs sample $x^{(i)}$ by $\nabla_{w_k} L \oslash \nabla_{b_k} L$. According to Theorem~\ref{theorem-attack-2-multiple-samples}, the key of the primary attack is to make each neural unit in the FCL only contain one sample's gradient, i.e., gradient separation. Therefore, a larger number of units in the FCL brings better effectiveness of gradient separation. For example, the existing attack~\cite{fowl2022robbing} requires about 1024 units to separate gradients of samples in ImageNet dataset~\cite{ILSVRC15} when the batch size is 16. Since the sample size is $16 \times 3 \times 224 \times 224$, the existing attacks introduce additional $1024 \times 16 \times 3 \times 224 \times 224$ values into the gradients, resulting in a large norm of gradients and causing gradient expansion.

We propose a gradient separation method against FL with LDP, which has the following advantages. Increasing the number of units in FCLs can improve the separation effect of the proposed method but would not increase the norm of gradients, effectively avoiding the gradient expansion. The following are the technical implementation details.

Suppose an FCL contains $K$ units with parameters  $\left[w,b\right]^\intercal$, the weights of all units in the FCL are set to equal, i.e., $w_1 = w_2 = \cdots = w_K$ (its value is a hyper-parameter given in Section~\ref{section-evaluation}). The bias parameters in the FCL are set according to the quantile function of a random variable following the Laplace distribution. In other words, given $X \sim \textrm{Laplace}(\mu,s)$, $b_j = - F^{-1}_X( j / K)$ for bias of unit $j$, where $s > 0$ is a scale factor and $F_X(\cdot)$ is the cumulative distribution function (CDF) of $X$. The output of the FCL $y_{\textrm{min}}$ is the minimum positive value of $\left( w^\intercal x + b\right)$. We refer to an FCL with the above setting as a \textit{separation layer} and introduce its variations to achieve the proposed attack in Section~\ref{section-all-in-one}.

The separation layer achieves the following property for gradient separation without gradient expansion. A \textit{reverse index} $i_0 \in \left\{1, 2, \dots, K\right\}$ of any sample $x^{\left(i\right)}$ is an index of the unit in the separation layer such that $w_{i_0}^\intercal x^{\left(i\right)} + b_{i_0} > 0 $ and $w_{i_0+1}^\intercal x^{\left(i\right)} + b_{i_0+1} \leq 0 $. The $i_0$-th unit in the separation layer is a \textit{reverse unit} of sample $x^{\left(i\right)}$.

\begin{theorem}
Each sample's corresponding gradients in the separation layer only exist in its reverse units.
\label{theorem-one-non-zero-gradients}
\end{theorem}

\begin{proof}
Bias of unit $j$ in the FCL is set to $b_j = - F^{-1}_X( j / K)$ where $X \sim \textrm{Laplace}(0,s)$ and $F_X(\cdot)$ is the CDF of $X$, and we have $b_1 > b_2 > \cdots > b_K$.
Since $w_1 = w_2 = \cdots = w_K$, for any sample $x^{\left(i\right)}$,
\begin{equation}
w_1^\intercal x^{\left(i\right)} + b_1 > w_2^\intercal x^{\left(i\right)} + b_2 > \cdots w_K^\intercal x^{\left(i\right)} + b_K.
\end{equation}
Combining the reverse index $i_0$, we have 
\begin{equation}
\left\{ \begin{matrix}
    w_1^\intercal x^{\left(i\right)} + b_1 > w_2^\intercal x^{\left(i\right)} + b_2 > \cdots > w_{i_0}^\intercal x^{\left(i\right)} + b_{i_0} > 0; \\ \\
    w_K^\intercal x^{\left(i\right)} + b_K < \cdots < w_{i_0+1}^\intercal x^{\left(i\right)} + b_{i_0+1} \leq 0.
\end{matrix}\right.
\end{equation}
Suppose $y^{\left(i\right)}_{\mathrm{min}}$ is the minimal non-zero positive value of $w^\intercal x^{\left(i\right)} + b$, $y^{\left(i\right)}_{\mathrm{min}} = w_{i_0}^\intercal x^{\left(i\right)} + b_{i_0}$. In other words, for any sample $x^{\left( i \right)}$, $y^{\left(i\right)}_{\mathrm{min}}$ depends on the $i_0$-th unit in the FCL, i.e., $w_{i_0}^\intercal x + b_{i_0}$. When only considering sample $x^{\left(i\right)}$, we have
\begin{equation}
\left\{ \begin{matrix}
    \nabla_{w_{i_0}}L = \frac{\partial L}{\partial w_{i_0}} = \frac{\partial L}{\partial y^{\left(i\right)}_{\mathrm{min}}} \cdot \frac{\partial y^{\left(i\right)}_{\mathrm{min}}}{\partial w_{i_0}} = \frac{\partial L}{\partial y^{\left(i\right)}_{\mathrm{min}}} \cdot x^{\left(i\right)}; \\
    \; \\
    \nabla_{b_{i_0}}L = \frac{\partial L}{\partial b_{i_0}} = \frac{\partial L}{\partial y^{\left(i\right)}_{\mathrm{min}}} \cdot \frac{\partial y^{\left(i\right)}_{\mathrm{min}}}{\partial b_{i_0}} = \frac{\partial L}{\partial y^{\left(i\right)}_{\mathrm{min}}}.
\end{matrix}\right.
\end{equation}
On the other hand, for any unit $k^\prime \neq i_0$ in the FCL,
\begin{equation}
\left\{ \begin{matrix}
    \nabla_{w_{k^\prime}}L = \frac{\partial L}{\partial w_{k^\prime}} = \frac{\partial L}{\partial y^{\left(i\right)}_{\mathrm{min}}} \cdot \frac{\partial y^{\left(i\right)}_{\mathrm{min}}}{\partial w_{k^\prime}} = 0; \\
    \; \\
    \nabla_{b_{k^\prime}}L = \frac{\partial L}{\partial b_{k^\prime}} = \frac{\partial L}{\partial y^{\left(i\right)}_{\mathrm{min}}} \cdot \frac{\partial y^{\left(i\right)}_{\mathrm{min}}}{\partial b_{k^\prime}} = 0.
\end{matrix}\right.
\end{equation}
The above equations lead to Theorem~\ref{theorem-one-non-zero-gradients}.
\end{proof}

More intuitively, Theorem~\ref{theorem-one-non-zero-gradients} proves that sample $x^{\left(i\right)}$ only generates gradients at the $i_0$-th unit in the separation layer, which achieves gradient separation. The adversary can reconstruct samples from gradients of reverse units in the separation layer when reverse units only contain one sample's corresponding gradients. Specifically, for any sample $x^{\left(i\right)}$, suppose that unit $i_0$ only contains the gradients of sample $x^{\left(i\right)}$, the adversary can reconstruct $x^{\left(i\right)}$ by $\nabla_{w_{i_0}}L \oslash \nabla_{b_{i_0}}L$.

Meanwhile, according to the following theorem, the number of units with non-zero gradients in the separation layer is not more than the batch size, preventing gradient expansion. Taking an example on ImageNet dataset~\cite{ILSVRC15} where the batch size is 16 and samples are images with $3 \times 224 \times 224$ pixels, the number of non-zero gradients in the separation layer is not greater than $16 \times 3 \times 224 \times 224$. However, the number of non-zero gradients in the existing method~\cite{fowl2022robbing} is $1024 \times 3 \times 224 \times 224$.

\begin{theorem}
Even multiple increases in the number of units in the separation layer reduce the probability that each unit contains multiple sample gradients, and the number of units with non-zero gradients is not greater than the batch size in the separation layer.
\label{theorem-two-performance-related}
\end{theorem}

%The intuition of the proof of Theorem~\ref{theorem-two-performance-related} is that reverse indexes of samples must fall into one unit of the separation layer. More units reduce the probability of reverse indexes falling into the same unit. 
\begin{proof}
We first analyze the probability that the reverse index of any sample $x^{\left(i\right)}$ is $k$, i.e., $\mathrm{Pr}\left\{i_0 = k\right\}$. Recall that the bias parameters of the FCL in the inference structure are set by the quantile function of a variable following the Laplace distribution, we have $b_j = - F^{-1}_X( j / K)$ where $X \sim \textrm{Laplace}(0,s)$, $F_X(\cdot)$ is the CDF of $X$, and $n=K$ is the number of units in the FCL. Specifically, according to the CDF of Laplace distribution, 
\begin{equation}
    F_X^{-1}\left( \frac{j}{K} \right) = - s \cdot \mathrm{sgn}\left( \frac{j}{K} - 0.5 \right) \ln \left(1 - 2 \left| \frac{j}{K} -0.5 \right| \right).
\end{equation}
For any sample $x^{\left(i\right)}$,
\begin{align}
    y_{k} & = w_{k}^\intercal x^{\left(i\right)} + b_{k} = w_{k}^\intercal x^{\left(i\right)} - F_X^{-1}\left( \frac{k}{K} \right) \nonumber \\
    & = w_{k}^\intercal x^{\left(i\right)} + s \cdot \mathrm{sgn}\left( \frac{k}{K} - 0.5 \right) \ln \left(1 - 2 \left| \frac{k}{K} -0.5 \right| \right).
\end{align}
Without loss of generality, assume that $\left(k/K - 0.5\right) > 0$,
\begin{equation}
    y_{k} = w_{k}^\intercal x^{\left(i\right)} + s \cdot \ln \left(2 - \frac{2k}{K} \right).
\end{equation}
Therefore, we have
\begin{align}
    & \mathrm{Pr} \left\{i_0 = k, n=K\right\} =  \mathrm{Pr} \left\{ y_k > 0 \; \mathrm{and} \; y_{k+1} \leq 0 \right\} \nonumber \\
    = & \mathrm{Pr} \left\{ p\left( \frac{2k}{K} \right) <  w_{k}^\intercal x^{\left(i\right)} \leq  p\left( \frac{2\left( k+1\right)}{K} \right) \right\} \nonumber \\
    = & \mathrm{Pr} \left\{ p\left( \frac{2\cdot2k}{2K} \right) <  w_{k}^\intercal x^{\left(i\right)} \leq  p\left( \frac{2\left( 2k+1\right)}{2K} \right) \right\} + \nonumber \\
    & \mathrm{Pr} \left\{ p\left( \frac{2\left( 2k+1\right)}{2K} \right) <  w_{k}^\intercal x^{\left(i\right)} \leq  p\left( \frac{2\cdot 2 \left( k+1\right)}{2K} \right) \right\} \nonumber \\
    = & \mathrm{Pr} \left\{i_0 = 2k, n=2K\right\} + \mathrm{Pr} \left\{i_0 = 2k+1, n=2K\right\},
\end{align}
since we set that $w_{1} = w_{2} = \cdots w_{n}$, abbreviating $- s \cdot \ln \left(2 - 2k/K \right)$ to $p\left( 2k/K\right)$.

The above results show that $i_0 = k$ means that $x^{\left(i\right)}$ belongs to $\left[ p\left( 2k/K \right) \oslash w_{k}^\intercal, p\left( 2\left( k+1 \right)/K \right) \oslash w_{k}^\intercal \right]$, which we refer to \textit{reverse interval}. Increasing $n$ to $2K$ is equivalent to dividing the reverse interval into two parts.

Given two samples $x^{\left(i\right)}$ and $x^{\left(j\right)}$, we have $i_0 \neq j_0$ when $n=2K$ if $i_0 \neq j_0$ when $n=K$ since the reverse intervals of $i_0$ and $j_0$ are not disjoint. Otherwise, assume that $i_0 = j_0 = k$ when $n=K$, increasing $n$ to $2K$ makes $i_0$ and $j_0$ change to $2k$ or $2k+1$, resulting in four combinations of $i_0$ and $j_0$. However, $i_0 = j_0$ only occurs in two of these combinations. Increasing $n$ from $K$ to $2K$ reduces the probability that $i_0$ equals $j_0$. When gradients of all training samples are separated into different units, the reverse indexes of the samples are $\left\{1_0, 2_0, \cdots, B_0 \right\}$ where $B$ is the batch size. In other words, only units in the above reverse index set have non-zero gradients, and $\left| \left\{1_0, 2_0, \cdots, B_0 \right\} \right| = B$.
\end{proof}

Theorem~\ref{theorem-two-performance-related} shows that the adversary can increase the number of units to separate gradients as much as possible, and the growth of units would not cause gradient expansion.

\subsection{Removing Background Gradients}
\label{section-removing-background}

We further compress gradients by removing background gradients to reduce the compressed degree of gradients in clipping. As shown in Fig.~\ref{fig-introduction-proposed-dp}, the image's subject is a dog, while the background is worthless to the adversary. Inspired by this, we propose keeping the pixels where subjects are located before samples enter the separation layer while the rest are set to 0. The above operation makes the gradient corresponding to the pixel with a value of 0 in the image also be 0 in the separation layer.

The implementation of subject extraction is based on the segment anything model (SAM)~\cite{kirillov2023segment}. SAM has significant advantages in image segmentation, and many machine learning models apply SAM to improve target models' performance~\cite{zhang2023comprehensive, ma2023segment,jing2023segment}. Most importantly, SAM is a zero-shot model, i.e., SAM can directly apply to all user samples without users performing any training on SAM.

SAM generates masks for any image to segment the image with multiple input modes. We set the center of images as the selecting (input) points and apply masks with higher scores to samples. Pixels in the mask with the highest score are kept, while the rest will be set to 0 for gradient compression.

\subsection{Sample Denoising}
\label{section-noise-filtering}

Since samples only retain subjects through masks generated by SAM while other pixels (i.e., backgrounds) are set to 0, the background pixels in reconstructed samples should also be 0. The sample denoising aims to restore the background pixels of reconstructed samples to 0 through noise filtering.

We first analyze the cause of noise in the backgrounds of reconstructed samples. Consider a pixel $p^{(i)}_{c,w,h}$ in the background of sample $x^{\left( i\right)}$, after subject extraction, $p^{(i)}_{c,w,h}$ is set to 0. Assuming that the relevant gradient is scaled by $\omega$ in gradient clipping, the reconstructed value $\hat{p}^{(i)}_{c,w,h}$ can be calculated as follows when the perturbation is not considered:
\begin{equation}
\hat{p}^{(i)}_{c,w,h} = \frac{\nabla_{w_{i_0, c,w,h}} L / \omega}{\nabla_{b_{i_0}} L / \omega} = \frac{\nabla_{w_{i_0, c,w,h}}L}{\nabla_{b_{i_0}}L},
\end{equation}
where $i_0$ is the reverse unit, and $w_{i_0, c,w,h}$ and $b_{i_0}$ are the corresponding weight and bias, respectively. When $p^{(i)}_{c,w,h} = 0$, we have $\nabla_{w_{i_0, c,w,h}}L = 0$ and $\hat{p}^{(i)}_{c,w,h} = 0$. When gradients are perturbed,
\begin{equation}
\label{eq-reconstructed-pixel}
\hat{p}^{(i)}_{c,w,h} = \frac{ 0 + n_w }{\nabla_{b_{i_0}} L / \omega + n_b} \neq 0,
\end{equation}
where $n_w$ and $n_b$ are noise for gradient perturbation, and this is why there is noise in the background of reconstructed samples.

Gradient filtering introduces extra zero gradients that reflect the noise distribution after perturbation to obtain the noise's confidence interval. As shown in Fig.~\ref{fig-attack-framework}, before passing through the separation layer, samples pass through a convolution layer with a kernel size of 1, a step size of 1, and an output channel of 6 without bias term. The weights of the first 3 channels and the last 3 channels are set to 1 and 0, respectively. The output of this layer is the same sample as processed samples and a vector of all zeros with the same size as the processed samples. The adversary can collect noise samples and get the noise distribution after perturbation from the corresponding gradients of introduced zero gradients.

Specifically, since the noise in LDP follows a normal distribution, the noise with negative values satisfies a half-normal distribution with a mean of $- \sigma \sqrt{2} / \sqrt{\pi}$ and a variance of $\sigma^2\left( 1 - 2/\pi\right)$ where $\sigma$ is the standard deviation of the noise~\cite{cooray2008generalization}. Extra zero gradients are artificially added to gradients, and there are enough noise samples in the gradient to infer $\sigma$ by the mean and variance of noise according to the law of large numbers (LLN)~\cite{dekking2005modern}. Further, we can generate a confidence interval $\left[-c, c\right]$ for noise by $\sigma$.

\begin{figure}
\centering
\includegraphics[width=3.2in]{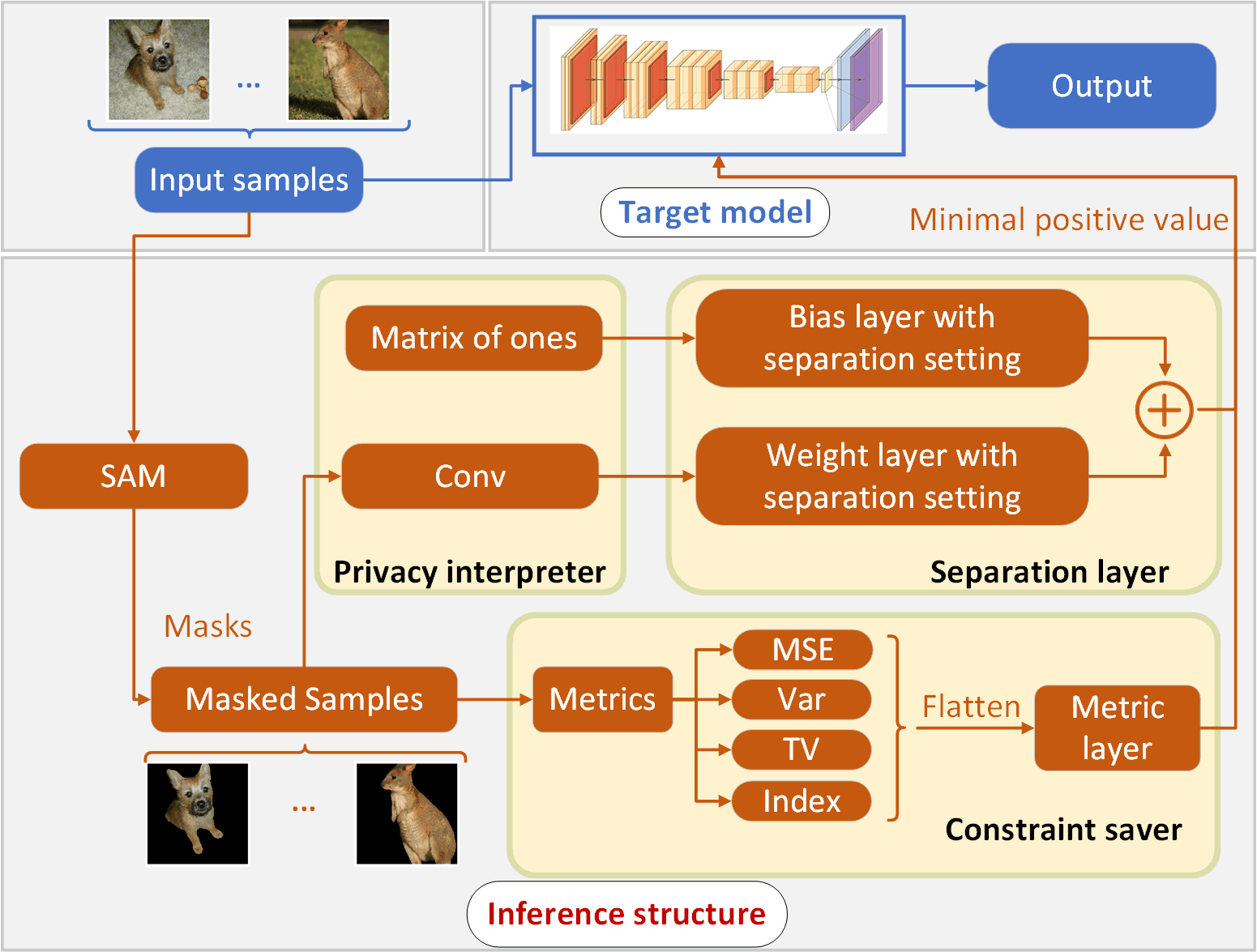}
\caption{The framework of the global model with any target models in the proposed attack.}
\label{fig-attack-framework}
\end{figure}

Note that the noise in backgrounds is scaled with different factors according to Equation~\eqref{eq-reconstructed-pixel}, i.e., $n_w$ is scaled by $\left(\nabla_{b_{i_0}} L / \omega + n_b\right)$, meaning the confidence interval also should be scaled. We propose an improved structure and modify the bias term of the primary inference structure so that bias gradients are repeated and identical. Specifically, as shown in Fig.~\ref{fig-attack-framework}, we delete the bias term of the FCL in the separation layer (i.e., the \textit{weight layer}) and use a \textit{bias layer} without the bias term to generate bias of the weight layer. The weights of the bias layer are set to the quantile function of a random variable following the Laplace distribution, which is mentioned in Section~\ref{section-removing-redundate-gradients}, and its input is a matrix of ones. The weight and bias layer output are added to the target model's input. The improved separation layer comprising the weight and bias layers satisfies the following theorem.

\begin{theorem}
\label{theorem-improved-structre}
In the improved separation layer, for any sample $x^{\left(i\right)}$, weight gradients of sample $x^{\left(i\right)}$ only exists in $i_0$-th unit of the weight layer, and gradients of the bias layer are repeated and identical bias gradients of sample $x^{\left(i\right)}$.
\end{theorem}
\begin{proof}
The weights of the weight layer are $w_1 = w_2 = \cdots = w_K$ where $K$ is the number of units. Suppose that the input size of the bias layer is $D$, the weights of the bias layer are set according to $b_{j,k} = - F^{-1}_X( k / K)$ where $X \sim \textrm{Laplace}(0,s)$, $F_X(\cdot)$ is the CDF of $X$ and $j \in \left[1, D\right]$. The weight layer is an FCL without bias term, given any input $x^{\left( i\right)}$, the output of the weight layer is $\left\{ w_1^\intercal x^{\left( i\right)}, w_2^\intercal x^{\left( i\right)}, \dots, w_K^\intercal x^{\left( i\right)} \right\}.$
Since the bias layer is an FCL without bias term, the input of the bias layer is a vector of ones, and $b_{1,k} = b_{2,k} = \cdots = b_{D,k}$, the output of the bias layer is
\begin{align}
    \sum_{j=1}^{D} \left\{  b_{j,1}, b_{j,2}, \dots, b_{j,K} \right\} 
    =  D \cdot \left\{  b_{1,1}, b_{1,2}, \cdots, b_{1,K}\right\}.
\end{align}
We have $b_{1,1} > b_{1,2} > \cdots > b_{1,K}$ according to the proof of Theorem~\ref{theorem-one-non-zero-gradients}, and $D \cdot b_{1,1} > D \cdot b_{1,2} > \dots > D \cdot b_{1,K}$. Thus, any sample $x^{\left( i\right)}$ has a unique reverse index $i_0$ such that
\begin{equation}
    \left\{ \begin{matrix}
        w_1^\intercal x^{\left(i\right)} + D \cdot b_{1,1} > \cdots > w_{i_0}^\intercal x^{\left(i\right)} + D \cdot b_{1, i_0} > 0; \\ \\
        w_K^\intercal x^{\left(i\right)} + D \cdot b_{1,K} < \cdots < w_{i_0+1}^\intercal x^{\left(i\right)} + D \cdot b_{1, i_0+1} \leq 0.
        
    \end{matrix}\right.
\end{equation}
The output of the improved inference structure only retains the minimal non-zero positive value, i.e., $y^{\left(i\right)}_{\mathrm{min}} = w_{i_0}^\intercal x^{\left(i\right)} + \sum_{j=1}^{D} b_{j, i_0}.$
Similar to the proof of Theorem~\ref{theorem-two-performance-related}, weight gradients of sample $x^{\left( i\right)}$ only exist in the $i_0$-th unit of the weight layer which is the only unit in the weight layer that affects the value of $y^{\left(i\right)}_{\mathrm{min}}$. For bias gradients, we have
\begin{equation}
    \left\{ \begin{matrix}
        \nabla_{w_{i_0}}L = \frac{\partial L}{\partial w_{i_0}} = \frac{\partial L}{\partial y^{\left(i\right)}_{\mathrm{min}}} \cdot \frac{\partial y^{\left(i\right)}_{\mathrm{min}}}{\partial w_{i_0}} = \frac{\partial L}{\partial y^{\left(i\right)}_{\mathrm{min}}} \cdot x^{\left(i\right)}; \\
        \; \\
        \nabla_{b_{j, i_0}}L = \frac{\partial L}{\partial b_{j, i_0}} = \frac{\partial L}{\partial y^{\left(i\right)}_{\mathrm{min}}} \cdot \frac{\partial y^{\left(i\right)}_{\mathrm{min}}}{\partial b_{j, i_0}} = \frac{\partial L}{\partial y^{\left(i\right)}_{\mathrm{min}}}.
    \end{matrix}\right.
\end{equation}
Any gradients of $b_{j, i_0}$ where $j\in \left[ 1,D\right]$ are bias gradients of sample $x^{\left( i\right)}$, and we have 
\begin{equation}
    \nabla_{b_{1, i_0}}L = \cdots = \nabla_{b_{D, i_0}}L = \frac{\partial L}{\partial y^{\left(i\right)}_{\mathrm{min}}}.
\end{equation}
In other words, there are $D$ identical bias gradients in the improved inference structure. The adversary can reconstruct sample $x^{\left( i\right)}$ by
\begin{equation}
    x^{\left( i\right)} = \nabla_{w_{i_0}}L \oslash \left( \frac{1}{D} \sum_{j=1}^{D} \nabla_{b_{j, i_0}}L  \right).
\end{equation}
\end{proof}

The intuition of the proof is that the bias of the weight layer changes from one value to multiple identical values (the number is equal to the input size of the bias layer). Therefore, we can obtain an accurate $\nabla_{b_{i_0}} L / \omega$ by averaging the corresponding gradients in the bias layer. For example, assume that the weight layer has 1024 units and the input size of the bias layer is 500, $\nabla_{b_{i_0}} L / \omega + n_b$ appears once in the primary inference structure. However, in the improved separation layer, $\left( \nabla_{b_{i_0}} L / \omega + n_b \right)$ appears 500 times in the gradients. Although noise in the repeated bias gradients ($n_b$) is different, the averaging effectively realizes noise cancellation according to the LLN because their mean is 0.

Finally, for any sample $x^{(i)}$, we scale the confidence interval $\left[-c, c\right]$ by the averaged $\nabla_{b_{i_0}} L / \omega$. When the value of a specific pixel $\hat{p}^{(i)}_{c,w,h}$ is in the interval, we think that $\hat{p}^{(i)}_{c,w,h}$ has a high probability of being 0, and set $\hat{p}^{(i)}_{c,w,h}$ to 0, which can effectively filter the noise in backgrounds.

\subsection{Metric-based Optimization}
\label{section-mectics-optimization}

The above noise filtering is mainly aimed at noise in backgrounds, and we further propose metric-based optimization to improve the quality of reconstructed samples. We set the optimization objective as
\begin{align}
\label{optimization-goal}
\underset{\hat{x}}{\min} & \; w_\mu \sum_l \|\mu_l\left(\hat{x}\right) - \mu_l\left(x\right)\|_2 +  w_\sigma \|\sigma_l^2\left(\hat{x}\right) - \sigma_l^2\left(x\right)\|_2 + \nonumber \\ & \; w_{\mathrm{TV}} \sum_l \| \mathrm{TV}_l\left(\hat{x}\right) - \mathrm{TV}_l\left(x\right) \|_2,
\end{align}
where $\mu_l\left(\cdot\right)$, $\sigma_l^2\left(\cdot\right)$, and $\mathrm{TV}_l\left(\cdot\right)$ are sample-wise mean, variance and total variation, respectively, and $w_\mu$, $w_\sigma$, and $w_{\mathrm{TV}}$ are weight coefficients. The problem is obtaining the above information of processed samples for optimization.

As shown in Fig.~\ref{fig-attack-framework}, we introduce a \textit{metric layer} to imprint the above information to gradients. Specifically, the model calculates the corresponding batch-wise metrics of processed samples to generate a metric matrix, which then is flattened to become the metric layer input. The output of the metric layer is directly added to the input of the target model. As discussed in Lemma~\ref{lemma-attack-1-one-sample}, the adversary can reconstruct the input of any FCL through its gradients when the batch size is 1 (the flattened metric matrix can be viewed as a single sample). For example, a batch of 16 images with 3 channels can generate a $16 \times 3 \times 3$ metric matrix, which is then flattened to $1\times 144$ as an input of the metric layer.

Another problem is that reconstructed samples are ordered by samples' reverse units, but metrics are ordered by input samples, which makes the order of reconstructed samples and metrics inconsistent. The metrics are sample-wise, and the non-corresponding order would optimize reconstructed samples in the wrong direction. Therefore, we also save sample reverse units in the gradients of the metric layer. The corresponding gradients can reflect the reverse units of samples, and the adversary can reorder the metrics according to these gradients so that the orders of metrics and reconstructed samples are consistent. In addition, the metric layer introduces a repeated gradient structure as the bias layer to realize noise cancellation and improve accuracy. Finally, the adversary can optimize the reconstructed samples according to ordered metrics according to Equation~\eqref{optimization-goal}.

\section{All-in-one and Implementation}
\label{section-all-in-one}

\subsection{Model Setting and Distribution}
We introduce the global model with any target models in which the inference structure is embedded through Fig.~\ref{fig-attack-framework}. In the local training process, samples pass through the SAM and the target model. The gradients in the target model are not affected by the proposed attack for non-target users, ensuring the accuracy of the converged target model. As discussed in Section~\ref{section-removing-background}, the SAM generates masked samples to remove background gradients. The masked samples are then sent to both the convolution and metric layers. The output of the convolution layer is the input of the weight layer, and the input of the bias layer is a matrix of ones. As discussed in Section~\ref{section-attack}, the above setting leaves sample information in the gradient without gradient expansion so the server can load the sample information from gradients for reconstruction attack.
The gradients of the metric layer are used to optimize the reconstructed samples, as discussed in Section~\ref{section-mectics-optimization}. When the user is a victim, the input of the target model is a linear combination of the outputs of weight, bias, and metric layers. For non-target users, the input of the target model is the original samples, which means that the inference structure does not impact the training of the target model.

In the model distribution phase, the adversary distributes the global model with the target model and the designed inference structure, as shown in Fig.~\ref{fig-attack-framework}, to non-target users and the victim. In order to prevent target model training of non-target users from being affected by the attack, the output coefficient of the inference structure of non-target users is set to small. The inference structure has little effect on non-target users, and the final model output of non-target users is almost consistent with the target model. The victim's gradient would be ignored in the gradient aggregation phase. Considering the thousands of users in FL, ignoring the victim's gradients has almost no effect on the aggregated gradient. We discuss the impact of the proposed attack on FL training and target model accuracy in Section~\ref{section-simulation-model-accuracy} through simulations. The victim generates gradients by training the global model with local samples, then clips and perturbs the gradients according to local privacy parameters and uploads the processed gradients to the server.

\begin{figure}
\centering
\includegraphics[width=3.2in]{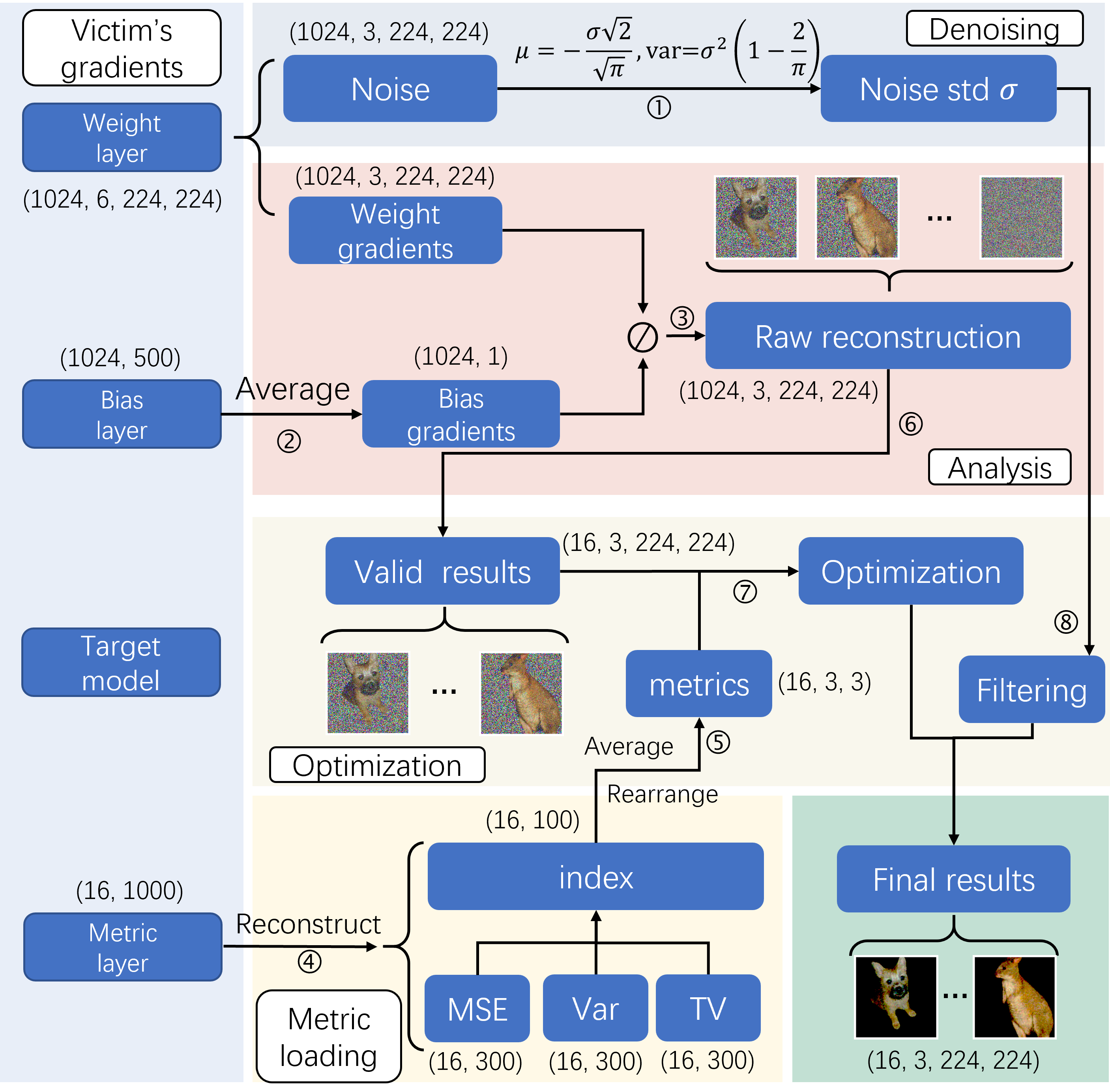}
\caption{The process of reconstructing training samples according to the victim's gradients with an example where the batch size is 16 and the sample size is $3 \times 224 \times 224$. The number of units in the weight, bias, and metric layers are 1024, 500, and 1000, respectively.}
\label{fig-attack-flow}
\end{figure}

\subsection{Implementation}
\label{section-all-in-one-implementation}

Figure~\ref{fig-attack-flow} shows the process of implementing the proposed sample reconstruction attack through the clipped and perturbed gradients uploaded by the victim. We cover each implementation detail step by step in the following.

\textbf{(1) Privacy parameters extraction.} As discussed in Section~\ref{section-noise-filtering}, the convolutional layer adds extra zero gradients the same size as training samples to the gradients of the weight layer. Through the mean and variance of negative gradients of the structure that introduces extra zero gradients before perturbation, the adversary obtains the standard variance $\sigma$ of the noise in gradient perturbation according to the mean of $-\sigma \sqrt{2} / \sqrt{\pi}$ and the variance of $\sigma^2\left( 1 - 2/\pi\right)$.

\textbf{(2) Bias term reconstruction.} As discussed in Section~\ref{section-noise-filtering}, the gradients of the bias layer are used to generate accurate bias gradients. Since bias gradients are repeated and identical in the gradients of the bias layer, we can obtain more accurate bias gradients by averaging.

\textbf{(3) Raw reconstruction.} In addition to noisy zero gradients, there are noisy weight gradients in the gradients of the weight layer. According to an element-wise division of weight and bias gradients discussed in Section~\ref{section-primacy-attack}, i.e., $\nabla_w L \oslash \nabla_b L$, the adversary performs a raw reconstruction for training samples.

\textbf{(4) Metric reconstruction.} As discussed in Section~\ref{section-mectics-optimization}, the metric layer gradients contain sample metrics, including mean, variance, total variation, and samples' reverse units. The adversary reconstructs the above metrics and samples' reverse units from metric layer gradients.

\textbf{(5) Metric alignment.} The order of reconstructed metrics and reconstruction samples is different, and the order of reconstructed samples can be inferred from the index of reverse units. Therefore, reconstructed metrics are reordered according to the reconstructed index of reverse units so that the orders of metrics and reconstructed samples are consistent.

\textbf{(6) Image filtering.} Since the number of units with non-zero gradients is not greater than the batch size according to Theorem~\ref{theorem-one-non-zero-gradients}, raw reconstruction contains many meaningless images. For example, there are 1008 meaningless images in the raw reconstruction in Fig.~\ref{fig-attack-flow}. However, these meaningless images are easy to distinguish, and reverse units reconstructed from the gradients of the metric layer provide accurate positions of valid samples in the raw reconstruction.

\textbf{(7) Metric-based optimization.} As discussed in Section~\ref{section-mectics-optimization}, the proposed attack establishes the optimization objective by reconstructed metrics and optimizes the reconstructed samples. The proposed attack optimizes the valid samples after image filtering according to Equation~\eqref{optimization-goal} to improve the reconstructed sample quality.

\textbf{(8) Noise filtering.} Generating confidence interval for noise according to the reconstructed standard variance $\sigma$ and scaling the confidence interval through average bias gradients of reverse units to filter the reconstructed samples' noise.

\section{Evaluation}
\label{section-evaluation}

\subsection{Evaluation Setup}
\label{section-evaluation-setup}

\textbf{Benchmark.}
We utilize the benchmark on Breaching, an open framework for reconstruction attacks against FL\footnote{\url{https://github.com/JonasGeiping/breaching}.}. We consider two analysis-based attacks (Fowl's attack~\cite{fowl2022robbing} and Boenisch's attack~\cite{Boenisch2021When}) in the evaluation. Both can almost completely reconstruct victims' training samples in FL mechanisms without LDP. We also introduce two optimization-based attacks for comparison (Yin's attack~\cite{Yin2021see} and Wei's attack~\cite{Wei2020Framework}). Breaching implements the above attacks. In addition, Hong's work~\cite{hong2024foreseeing} discusses model vulnerability through the Hessian matrix to gradient difference. We implement the optimization algorithm in Hong's work~\cite{hong2024foreseeing} for comparison.

\textbf{Datasets and models.}
We consider four image datasets: ImageNet~\cite{ILSVRC15}, CIFAR-100~\cite{krizhevsky2009learning}, Caltech-256~\cite{griffin2007caltech}, and Flowers102~\cite{nilsback2008automated}. Users' training samples are randomly selected from the above datasets, and training samples are normalized between 0 and 1. The default target model is ResNet101~\cite{he2016deep}.

\textbf{Metrics.}
Following metrics measure the quality of reconstructed samples: mean square error (MSE), peak signal-to-noise ratio (PSNR), and complex wavelet structural similarity (CW-SSIM)~\cite{cw-ssim2009Sampat}. MSE reflects the difference between two images. PSNR quantifies reconstruction quality for images subject to lossy compression, and a higher PSNR generally indicates that given images have a higher reconstruction quality. CW-SSIM is an index that varies between 0 and 1 to measure the similarity of two images, and a larger CW-SSIM refers to a higher similarity between two images. The above metric calculates the similarity between masked training samples and reconstructed samples generated by the proposed attack~\cite{Nguyen2024Preserving,Song2023lsecnet}.

\textbf{Parameter setting.}
The Gaussian noise scale is calculated based on the setting of the existing FL mechanism with LDP~\cite{Zhou2022Differentially, Kang2020FLwithDP, Naseri2022local}, i.e., $\sigma_U = 2cC / m\varepsilon$ where $c$ is a constant, $C$ is the clipping bound, $m$ is the minimal size of local datasets, and $\varepsilon$ is the privacy parameter. The default values of the above parameters are close to those in the experiments of existing works, specifically, $c=1$, $C=10$, $m=1000$, $\varepsilon=10$, and $\delta=0.01$. As shown in Table \ref{table-prvacy-parameters}, setting $\varepsilon$ to 10 in the existing FL mechanisms with LDP causes significant performance degradation to the target model. Smaller $\varepsilon$ will make the target model wholly inapplicable due to low accuracy. The data presented in the evaluation is the average of 10 results under the same conditions. The adversary implements attacks through a round of gradients. Other parameters are given in Table $\mathrm{VIII}$ of Appendix A.
% Other parameters are given in Table~\ref{table-hyper-parameters} of Appendix~\ref{appendix-different-gradients-protection}.

\textbf{Experiment environment.}
All experiments are conducted on a system equipped with an NVIDIA RTX 4090D GPU with 24GB memory and an Intel Xeon Platinum 8474C CPU. The proposed attack is implemented by PyTorch, and all computations are accelerated by CUDA.

\subsection{Reconstructed Samples Quality}
\label{section-evaluation-quality-comparison}

\begin{table*}
\caption{Original training samples from ImageNet and reconstructed samples generated by different attacks in FL with LDP.}
\centering
\label{table-evaluation-vision}
\begin{tblr}{
    column{1} = {6.8in, m}, hline{1,Z}={2pt}, hline{3,5,7,9,11,13,15,17}={dashed,1pt}, rows={c},
    row{1,3,5,7,9,11,13}={b,5pt}, row{2,4,6,8,10,12,14}={t,10pt},
}
Original training samples\\
\includegraphics*[width=6.8in,valign=m]{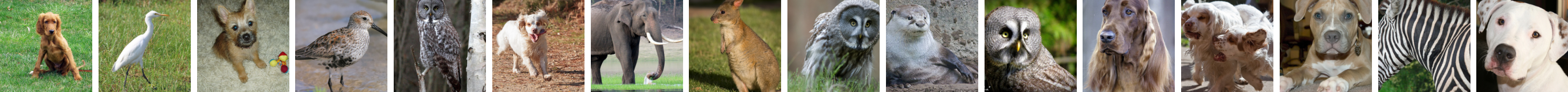}\\
\textbf{The proposed attack}\\
\includegraphics*[width=6.8in,valign=m]{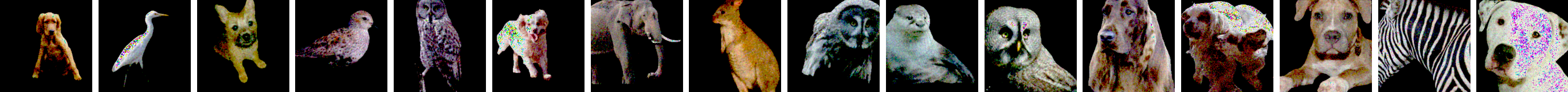}\\
Proposed attack without optimization\\
\includegraphics*[width=6.8in,valign=m]{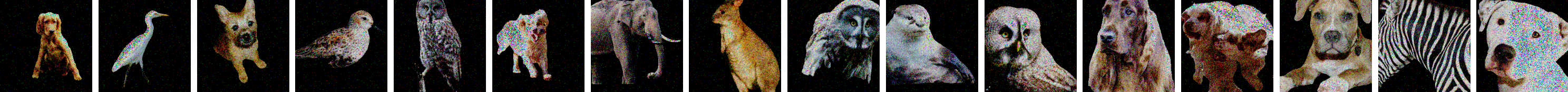}\\
Proposed attack without denoising\\
\includegraphics*[width=6.8in,valign=m]{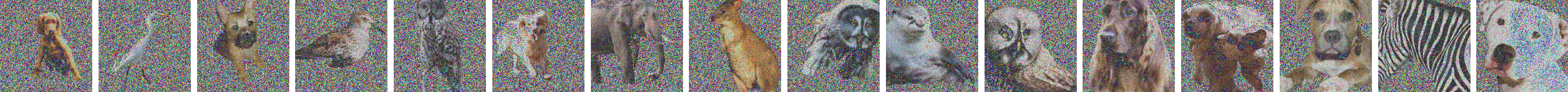}\\
Fowl's work~\cite{fowl2022robbing}\\
\includegraphics*[width=6.8in,valign=m]{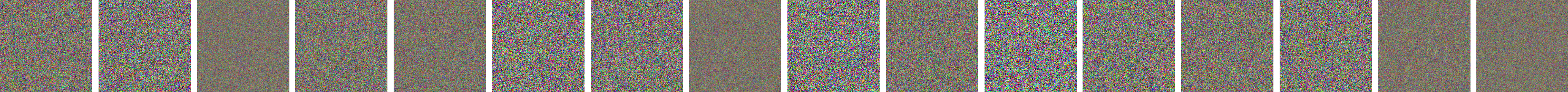}\\
Boenisch's attack~\cite{Boenisch2021When}\\
\includegraphics*[width=6.8in,valign=m]{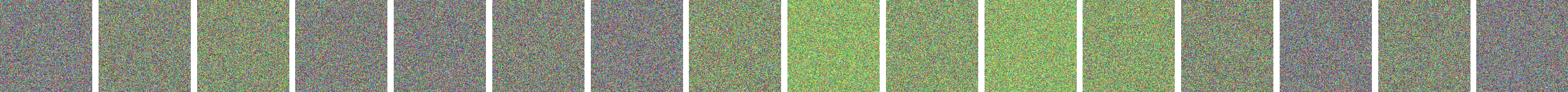}\\
Yin's work~\cite{Yin2021see}\\
\includegraphics*[width=6.8in,valign=m]{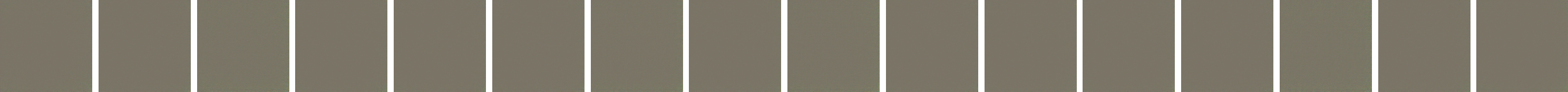}\\
Wei's work~\cite{Wei2020Framework}\\
\includegraphics*[width=6.8in,valign=m]{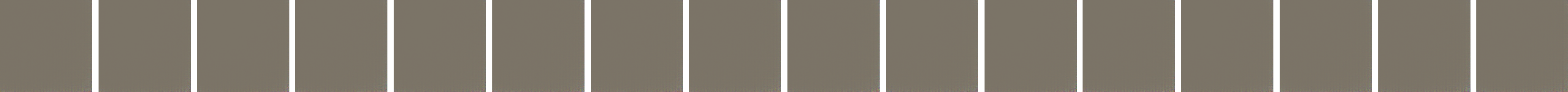}\\
Hong's work~\cite{hong2024foreseeing}\\
\includegraphics*[width=6.8in,valign=m]{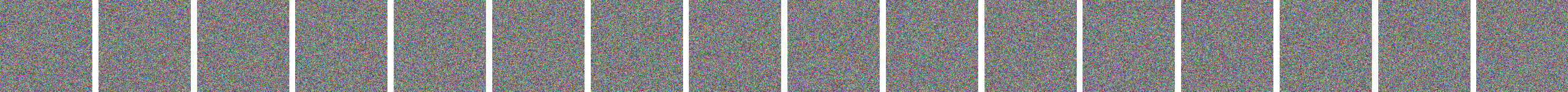}\\
\end{tblr}
\label{table-experiment}
\end{table*}

We compare the reconstructed samples' quality of various attacks when gradients are clipped and perturbed in FL with LDP. The victim has the same samples when the adversary performs different attacks to compare the effectiveness of different attacks. We first set the batch size to 16 and discuss the impact of  batch sizes on the performance in Section~\ref{section-evaluation-factors}. The global target model is ResNet101~\cite{he2016deep}, while both the clipping bound and $\varepsilon$ are set to 10.

Table~\ref{table-evaluation-vision} compares the victim's original samples from ImageNet and reconstructed samples generated by various attacks. Data in random reconstruction is randomly generated by a uniform distribution from 0 to 1. Most samples reconstructed by the proposed attack can provide the primary information in the training samples.  Meanwhile, there is a slight noise in the reconstructed samples, and some reconstructed images cannot present any information since LDP protects victims' gradients. Other sample reconstruction attacks hardly reconstruct the victim's samples when gradients are clipped and perturbed.

Table~\ref{table-evaluation-metrics-comparison-dp} compares the quality of samples reconstructed by various attacks under different datasets. \textit{Without optimization} is the result of the proposed attack without metric-based optimization. \textit{Without denoising} is the result of the proposed attack without image denoising (noise filtering and metric-based optimization). The above two settings are to present the effect of noise filtering and metric-based optimization on the proposed attack. 
Furthermore, Table $\mathrm{X}$ in Appendix A compares the quality of reconstructed samples of different attacks against gradients without LDP protection, clipped gradients, and perturbed gradients, respectively.
% Furthermore, Table~\ref{table-evaluation-metrics-comparison-no-dp} in Appendix~\ref{appendix-different-gradients-protection} compares the quality of reconstructed samples of different attacks against gradients without LDP protection, clipped gradients, and perturbed gradients, respectively.
Figure~\ref{fig-evaluation-one-different-protection} provides a more straightforward presentation of training samples reconstructed by the proposed attack when different LDP materials protect victims' gradients. Most attacks can extract victims' sensitive information from reconstructed samples without LDP protection. The analysis-based attacks (the proposed attack, Fowl's attack~\cite{fowl2022robbing}, and Boenisch's attack~\cite{Boenisch2021When}) almost wholly reconstruct the original samples.

\begin{table*}
\caption{Quality of samples reconstructed by various attacks under different datasets in FL with LDP.}
\centering
\label{table-evaluation-metrics-comparison-dp}
\begin{tblr}{
    columns = {0.42in, l, m},
    hline{1,Z}={2pt}, hline{2}={2,3,5,6,8,9,11,12,13}{1pt}, hline{3}={1pt}, hline{2}={4,7,10}{rightpos=-1, 1pt},
    hline{4-Y}={dotted},
    column{1} = {1.2in},
    colsep = {2pt},
}
\SetCell[r=2]{c,m} Attack & \SetCell[c=3]{c} ImageNet & & & \SetCell[c=3]{c} CIFAR100 & & & \SetCell[c=3]{c} Caltech256 & & & \SetCell[c=3]{c} Flowers102\\
& MSE & PSNR & SSIM & MSE & PSNR & SSIM & MSE & PSNR & SSIM & MSE & PSNR & SSIM\\
\textbf{Proposed attack} & \textbf{0.0002} & \textbf{28.520} & \textbf{0.836} & \textbf{0.0018} & \textbf{26.869} & \textbf{0.868} & \textbf{0.0016} & \textbf{29.670} & \textbf{0.853} & \textbf{0.0017} & \textbf{29.705} & \textbf{0.881}\\
Without optimization & 0.0004 & 25.800 & 0.481 & 0.0039 & 25.799 & 0.482 & 0.0039 & 25.806 & 0.482 & 0.0039 & 25.778 & 0.481\\
Without denoising & 0.0067 & 21.805 & 0.384 & 0.0069 & 21.795 & 0.384 & 0.0068 & 21.803 & 0.385 & 0.0070 & 21.785 & 0.384\\
Fowl's attack~\cite{fowl2022robbing} & 0.319 & 4.973 & 0.204 & 0.316 & 5.015 & 0.199 & 0.318 & 4.994 & 0.203 & 0.332 & 4.794 & 0.191\\
Boenisch's attack~\cite{Boenisch2021When} & 0.168 & 8.234 & 0.213 & 0.164 & 8.531 & 0.134 & 0.199 & 7.503 & 0.185 & 0.186 & 7.760 & 0.195\\
Yin's attack~\cite{Yin2021see} & 0.067 & 12.129 & 0.139 & 0.069 & 12.443 & 0.153 & 0.099 & 10.677 & 0.113 & 0.130 & 8.947 & 0.196\\
Wei's attack~\cite{Wei2020Framework} & 0.063 & 12.179 & 0.093 & 0.115 & 11.036 & 0.127 & 0.099 & 10.652  & 0.108 & 0.127 & 9.924 & 0.109\\
Hong's work~\cite{hong2024foreseeing} & 0.3140 & 6.249 & 0.069 & 0.436 & 6.377 & 0.053 & 0.327 & 4.528 & 0.061 & 0.216 & 4.320 & 0.094\\
Random guess & 0.155 & 8.165 & 0.214 & 0.156 & 8.213 & 0.231 & 0.176 & 7.665 & 0.187 & 0.174 & 7.659 & 0.186\\
\end{tblr}
\end{table*}

\begin{figure}
\centering
\subfigure[]{
    \centering
    \includegraphics*[width=0.52in]{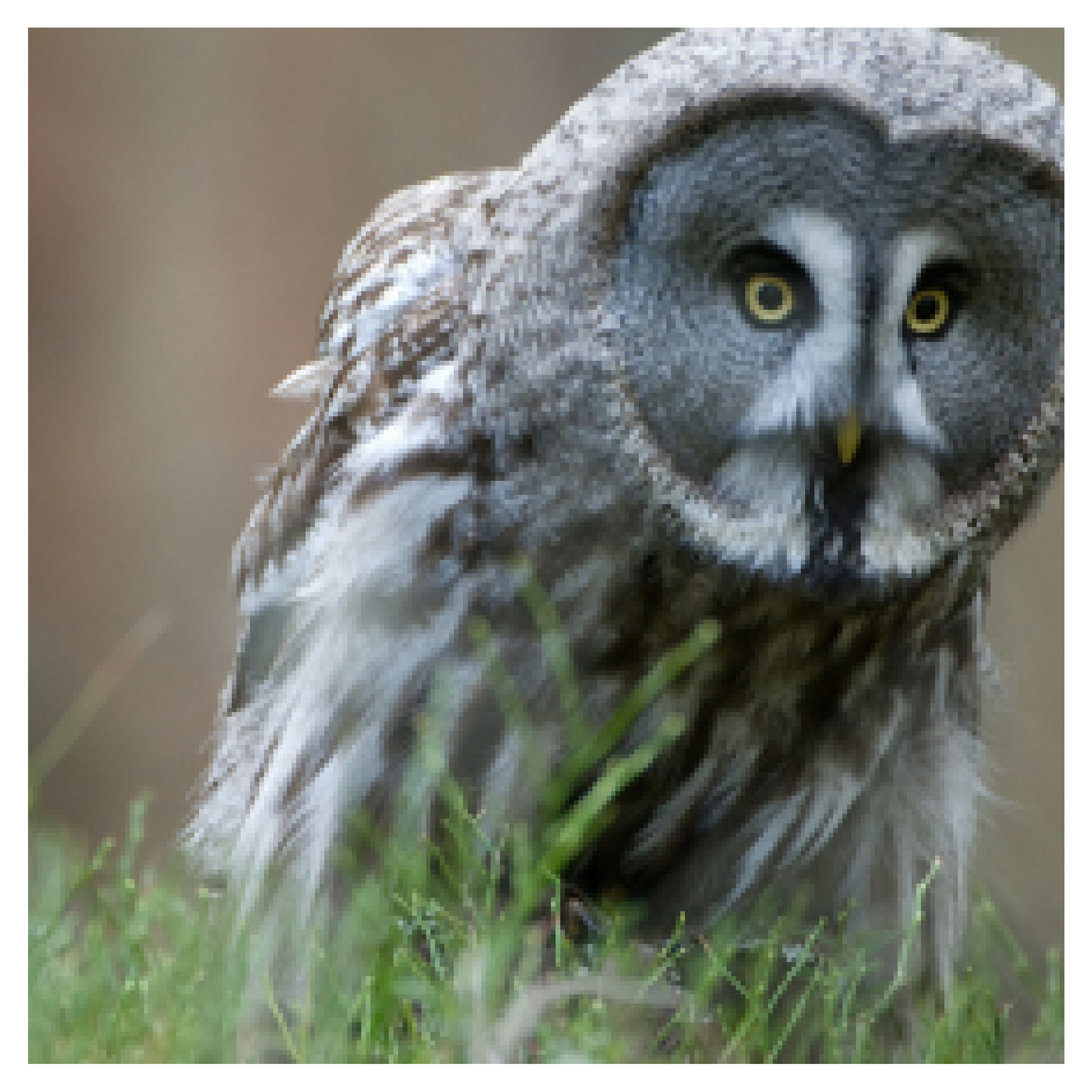}
}\hspace{-0.1in}
\subfigure[]{
    \centering
    \includegraphics*[width=0.52in]{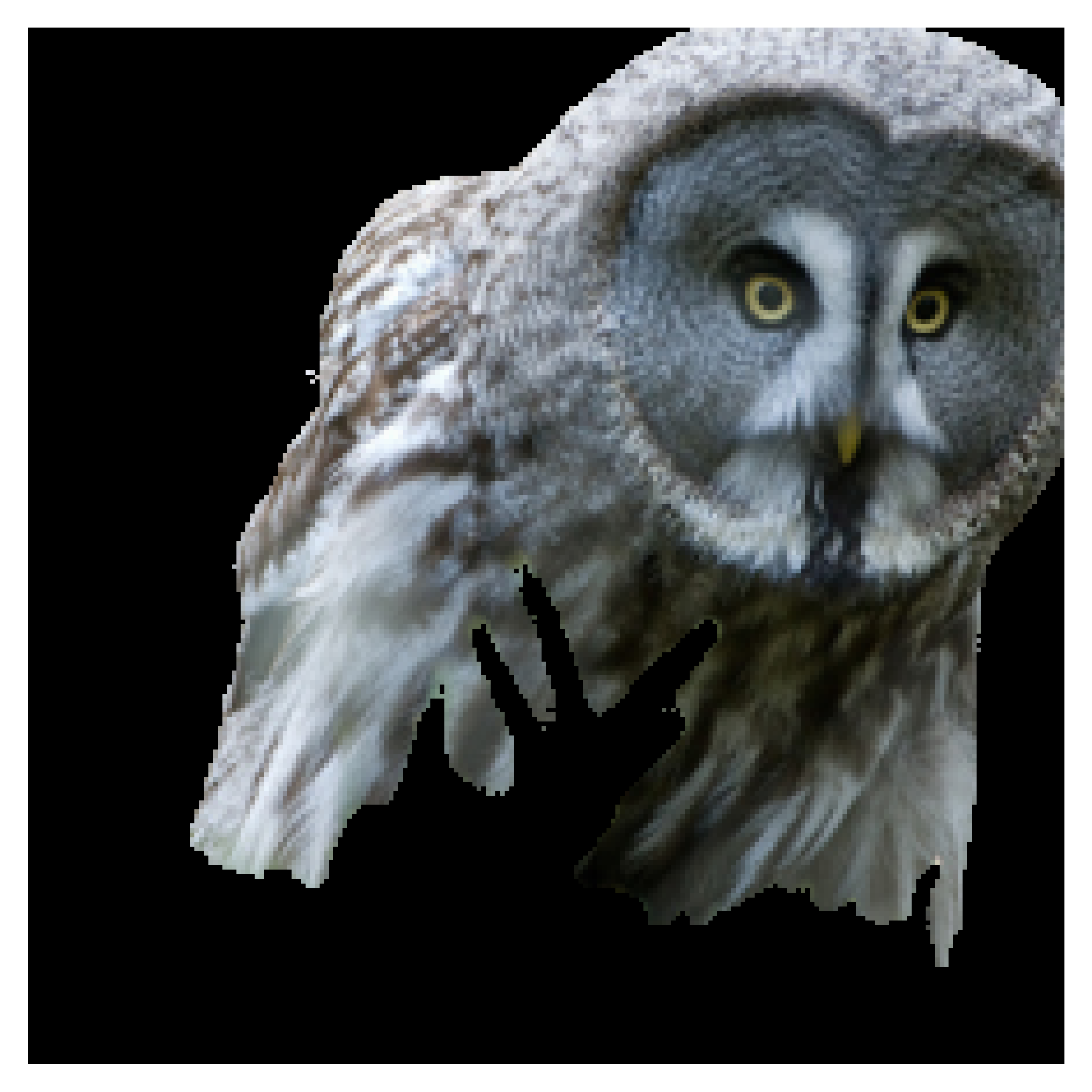}
}\hspace{-0.1in}
    \subfigure[]{
        \centering
    \includegraphics*[width=0.52in]{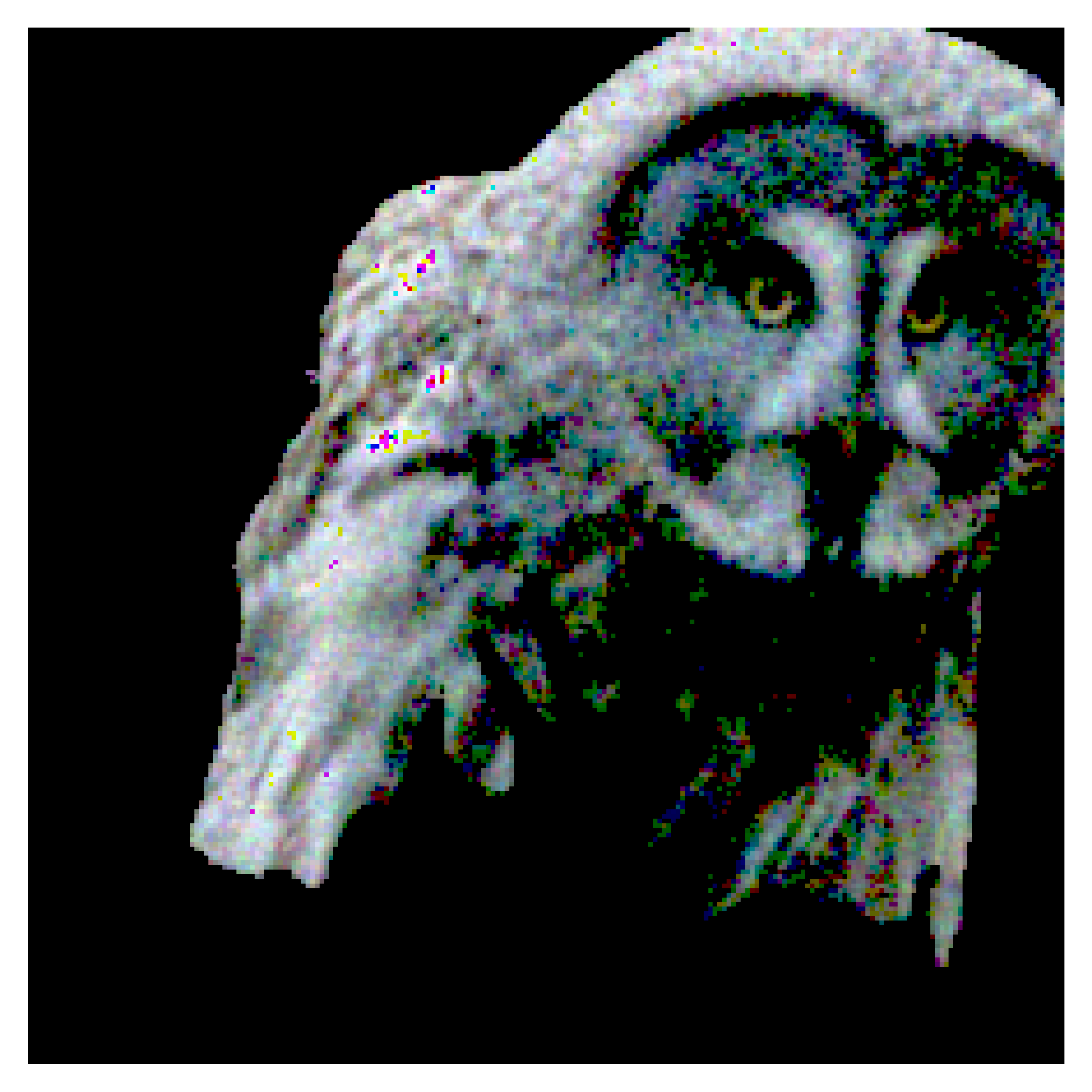}
    }\hspace{-0.1in}
\subfigure[]{
    \centering
    \includegraphics*[width=0.52in]{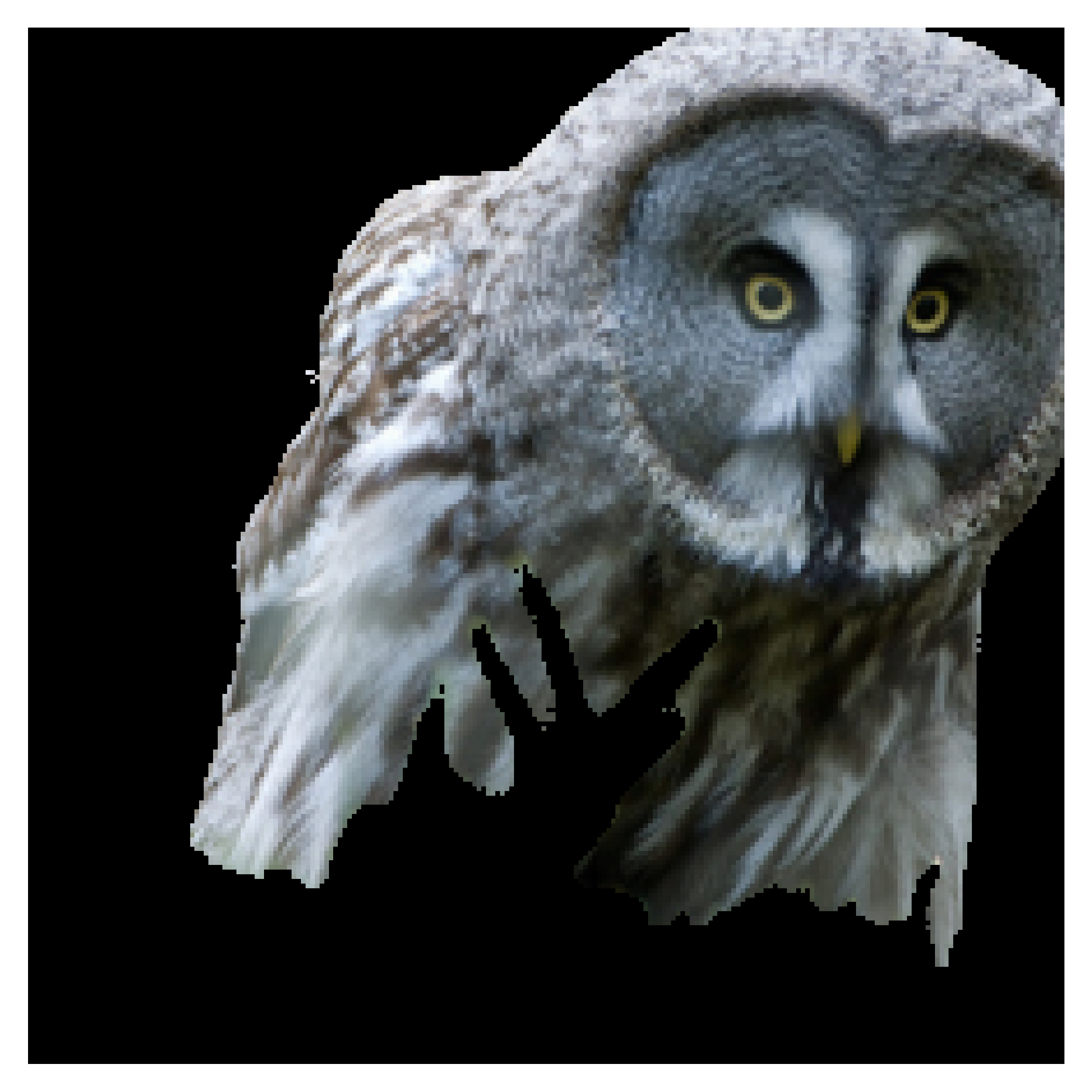}
}\hspace{-0.1in}
\subfigure[]{
    \centering
    \includegraphics*[width=0.52in]{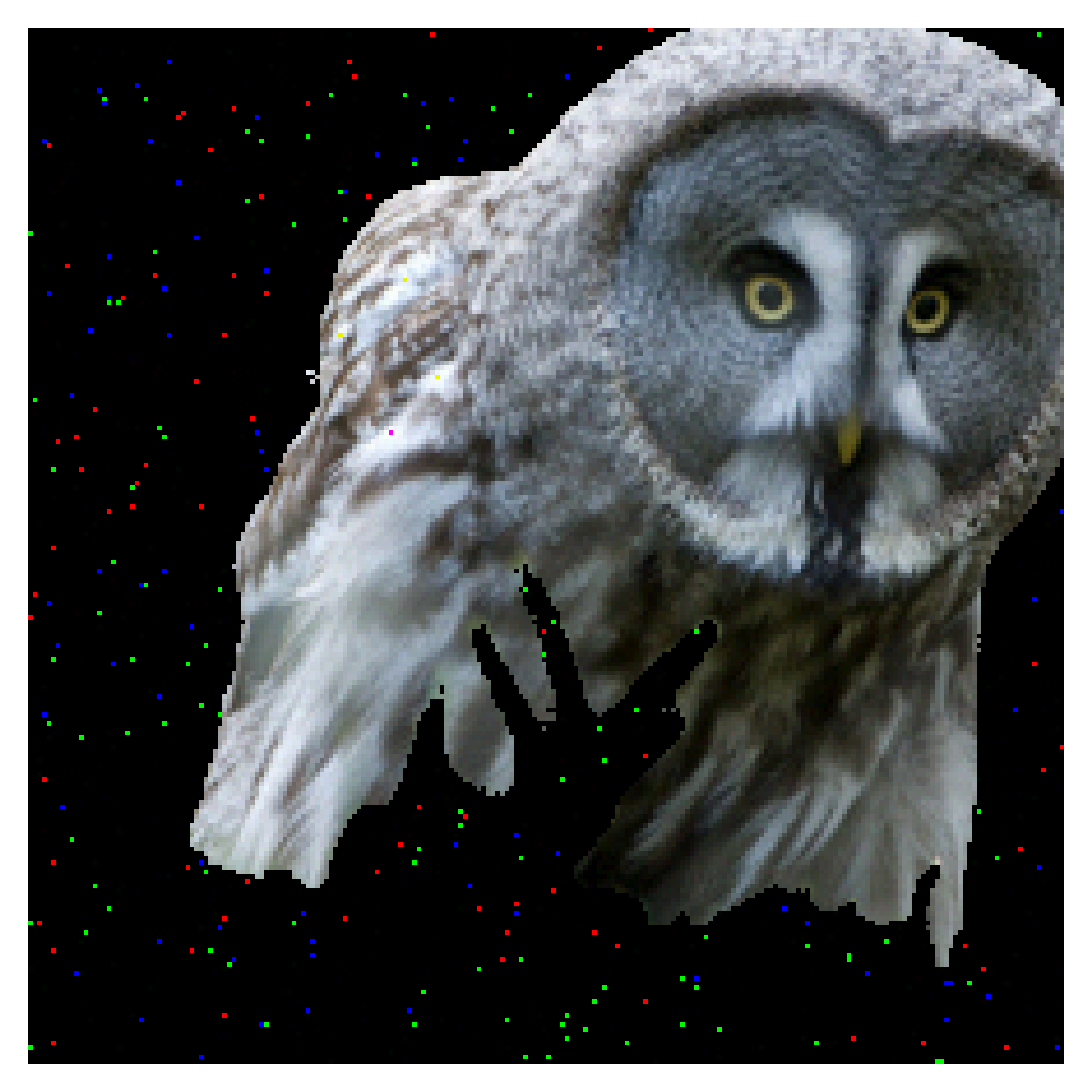}
}\hspace{-0.1in}
\subfigure[]{
    \centering
    \includegraphics*[width=0.52in]{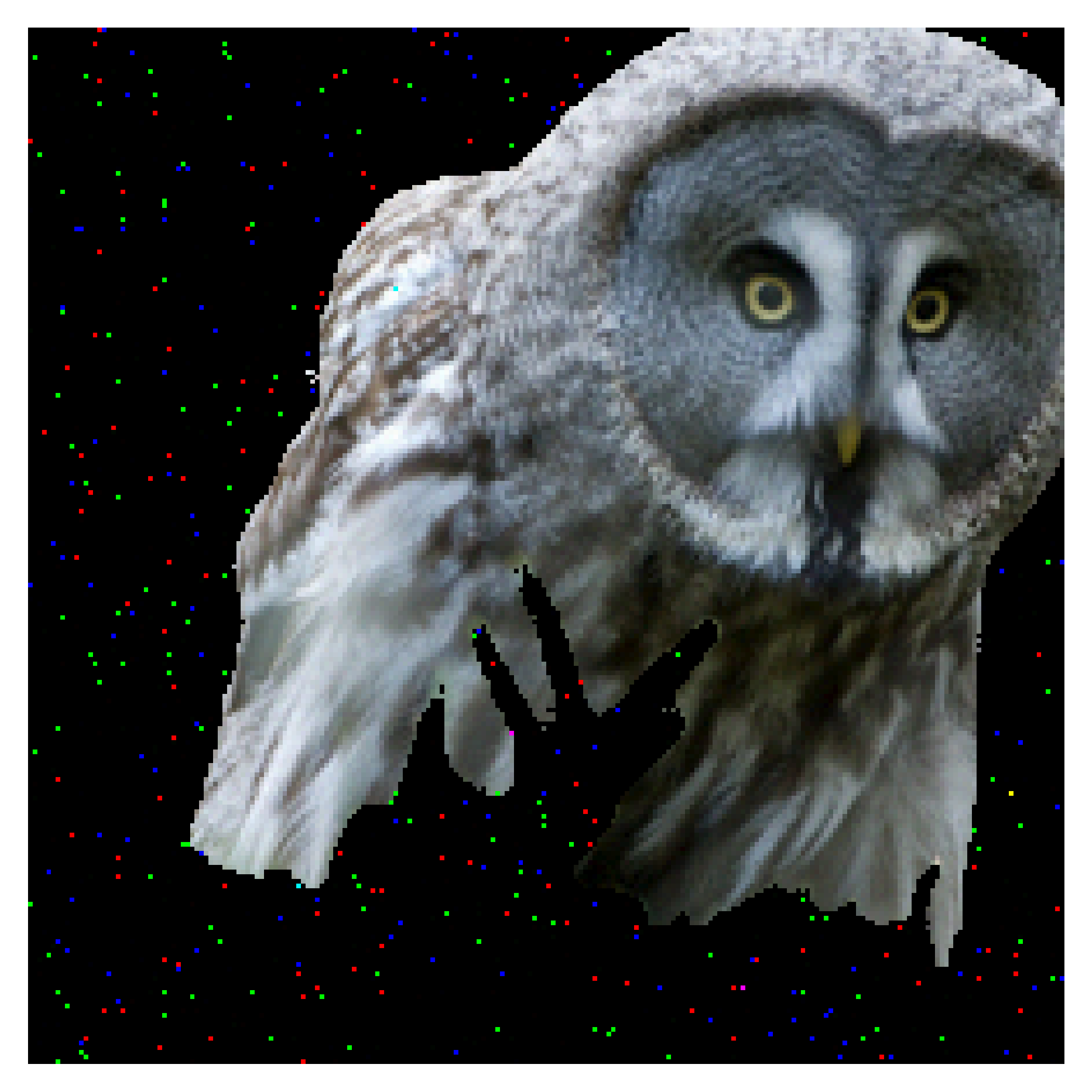}
}
\caption{Reconstructed samples generated by the proposed attack under different gradient protections: (a) the original sample; (b) masked samples; (c) gradients protected by clipping and perturbation; (d) no gradient protection; (e) gradients protected by clipping; (f) gradients protected by perturbation.}
\label{fig-evaluation-one-different-protection}
\end{figure}

We also find that protecting gradients by clipping or perturbation alone has little effect on the proposed attack. Gradients are clipped by $\nabla_{\omega} L = \nabla_\omega L / \max \left(1 , \frac{\| \nabla_\omega L \|}{C} \right)$, and samples are reconstructed by $x= \nabla_{w_i} L \oslash \nabla_{b_i} L$ through gradients of unit $i$ in the separation layer of the inference structure. Assume that $\max \left(1 , \frac{\| \nabla_\omega L \|}{C} \right) = \Delta$, after gradients clipping, we have $\left(\frac{\nabla_{w_i} L}{\Delta}\right) \oslash \left(\frac{\nabla_{b_i} L}{\Delta}\right) = \nabla_{w_i} L \oslash \nabla_{b_i} L = x.$ Clipping gradients without perturbation cannot prevent the proposed attack from reconstructing victims' samples. Adding noise to gradients without clipping also does little to defend against the proposed attack because the norm of gradients in the inference structure is significantly larger than the noise. Figure~\ref{fig-evaluation-norm-comparison} compares the average $\ell2$ norms and the absolute values of weight layer gradients under different conditions when the batch size is 16. The norm of the gradients in the inference structure is greater than the norm of the gradients generated by the regular model, and gradients in the proposed attack should be clipped in most cases. Figure~\ref{fig-evaluation-abs-value} compares the absolute values of weight layer gradients with and without clipping. The perturbation is too small relative to the gradients without clipping, which has little effect on defending the proposed attack.

\begin{figure}
\centering
\subfigure[$\ell2$ norm of gradients]{
    \centering
    \includegraphics*[width=1.63in]{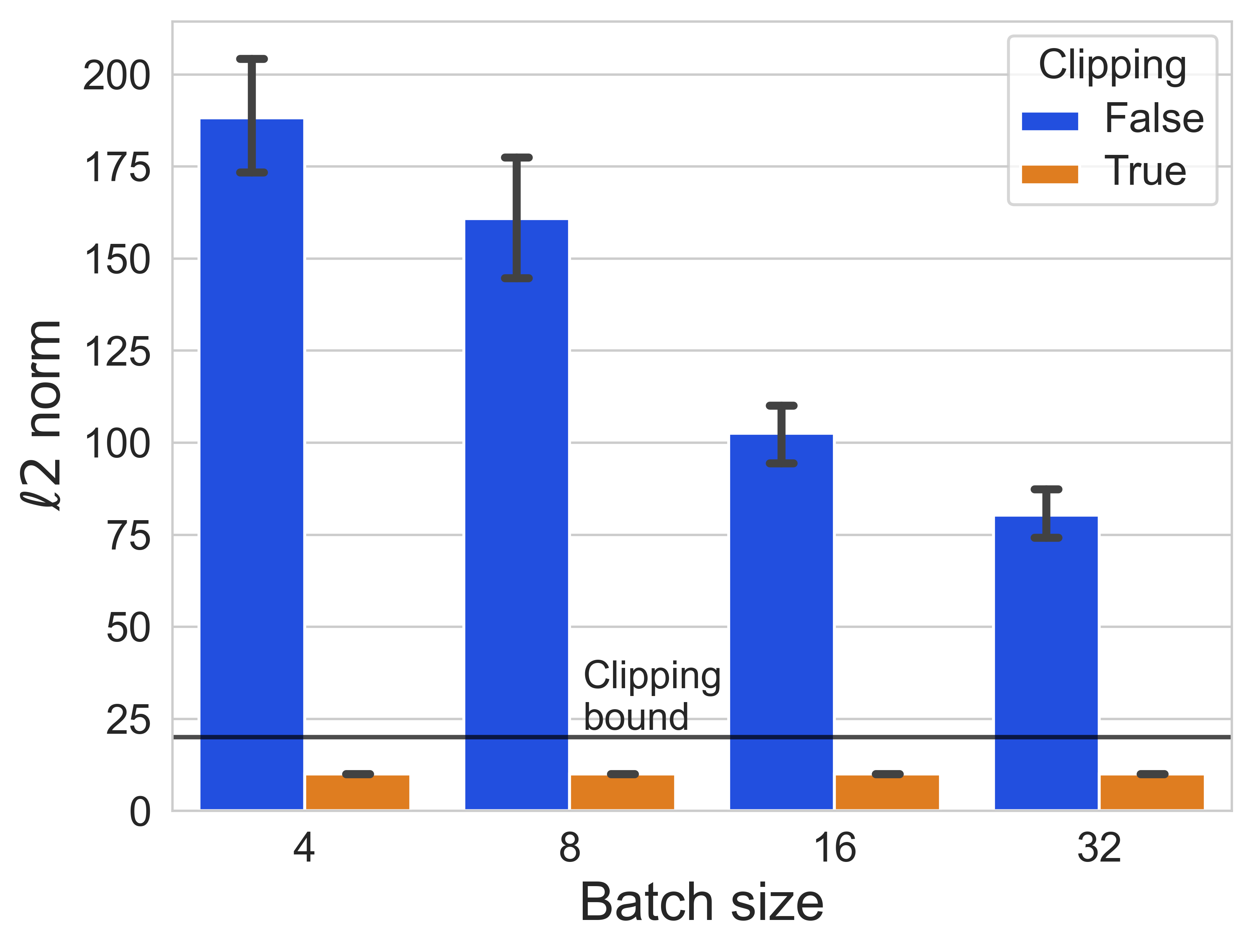}
}\hspace{-0.1in}
\subfigure[Absolute value of gradients]{
    \centering
    \includegraphics*[width=1.63in]{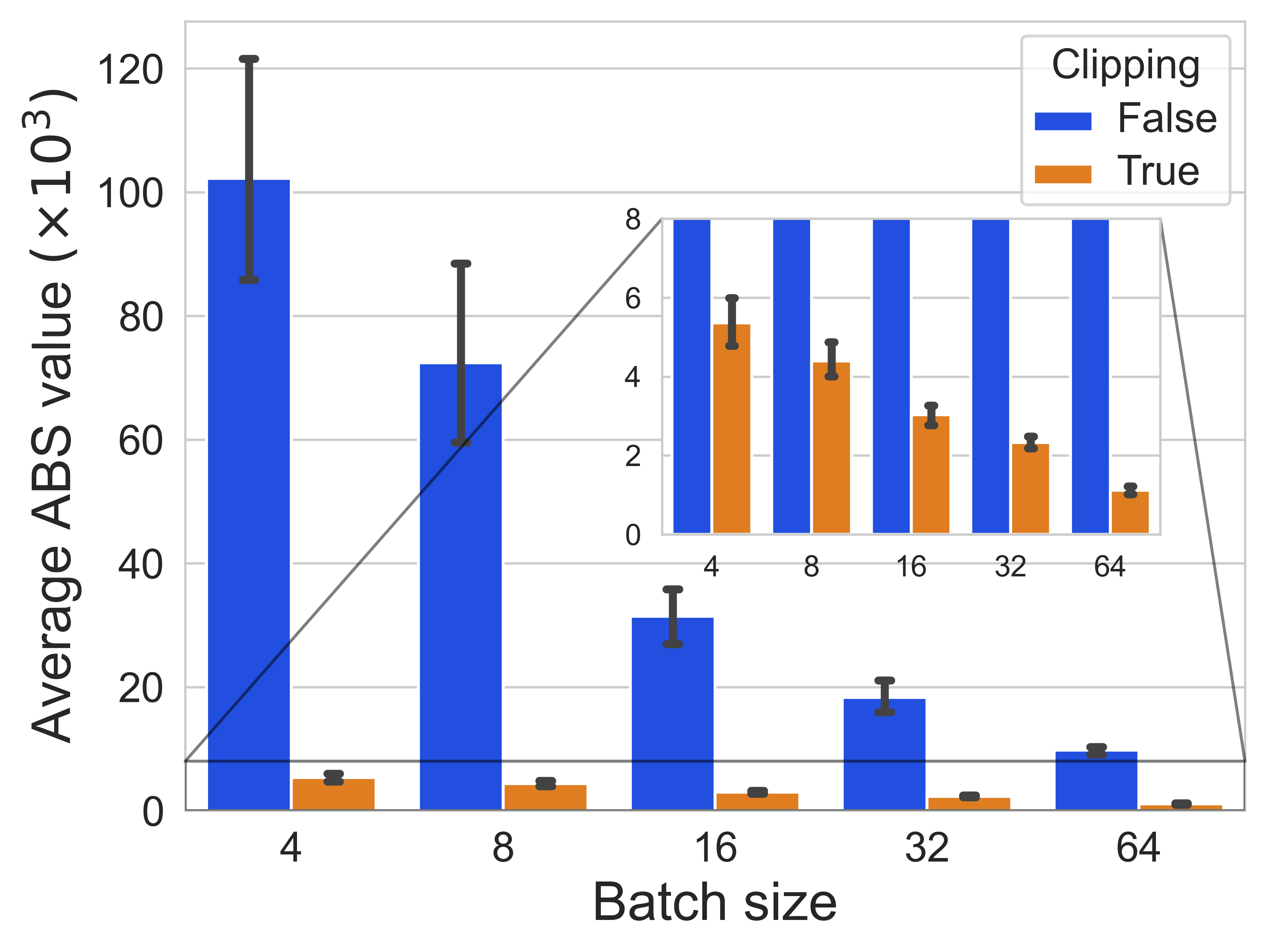}
    \label{fig-evaluation-abs-value}
}
\caption{The $\ell2$ norms and absolute values of gradients under different conditions.}
\label{fig-evaluation-norm-comparison}
\end{figure}

\subsection{Performance Factors}
\label{section-evaluation-factors}

We discuss and analyze several factors that affect the performance of the proposed attack.
First, we focus on the impact of user training models with different batch sizes on the quality of reconstructed samples. Figure~\ref{fig-evaluation-performance-batch-size} provides PSNR and CW-SSIM of reconstructed samples under different batch sizes where the number of the weight layer units in the inference structure is 2048. A larger batch size means that users' gradients contain more sample information. However, gradients are clipped by the clipping bound so that the $\ell2$ norm of gradients does exceed the clipping bound. Therefore, a larger batch size reduces each sample's information in users' gradients, increasing the difficulty of reconstructing samples and reducing the quality. The evaluation shows that the proposed attack reconstructs samples with high quality when the batch size is small, a typical result of existing reconstruction attacks. Choosing a larger batch size in FL is a straightforward and effective way to defend against reconstruction attacks.

\begin{figure}
\centering
\subfigure[PSNR]{
    \centering
    \includegraphics*[width=1.63in]{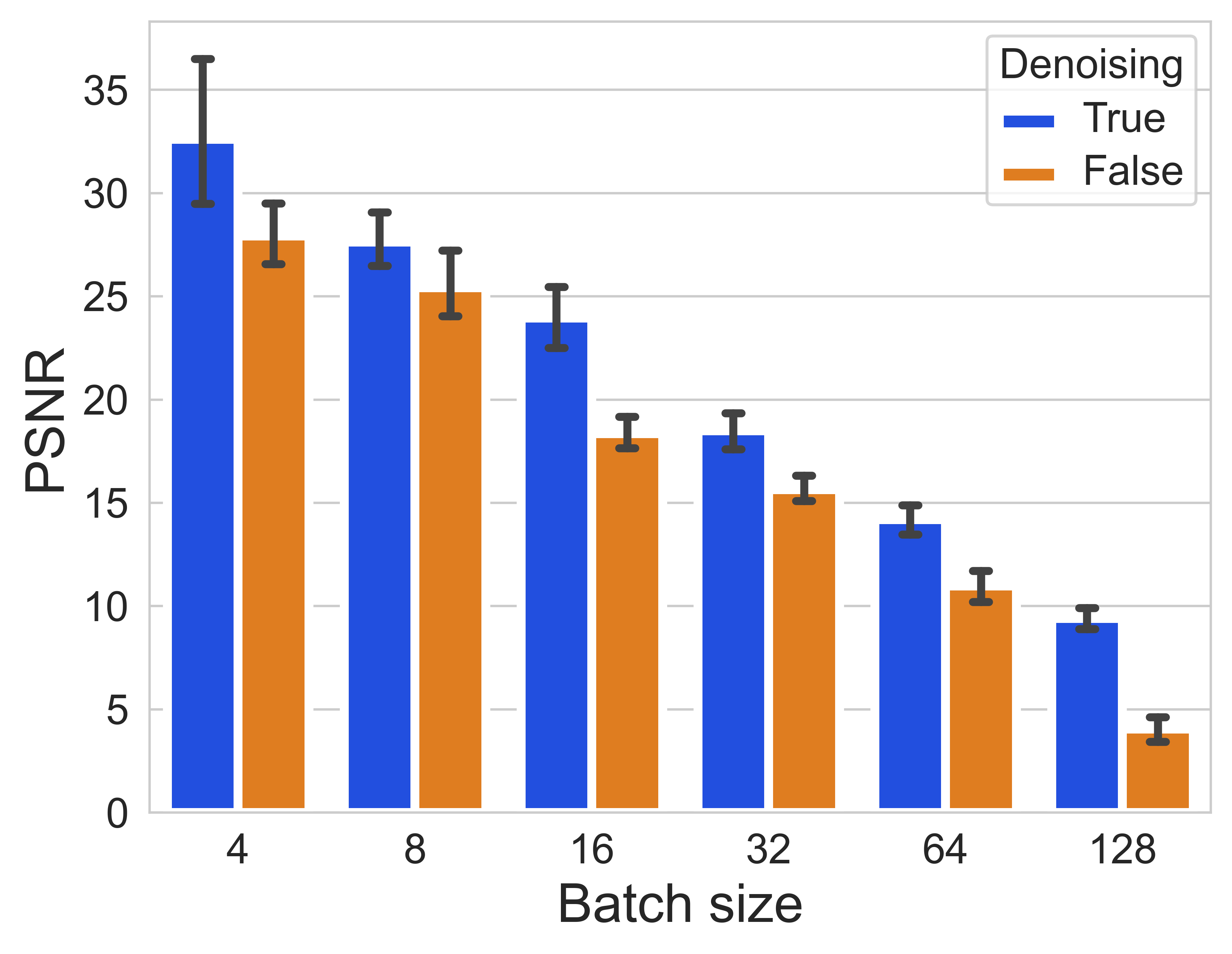}
}\hspace{-0.1in}
\subfigure[CW-SSIM]{
    \centering
    \includegraphics*[width=1.63in]{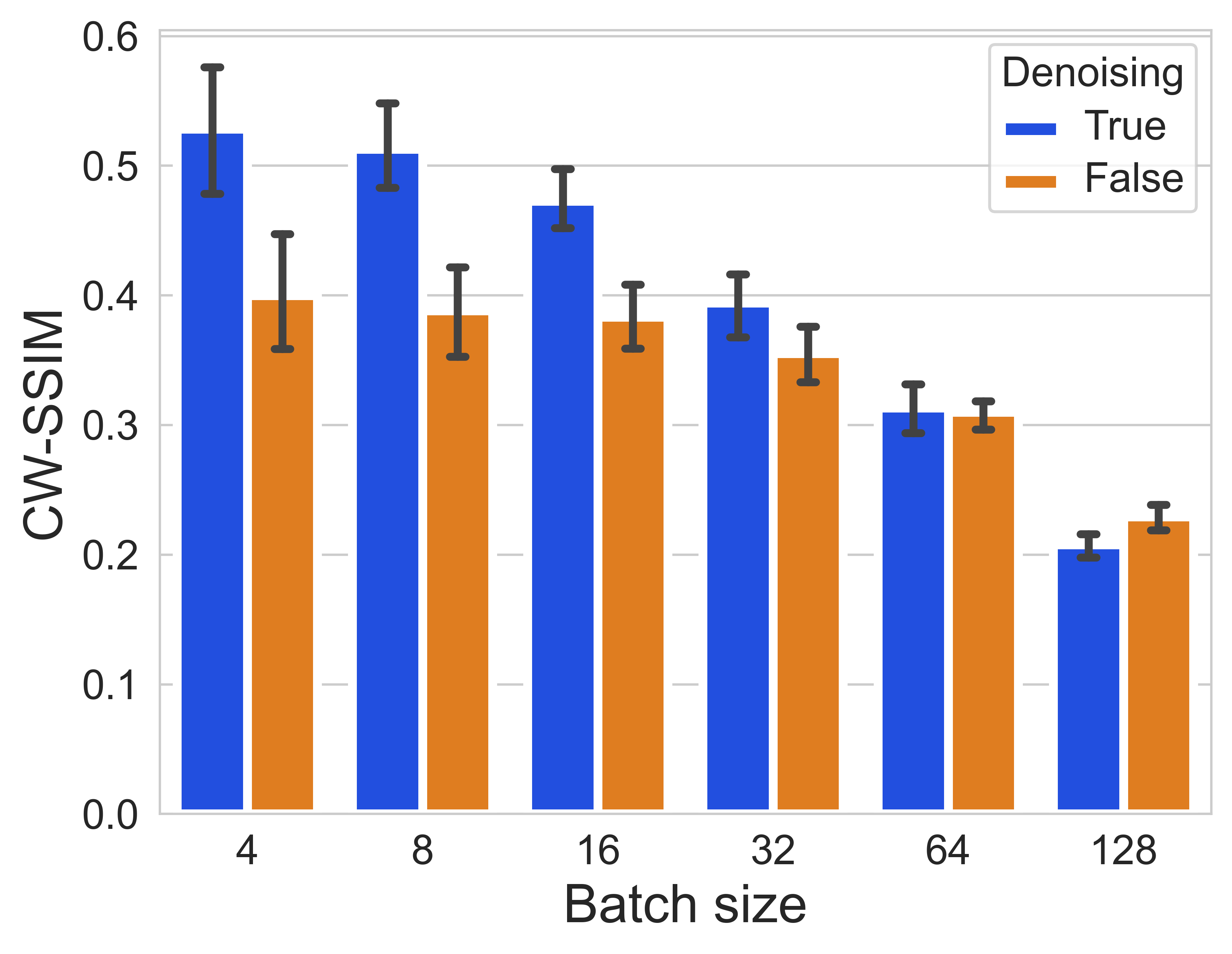}
}
\caption{Impact of batch size on reconstructed image quality.}
\label{fig-evaluation-performance-batch-size}
\end{figure}

We next discuss the impact of the number of the weight layer units in the inference structure on the proposed attack. Figure~\ref{fig-evaluation-performance-bin-num} shows that the number of weight layer units has little impact on the quality of reconstructed images but a significant impact on gradient separation. The separation ratio refers to the proportion of separated reconstructed images, rather than overlapped images (such as Fig.~\ref{fig-attack-example-multiples-sum}), in a batch of training samples. Increasing the number of units is an effective way to avoid overlapped images without considering the gradient expansion. Figure~\ref{fig-evaluation-performance-bin-num-norm} provides the $\ell2$ norms and absolute values of the weight layer gradients under different numbers of units. As given in Theorem~\ref{theorem-two-performance-related}, the $\ell2$ norm of gradients does not increase with the number of units because the number of units with non-zero gradients is not greater than the batch size. The number of batch sizes has the most significant impact on the absolute value of gradients. Given a clipping bound, a larger batch size means more unit has non-zero gradients, which reduces the absolute values of gradients in each unit.

\begin{figure}
\centering
\subfigure[Batch size = 16]{
    \centering
    \includegraphics*[width=1.63in]{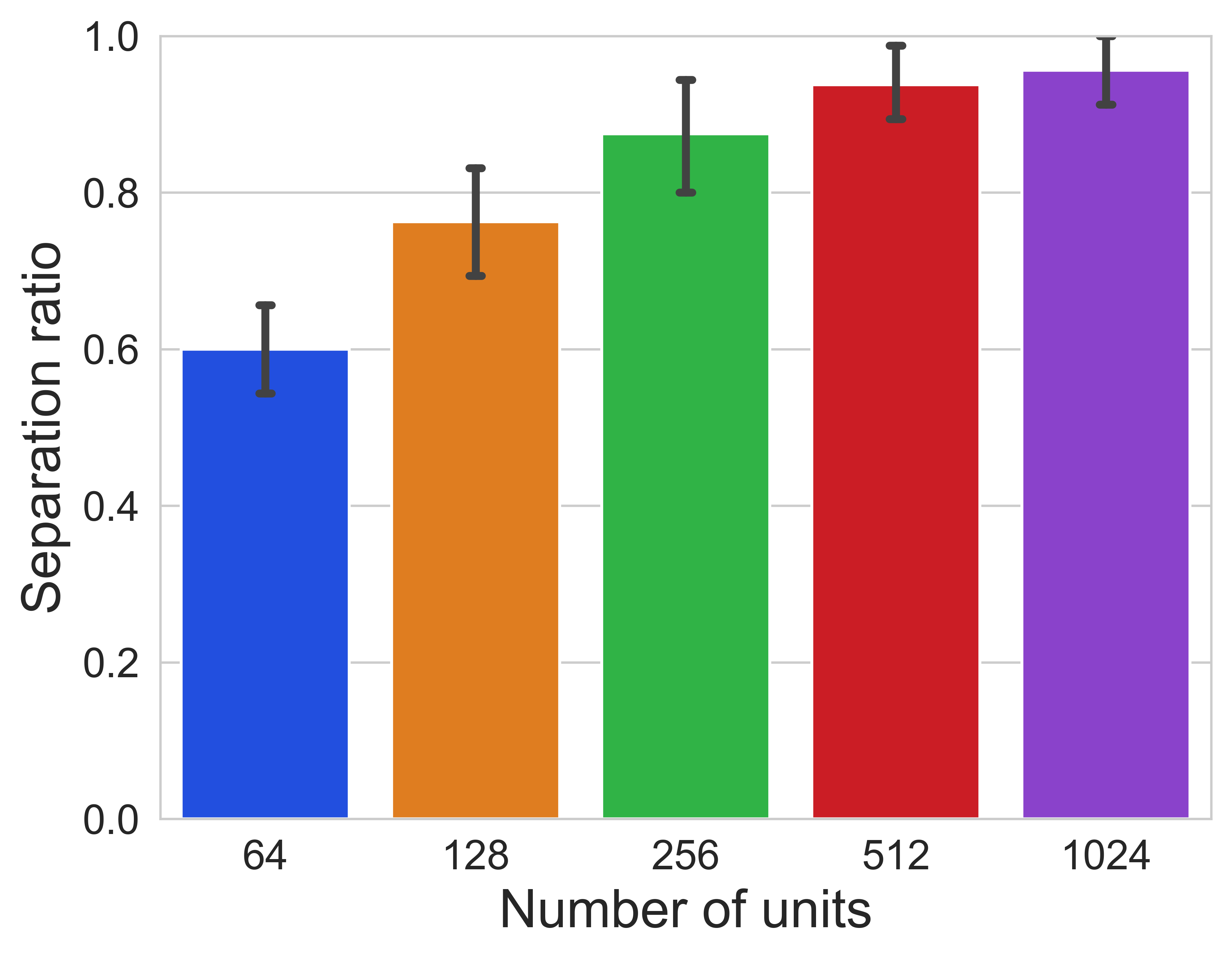}
}\hspace{-0.1in}
\subfigure[Batch size = 64]{
    \centering
    \includegraphics*[width=1.63in]{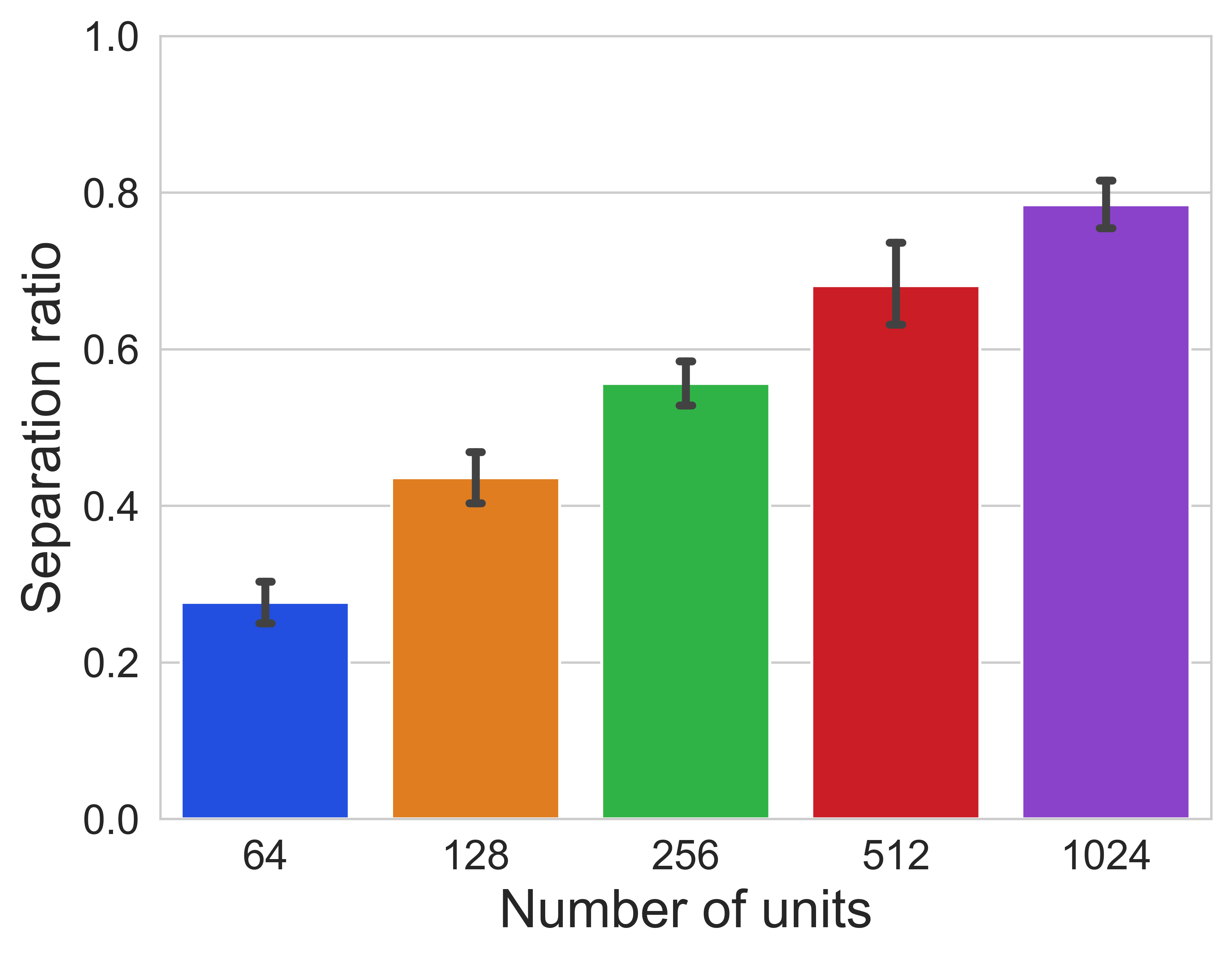}
}
\caption{The ratio of separated reconstructed samples in a batch of training samples.}
\label{fig-evaluation-performance-bin-num}
\end{figure}

\begin{figure}
\centering
\subfigure[$\ell2$ norm of gradients]{
    \centering
    \includegraphics*[width=1.63in]{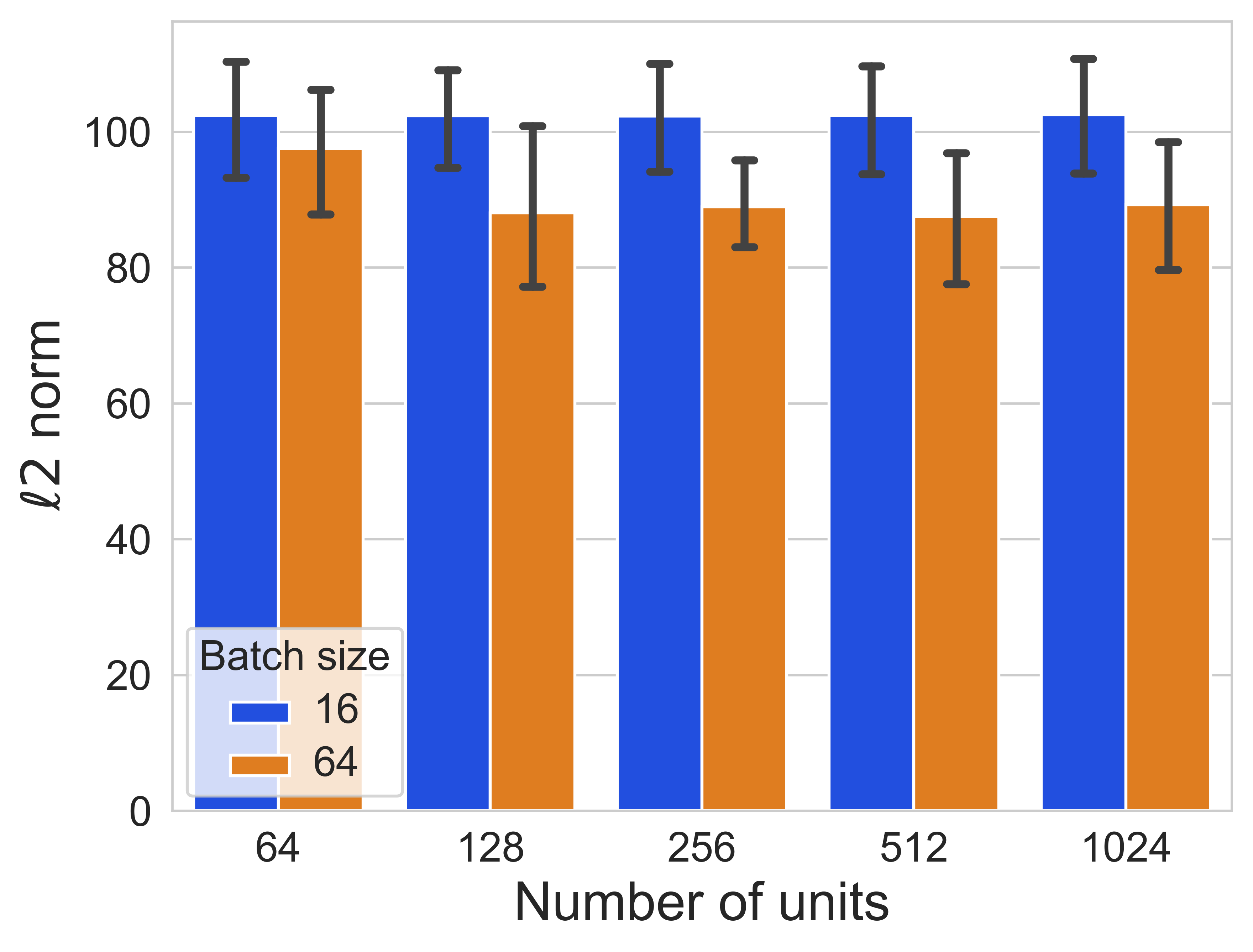}
}\hspace{-0.1in}
\subfigure[Absolute value of gradients]{
    \centering
    \includegraphics*[width=1.63in]{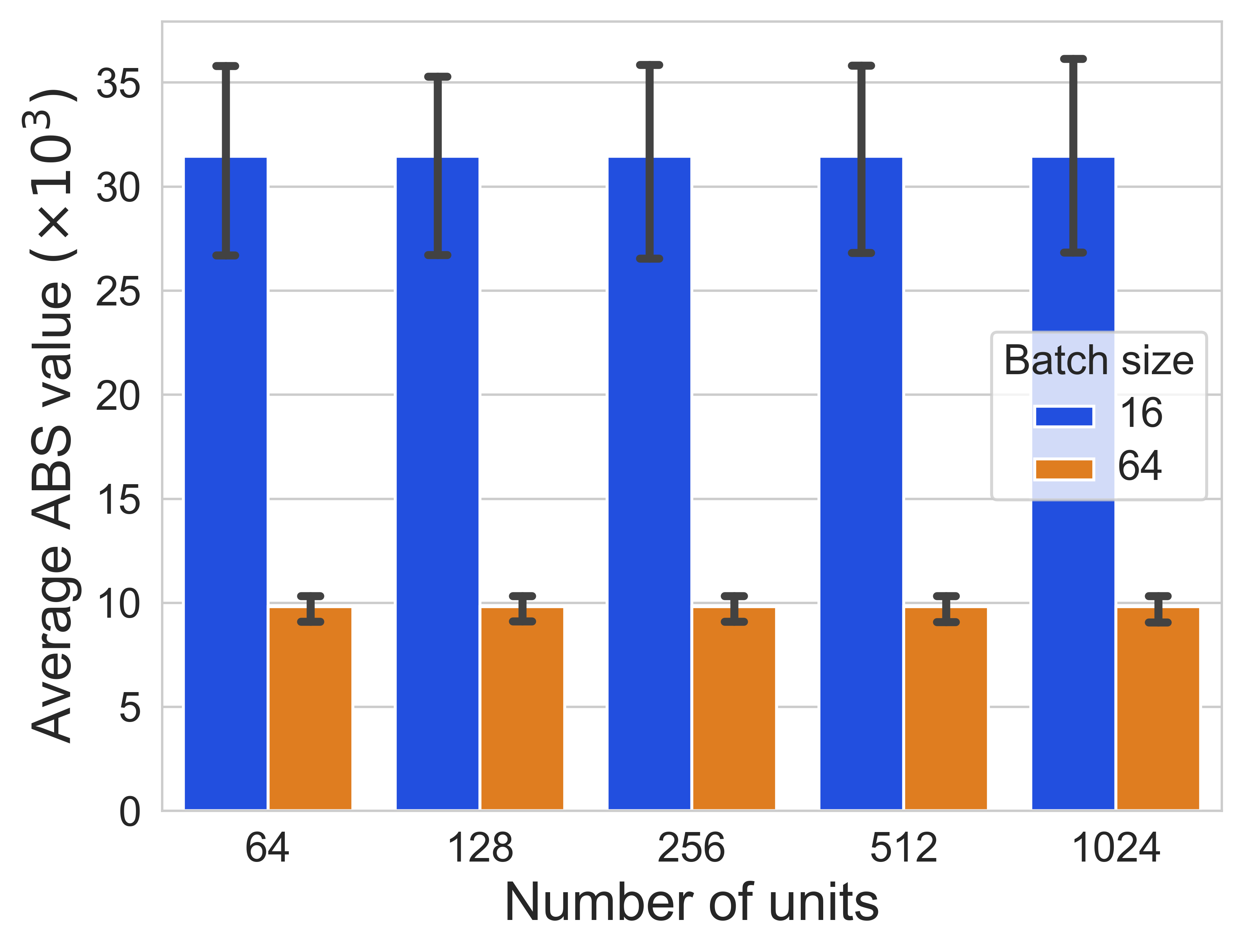}
}
\caption{The $\ell2$ norms and absolute values of gradients without clipping under different numbers of the weight layer units in the inference structure.}
\label{fig-evaluation-performance-bin-num-norm}
\end{figure}

\subsection{Privacy Parameter Setting}
\label{section-evaluation-privacy-setting}

We discuss the impact of LDP privacy parameters on the performance of the proposed attack. The most direct impact on performance is the privacy parameter $\varepsilon$. As shown in Section~\ref{section-evaluation-setup}, the scale of the Gaussian noise for gradient perturbation depends on $\varepsilon$. A larger $\varepsilon$ means a smaller noise scale and less perturbation noise in general. Thus, as given in Fig.~\ref{fig-evaluation-performance-epsilon}, a smaller $\varepsilon$ provides better protection for gradients and reduces reconstructed sample quality.

\begin{figure}
\centering
\subfigure[PSNR]{
    \centering
    \includegraphics*[width=1.63in]{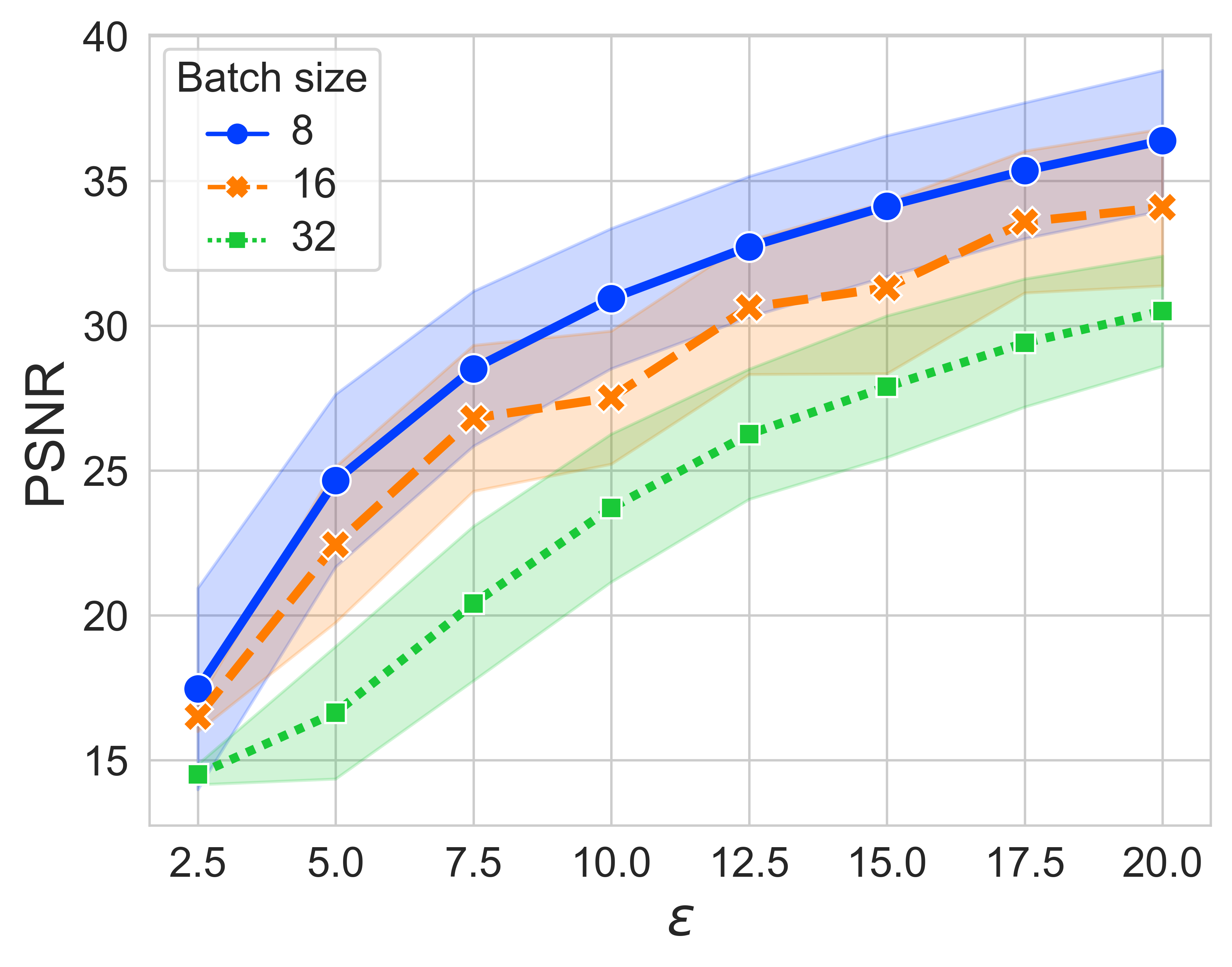}
}\hspace{-0.1in}
\subfigure[CW-SSIM]{
    \centering
    \includegraphics*[width=1.63in]{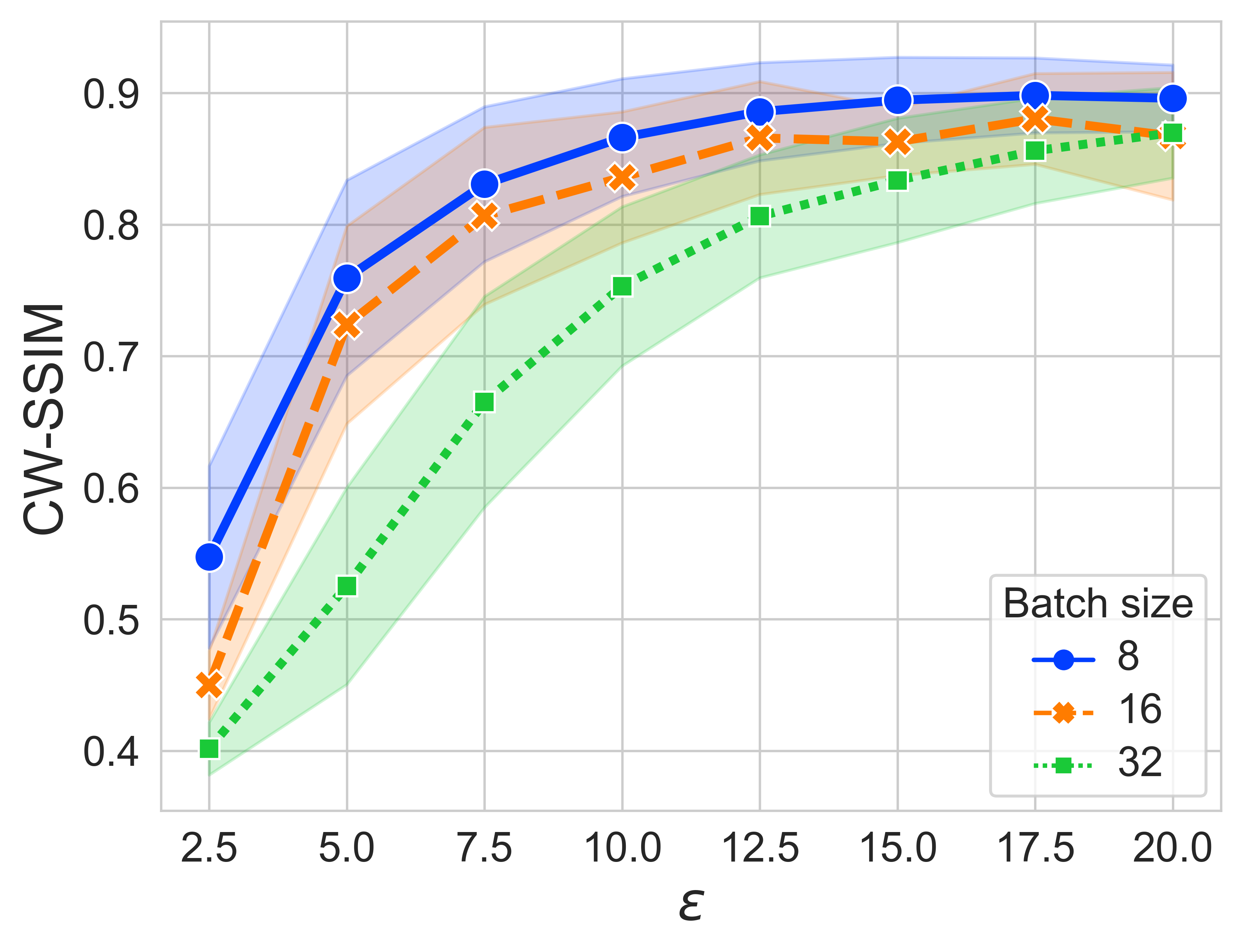}
}
\caption{The quality of reconstructed samples under various $\varepsilon$.}
\label{fig-evaluation-performance-epsilon}
\end{figure}

\begin{figure}
\centering
\subfigure[PSNR]{
    \centering
    \includegraphics*[width=1.63in]{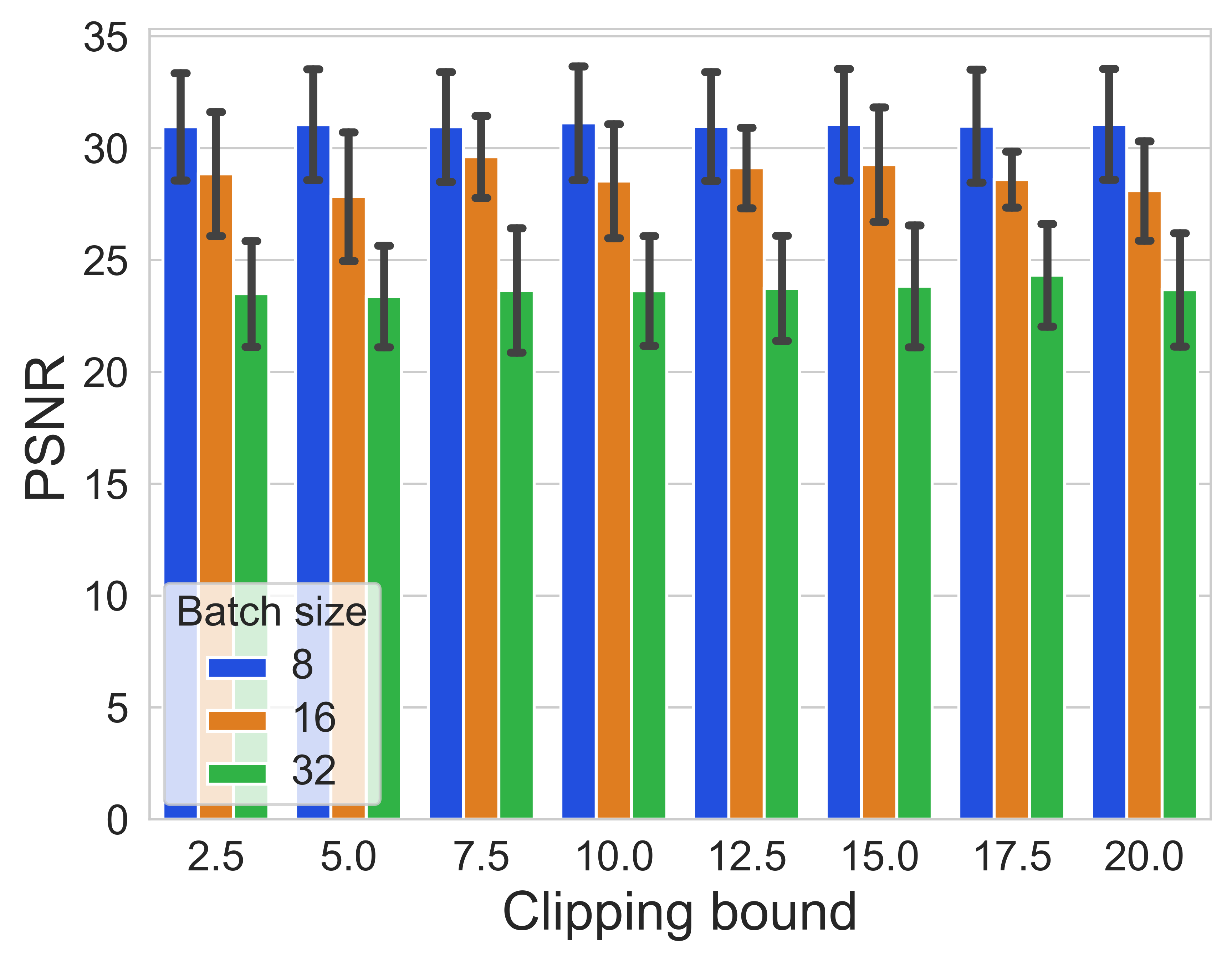}
}\hspace{-0.1in}
\subfigure[CW-SSIM]{
    \centering
    \includegraphics*[width=1.63in]{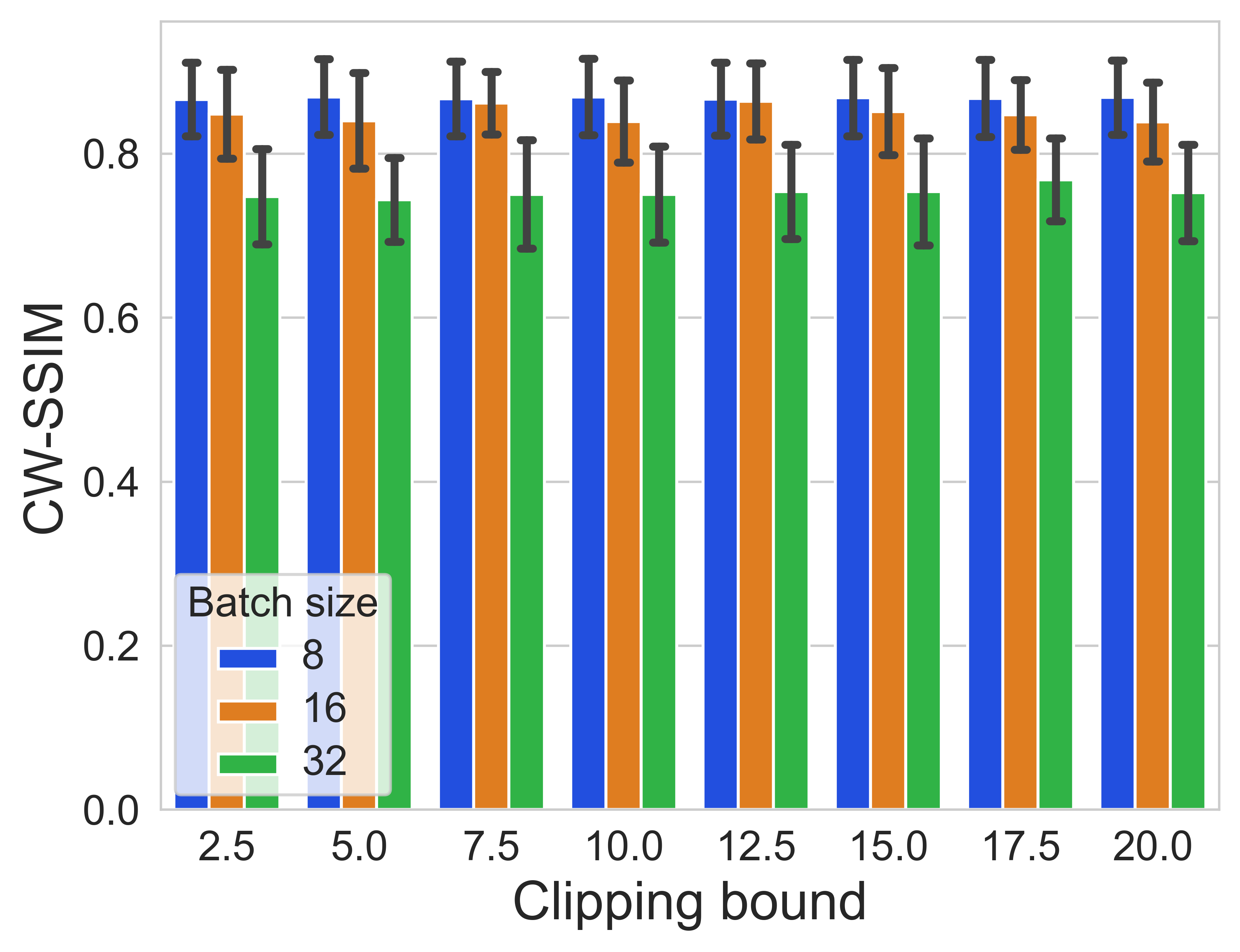}
}
\caption{The quality of reconstructed samples under various clipping bounds.}
\label{fig-evaluation-performance-clipping-bound}
\end{figure}

We next discuss the impact of clipping bounds on the proposed attack with Fig.~\ref{fig-evaluation-performance-clipping-bound}. Clipping bounds represent the compression degree of users' gradients. However, we find that a larger clipping bound does not improve the quality of reconstructed samples. The reason is that the scale of perturbation noise also depends on clipping bounds, and a larger clipping bound means more noise in users' gradients. Therefore, increasing clipping bounds has little effect on improving reconstructed quality.

\subsection{Impact of the Proposed Attack on FL Training}
\label{section-simulation-model-accuracy}

In this part, we explore the impact of the proposed attack on FL training. We first show the difference in the aggregated target model gradients between the two cases of whether or not the proposed attack is implemented. Then, we compared the accuracy of the final global model under two conditions to illustrate that the proposed attack has almost no impact on the performance of the target model.

Although the proposed attack only requires one training round of gradients to implement, we attack different victims in each round since a user's gradients in FL training have a limited impact on the final model. For example, when there are 50 users in each FL training round, and the number of training rounds is 10K, the server obtains a total of 500K users' gradients. The proposed attack on a single user only removes one gradient from 500K gradients. Therefore, we carry out the proposed attack in each training round, and the number of victims in the above example is 10K.

\begin{table}
\caption{Difference proportion between gradients with and without implementing the proposed attack under various users.}
\centering
\label{table-attack-impact}
\begin{tblr}{columns = {c, m}, column{2-Z}={0.8cm}, hlines={1pt}, hline{1,Z}={2pt}, vline{2}={1pt}}
\diagbox{Range}{Users} & 100 & 50 & 10 & 5 & 2\\
$(-\infty, 10^{-6})$ & 36\% & 36\% & 32\% & 19\% & 17\%\\
$[10^{-6}, 10^{-5})$ & 23\% & 22\% & 17\% & 15\% & 7\%\\
$[10^{-5}, 10^{-4})$ & 24\% & 23\% & 24\% & 23\% & 17\%\\
$[10^{-4}, 10^{-3})$ & 16\% & 18\% & 22\% & 25\% & 23\%\\
$[10^{-3}, 10^{-2})$ & 1\% & 1\% & 5\% & 18\% & 26\%\\
$[10^{-2}, +\infty)$ & 0\% & 0\% & 0\% & 0\% & 10\%\\
\end{tblr}
\end{table}

Table~\ref{table-attack-impact} compares the difference between aggregated gradients of the target model with and without implementing the proposed attack under various training users in each round. In the case of extremely few users (number of users is 2), whether or not the attack is implemented significantly impacts the aggregation gradient. In this case, the proposed attack causes aggregated gradients to be determined by only one user (the only non-target user); otherwise, the aggregate gradients are the average of two users. However, as the number of users increases, the gradient difference caused by the proposed attack becomes smaller. The number of training users can easily exceed 100 in FL, and the proposed attack has almost no effect on the aggregated gradient.

The same phenomenon also appears in the model accuracy. 
We set the number of training users to 50 and calculate the top-1 accuracy (Acc@1) and top-5 accuracy (Acc@5) of various convergent target models (\cite{liu2022convnet,huang2018densely,tan2020efficientnet,szegedy2014going}) on the ImageNet~\cite{ILSVRC15}, while the results are summarized in Table $\mathrm{IX}$ of the Appendix A. 
% We set the number of training users to 50 and calculate the top-1 accuracy (Acc@1) and top-5 accuracy (Acc@5) of various convergent target models on the ImageNet~\cite{ILSVRC15}, while the results are summarized in Table~\ref{table:model-accuracy} of the Appendix~\ref{appendix-different-gradients-protection}. 
Training users' samples are random samples from the ImageNet training set. Even without perturbing gradients, the proposed inference structure has little impact on the performance of the target model when the number of training rounds is sufficient. When considering the noise in the gradient, the performance difference is more challenging to detect since LDP causes a certain degree of performance loss.

\section{Discussion}
\subsection{Attack Consumption}

The proposed attack only requires users to upload the gradient once, and samples can be reconstructed through gradients in a single training round. In other words, users do not repeatedly upload multiple rounds of gradients, which increases the stealth of the proposed attack. Although training users are randomly selected in each training round, the proposed attack can reconstruct victims' samples through one-round gradients. In addition, experiments in Section~\ref{section-simulation-model-accuracy} show that the proposed attack has little effect on the training and accuracy of the model when there is only one victim per training round. The size of gradients is related to model complexity. For a global model with parameters $\omega$, the size of local gradients is also $\left|\omega\right|$. Since the victim only uploads one round of gradients in the proposed attack, the communication consumption of the proposed attack is $\mathcal{O}(\left|\omega\right|)$.

We discuss the computational cost of the proposed attack by providing the time complexity of each step in Section~\ref{section-all-in-one-implementation}. Without loss of generality, we consider the batch size to be $B$ and the size of samples to be $S$. The number of units in weight, bias, and metrics layers in the inference structure are $N_w$, $N_b$, and $N_m$, respectively. Step~(1) infers the privacy parameter from the mean and variance of weight gradients in the separation layer, and the complexity is $\mathcal{O}(N_w S)$. Steps~(2) and~(3) calculate the bias gradient by mean and perform raw reconstruction by element-wised division, which are of complexity $\mathcal{O}(N_w N_b)$ and $\mathcal{O}(N_w S)$, respectively. Step~(4) classifies the metrics according to their position, which is an in-place operation. Step~(5) averages and reorders the metrics, and its complexity is $\mathcal{O}(n_m B) + \mathcal{O}(B \log B)$. Step~(6) selects valid results by reconstructed reverse units, and its complexity is $\mathcal{O}(B)$. Step~(7) is an optimization problem involving  $BS$ variables. The essence of step~(8) is the comparison and assignment operation, and its complexity is $\mathcal{O}(BS)$.

\begin{table}
\caption{Running time (s) of the proposed attack under various batch sizes ($B$) and optimization rounds ($R$).}
\centering
\label{table-time}
\begin{tblr}{column{1} = {c, m}, column{2-6}={r,m}, row{1}={c,m},,column{2-Z}={1cm}, hlines={1pt}, hline{1,Z}={2pt}, vline{2}={1pt}}
\diagbox{$R$}{$B$} & 4 & 8 & 16 & 32 & 64\\
$0$ & 0.032 & 0.034 & 0.036 & 0.037 & 0.039\\
$500$ & 1.049 & 1.101 & 1.157 & 1.389 & 1.483\\
$1000$ & 1.871 & 1.994 &  2.038 & 2.552 & 2.767\\
$2000$ & 3.558 & 4.298 &  4.299 & 4.641 & 5.645\\
$4000$ & 7.289 & 8.433 &  9.075 & 9.330 & 11.275\\
\end{tblr}
\end{table}

\begin{table}
\caption{Power consumption (W) of the proposed attack under various batch sizes ($B$) and optimization rounds ($R$).}
\centering
\label{table-power}
\begin{tblr}{columns = {c, m}, column{2-6}={r,m}, row{1}={c,m}, column{2-Z}={1cm}, hlines={1pt}, hline{1,Z}={2pt}, vline{2}={1pt}}
\diagbox{$R$}{$B$} & 4 & 8 & 16 & 32 & 64\\
$0$ & 58.961 & 60.706 & 62.283 & 64.763 & 73.595\\
$500$ & 67.617 & 72.269 & 85.212 & 110.673 & 164.893\\
$1000$ & 68.403 & 75.083 &  90.305 & 123.168 & 195.431\\
$2000$ & 69.678 & 82.241 &  91.052 & 139.835 & 207.097\\
$4000$ & 73.760 & 84.106 &  94.220 & 140.651 & 214.178\\
\end{tblr}
\end{table}

TABLE~\ref{table-time} provides the running time for performing a single proposed attack through the victim's clipped and perturbed gradients. The computational consumption of the proposed attack is concentrated in metric-based sample optimization. Batch sizes also significantly impact running time since the increase in training samples increases the number of optimization variables and expands the solution space. Besides, Table~\ref{table-power} records the average powers of the GPU when performing the proposed attack. Logging of power consumption data is implemented by PyTorch.\footnote{\url{https://pytorch.org/docs/stable/generated/torch.cuda.power_draw.html}.} The number of optimization rounds has little effect on the average power, as the video memory consumption tends to stabilize during the optimization. The increase in optimization variables means more memory consumption. As a result, batch sizes also significantly impact the average power.

\subsection{Possible Defense and attack limitations}

Possible defenses against the proposed attacks include using dynamic clipping bounds and privacy parameters in local training, model malicious structure detection for users, and cryptography-based secure gradient aggregation (e.g., \cite{Nguyen2024Preserving,Song2023lsecnet}). In addition, we will consider attacks on model updates, design better optimization objectives, and expand to more data types in future work.
% Due to the length limitation of the paper, we discuss possible defenses and limitations of the proposed attack in more detail in Appendix~\ref{appendix-discussion}.
Due to the length limitation of the paper, we discuss possible defenses and limitations of the proposed attack in more detail in Appendix B.

\section{Conclusion}

This paper proposes a sample reconstruction attack against FL mechanisms with LDP, in which gradients are clipped and noisy. We briefly analyze the reason for the gradient expansion and the failure of the existing separation methods. Based on the above analysis, we design the proposed attack from two aspects: gradient separation without expansion and sample quality improvement against FL with LDP. We present a separation layer such that the gradients of each sample only exist in its reverse unit, which effectively separates gradients in FL with LDP without causing gradient expansion. Besides, the subjects of samples selected by SAM can further compress victims' gradients. For sample quality improvement, we infer the confidence interval of the noise by the artificially added zero gradients and filter the noise in the background of reconstructed samples. In addition, a metric-based optimization is proposed to improve the sample quality further. Theory and evaluations show that the proposed attack is the only reconstruction attack that effectively reconstructs victims' samples when gradients are clipped and noisy. Simulation results show that the proposed attack has little impact on FL training and model accuracy. Finally, we provide possible defenses and discuss the limitations and future works.

\bibliographystyle{ieeetr}
\bibliography{bibliography}

\begin{IEEEbiography}[{\includegraphics[width=0.9in,height=1.1in,clip,keepaspectratio]{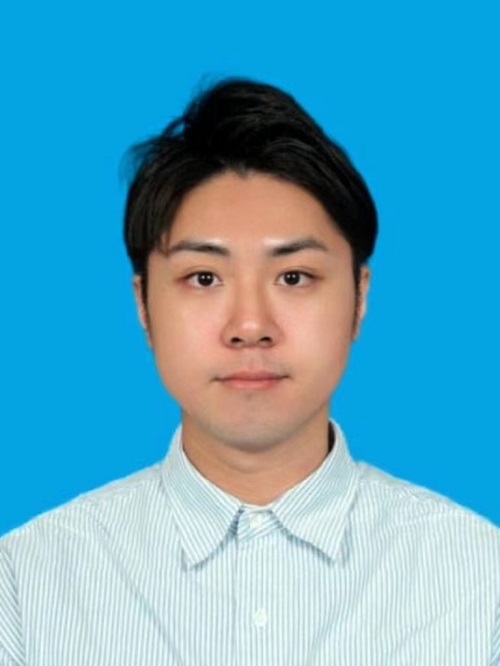}}]{Zhichao You} received his B.E.\ degree in Information and Computing Science from South China Agricultural University, Guangzhou, China in 2018, the M.S.\ degree in Computer Science and Technology in Xidian University, Xi'an, China in 2021. He is now pursuing his Ph.D.\ degree in computer science and technology at Xidian University , Xi'an, China. His research interests include privacy-preserving federated learning and wireless network security.
\end{IEEEbiography}

\begin{IEEEbiography}[{\includegraphics[width=0.9in,height=1.1in,clip,keepaspectratio]{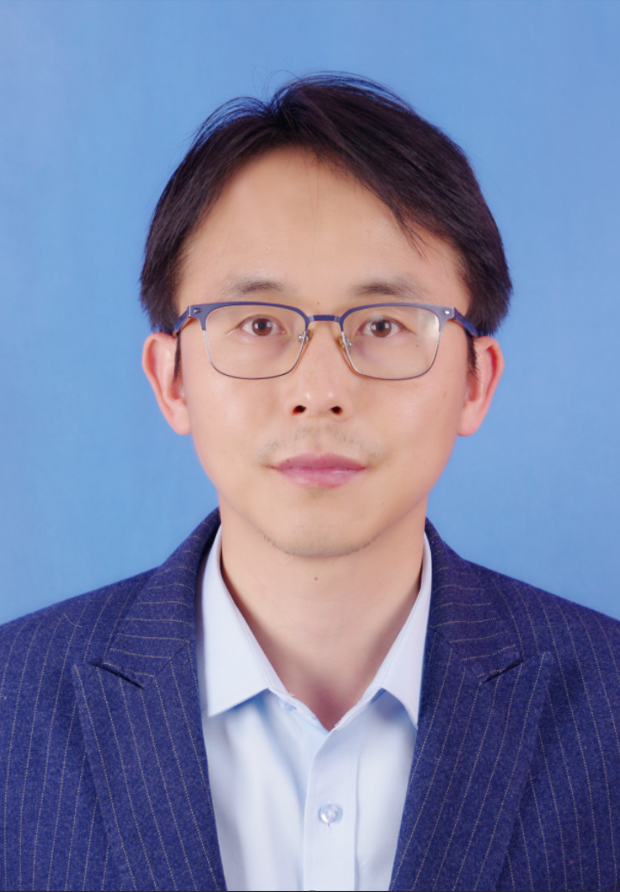}}]{Xuewen Dong} (Member, IEEE) received the B.E., M.S., and Ph.D.\ degrees in Computer Science and Technology from Xidian University, Xi'an, China, in 2003, 2006, and 2011, respectively. From 2016 to 2017, he was with Oklahoma State University, OK, USA, as a visiting scholar. Now, he is an associate professor in the School of Computer Science in Xidian University. His research interests include cognitive radio networks, wireless network security, and blockchain.
\end{IEEEbiography}

\begin{IEEEbiography}[{\includegraphics[width=0.9in,height=1.1in,clip,keepaspectratio]{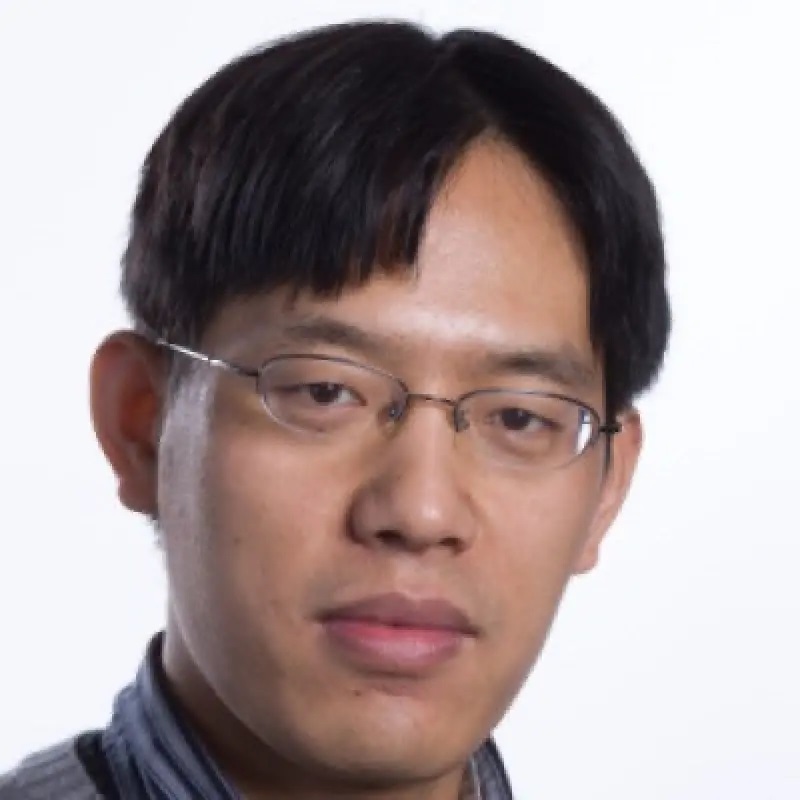}}]{Shujun Li} received the BE degree in Information Science and Engineering, and the PhD degree in Information and Communication Engineering from Xi'an Jiaotong University, China, in 1997 and 2003, respectively. Currently, he is Professor of Cyber Security and directing the Institute of Cyber Security for Society (iCSS) at the University of Kent, UK. His current research interests mainly focus on interplays between several interdisciplinary research areas, including cyber security and privacy, cybercrime and digital forensics, human factors, multimedia computing, AI and NLP, and more recently education. He is a Fellow of the BCS, The Chartered Institute for IT and a Vice President of the Association of British Chinese Professors (ABCP).
\end{IEEEbiography}

\begin{IEEEbiography}[{\includegraphics[width=0.9in,height=1.1in,clip,keepaspectratio]{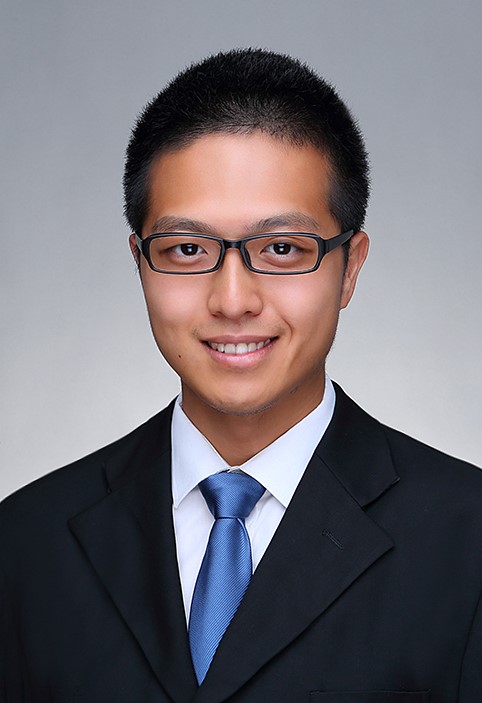}}]{Ximeng Liu} (Senior Member, IEEE) received the B.Sc. degree in electronic engineering and the Ph.D.\ degree in cryptography from Xidian University, Xi'an, China, in 2010 and 2015, respectively. He is currently a Full Professor with the College	of Mathematics and Computer Science, Fuzhou	University, China. He was also a Research Fellow with the School of Information System, Singapore Management University, Singapore. His research interests include cloud security, applied cryptography, and big data security.
\end{IEEEbiography}

\begin{IEEEbiography}[{\includegraphics[width=0.9in,height=1.1in,clip,keepaspectratio]{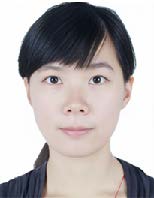}}]{Siqi Ma} (Member, IEEE) received the B.S.\ degree in computer science from Xidian University, Xi'an, China, in 2013, and the Ph.D.\ degree in information	systems from Singapore Management University in	2018. She was a Research Fellow with the Distinguished System Security Group, CSIRO. She is currently a Senior Lecturer with the University of New South Wales, Canberra Campus, Australia. Her	current research interests include data security, the IoT security, and software security.
\end{IEEEbiography}

\begin{IEEEbiography}[{\includegraphics[width=0.9in,height=1.1in,clip,keepaspectratio]{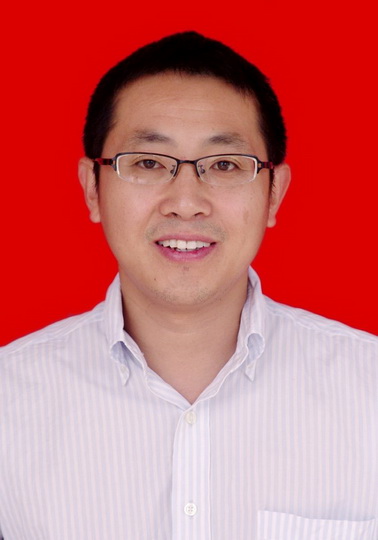}}]{Yulong Shen} (Member, IEEE) received the B.S.\ and M.S.\ degrees in computer science and the Ph.D.\ degree in cryptography from Xidian University, Xi'an, China, in 2002, 2005, and 2008, respectively. He is currently a Professor with the School of Computer Science and Technology, Xidian University, and also an Associate Director of the Shaanxi Key Laboratory of Network and System Security. His research interests include wireless network security and cloud computing security.
\end{IEEEbiography}
\clearpage

\begin{appendices}

\section{Experiment Supplement}
\label{appendix-different-gradients-protection}

\begin{table}[]
\centering
\caption{Parameter values not mentioned in simulations.}
\begin{tblr}{
hline{2-Y}={1pt}, hline{1,Z}={2pt}, columns={c,m},
}
\SetCell[r=2]{c,m} Parameters & \SetCell[c=2]{c,m} Values\\
& CIFAR100 & other datasets\\
$w_\mu$ & $1\times10^{6}$ & $1\times10^{6}$\\
$w_\sigma$ & $2\times10^{4}$ & $2\times10^{4}$\\
$w_{\mathrm{TV}}$ & $1 \times 10^{-6}$ & $1 \times 10^{-6}$\\
$\mu$ & $1.7 \times 10^{-3}$ & $3 \times 10^{-3}$\\
$s$ & $1.3 \times 10^{-3}$ & $3 \times 10^{-3}$\\
Mean factor & 1 & 1\\
Variance factor & $1\times10$ & $1\times10$\\
TV factor & $1\times10^{-3}$ & $1\times10^{-3}$\\
Separation layer weight & $2 \times 10^{-4}$ & $1\times 10^{-5}$\\
\end{tblr}
\label{table-hyper-parameters}
\end{table}

\begin{table}
\centering
\caption{Convergent target model accuracy comparison (\%).}
\begin{tblr}{hlines={1pt}, column{3,4}={c}, hline{1,Z}={2pt}, column{1}={2cm}}
& & Under attack & Normal case\\
\SetCell[r=2]{c,m} ResNet18 \newline~\cite{he2016deep} & Acc@1 & $67.28 \pm 0.02$ & $67.24 \pm 0.04$\\
& Acc@5 & $86.78 \pm 0.04$ & $86.75 \pm 0.05$\\
\SetCell[r=2]{c,m} ResNet50 \newline~\cite{he2016deep} & Acc@1 & $78.03 \pm 0.14$ & $78.16 \pm 0.16$\\
& Acc@5 & $92.94 \pm 0.06$ & $93.01 \pm 0.08$\\
\SetCell[r=2]{c,m} ConvNeXt \newline~\cite{liu2022convnet} & Acc@1 & $79.59 \pm 0.03$ & $79.43 \pm 0.29$\\
& Acc@5 & $93.80 \pm 0.03$ & $93.85 \pm 0.10$\\
\SetCell[r=2]{c,m} DenseNet121 \newline~\cite{huang2018densely} & Acc@1 & $72.41 \pm 0.08$ & $72.24 \pm 0.06$\\
& Acc@5 & $89.84 \pm 0.03$ & $89.92 \pm 0.05$\\
\SetCell[r=2]{c,m} EfficientNet \newline~\cite{tan2020efficientnet} & Acc@1 & $71.79 \pm 0.14$ & $71.34 \pm 0.46$\\
& Acc@5 & $89.49 \pm 0.08$ & $89.18 \pm 0.25$\\
\SetCell[r=2]{c,m} GoogLeNet \newline~\cite{szegedy2014going} & Acc@1 & $72.59 \pm 0.09$ & $72.73\pm 0.05$\\
& Acc@5 & $90.16 \pm 0.07$ & $90.25 \pm 0.02$\\
\end{tblr}
\label{table:model-accuracy}
\end{table}

\begin{table*}
\caption{Quality of reconstructed samples generated by different sample reconstruction attacks against gradients without LDP protection, clipped gradients, and noisy gradients, respectively.}
\centering
\SetTblrInner{colsep=0pt}
\label{table-evaluation-metrics-comparison-no-dp}
\begin{tblr}{
    columns = {0.4in, c, m},
    hline{1,Z}={2pt}, hline{2}={3,4,6,7,9,10,12,13,14}{1pt}, hline{2}={5,8,11}{rightpos=-1, 1pt},
    hline{4-Y}={dotted}, hline{3,9,15,21}={1pt},
    column{1} = {0.6in, c}, column{2} = {0.6in, c},
    colsep = {2pt},
}
\SetCell[r=2]{c,m} Protection & \SetCell[r=2]{c,m} Attack & \SetCell[c=3]{c} ImageNet & & & \SetCell[c=3]{c} CIFAR100 & & & \SetCell[c=3]{c} Caltech256 & & & \SetCell[c=3]{c} Flowers102\\
& & MSE & PSNR & SSIM & MSE & PSNR & SSIM & MSE & PSNR & SSIM & MSE & PSNR & SSIM \\
\SetCell[r=6]{c,m} Gradients without LDP protection & Proposed attack & \textbf{0} & - & \textbf{1} & \textbf{0} & - & \textbf{1} & \textbf{0} & - & \textbf{1} & \textbf{0} & - & \textbf{1}\\
& Without denoising &  0 & - & 1 & 0 & - & 1 & 0 & - & 1 & 0 & - & 1\\
& Fowl's attack~\cite{fowl2022robbing} & 0 & - & 1 & 0 & - & 1 & 0 & - & 1 & 0 & - & 1\\
& Boenisch's attack~\cite{Boenisch2021When} & 0 & - & 1 & 0 & - & 1 & 0 & - & 1 & 0 & - & 1\\
& Yin's attack~\cite{Yin2021see} & 0.068 & 12.043 & 0.181 & 0.069 & 12.457 & 0.156 & 0.100 & 10.634 & 0.165 & 0.084 & 11.014 & 0.084\\
& Wei's attack~\cite{Wei2020Framework} & 0.207 & 7.566 & 0.247 & 0.125 & 9.641 & 0.162 & 0.171 & 8.017 & 0.236 & 0.160 & 8.236 & 0.240\\
\SetCell[r=6]{c,m} Clipping gradients & Proposed attack  & \textbf{0} & - & \textbf{1} & \textbf{0} & - & \textbf{1} & \textbf{0} & - & \textbf{1} & \textbf{0} & - & \textbf{1}\\
& Without denoising  & 0 & - & 1 & 0 & - & 1 & 0 & - & 1 & 0 & - & 1\\
& Fowl's attack~\cite{fowl2022robbing} & 0 & - & 1 & 0 & - & 1 & 0 & - & 1 & 0 & - & 1\\
& Boenisch's attack~\cite{Boenisch2021When} & 0 & - & 1 & 0 & - & 1 & 0 & - & 1 & 0 & - & 1\\
& Yin's attack~\cite{Yin2021see} & 0.068 & 12.103 & 0.149 & 0.067 & 12.671 & 0.158 & 0.099 & 10.686 & 0.122 & 0.087 & 10.883 & 0.156\\
& Wei's attack~\cite{Wei2020Framework} & 0.067 & 12.201 & 0.067 & 0.113 & 11.281 & 0.129 & 0.098 & 10.756 & 0.085 & 0.136 & 9.545 & 0.136\\
\SetCell[r=6]{c,m} Perturbing gradients & Proposed attack & \textbf{0.0001} & \textbf{47.043} & \textbf{0.757} & \textbf{0.0001} & \textbf{46.793} & \textbf{0.778} & \textbf{0.0001} & \textbf{46.842} & \textbf{0.760} & \textbf{0.0001} & \textbf{47.028} & \textbf{0.774}\\
& Without denoising & 0.0001 & 41.934 & 0.540 & 0.0001 & 41.935 & 0.541 & 0.0001 & 41.932 & 0.541 & 0.0001 & 41.934 & 0.541\\
& Fowl's work~\cite{fowl2022robbing} & 0.313 & 5.066 & 0.203 & 0.304 & 5.206 & 0.127 & 0.328 & 4.884 & 0.181 & 0.331 & 4.822 & 0.191\\
& Boenisch's work~\cite{Boenisch2021When} & 0.171 & 8.182 & 0.213 & 0.165 & 8.450 & 0.134 & 0.197 & 7.544 & 0.190 & 0.189 & 7.640 & 0.195\\
& Yin's work~\cite{Yin2021see} & 0.069 & 12.028 & 0.184 & 0.067 & 12.593 & 0.161 & 0.101 & 10.576 & 0.171 & 0.085 & 10.972 & 0.175\\
& Wei's work~\cite{Wei2020Framework} & 0.126 & 9.256 & 0.271 & 0.124 & 9.685 & 0.164 & 0.161 & 8.292 & 0.243 & 0.145 & 8.576 & 0.263\\
\SetCell[c=2]{m,c} Random Guess & & 0.155 & 8.165 & 0.214 & 0.156 & 8.213 & 0.231 & 0.176 & 7.665 & 0.187 & 0.174 & 7.659 & 0.196\\
\end{tblr}
\end{table*}

Table~\ref{table-hyper-parameters} provides parameter values not mentioned in simulations. The mean, variance, and TV factor are the weights of each output when metric layer outputs are added to the original samples. Table~\ref{table:model-accuracy} is top-1 accuracy (Acc@1) and top-5 accuracy (Acc@5) of various convergent target models on the ImageNet. The separation layer weight is the equal weight of the separation layer in Section~\ref{section-removing-redundate-gradients}. Table~\ref{table-evaluation-metrics-comparison-no-dp} provides a detailed quality comparison between samples generated by different attacks under various gradient protections.

\section{Discussion Supplement}
\label{appendix-discussion}

\subsection{Possible Defense}

\textbf{Personalized and dynamic clipping bound.}
As discussed in the evaluation, the proposed attack can eliminate the effect of clipping when weight and bias gradients in FCL are clipped by the same clipping bound. A possible defense is to clip weight and bias gradients with different clipping bounds to reduce the two gradients to varying degrees. Clipping bounds can be determined by the characteristics of users' gradients (such as the norm of each layer gradients) instead of a fixed clipping bound.  A dynamic clipping bound prevents the server from reconstructing samples through gradients clipped by a fixed clipping bound.

\textbf{Malicious model component detection.}
In the proposed attack, the global model consists of an inference structure and a target model. The inference structure does not affect the performance of the target model, and it is hard to detect the malicious model component through model performance. A feasible method is to distinguish the importance of the global model to the prediction results according to different model activities in forward propagation. However, distinguishing the inference structure from the inactive part requires more discussion, and we left it as our future work.

\textbf{Secure aggregation with LDP.}
Existing work has found that secure aggregation in FL cannot fully protect user privacy~\cite{Pasquini2022secure} since encryption-based protection for gradients cannot detect malicious parts of the model and defend the proposed attack. However, secure aggregation with LDP can effectively defend against the proposed attack. Preventing malicious servers from obtaining users' gradients can significantly reduce the risk of user sample privacy leakages. However, secure aggregation algorithms incur significant communication and computational overhead, while the noise in LDP also weakens the performance of the global model. Although secure aggregation with LDP provides a strong privacy guarantee for users, it also combines the shortcomings of secure aggregation and LDP, which makes it hard to apply to FL.

\subsection{Limitations and Future Works}

\textbf{Attacks against model updates.}
Some FL mechanisms assume that users train the global model multiple rounds and generate model updates. Users only upload model updates or new parameters in this setting instead of gradients. We will discuss in future work how to reconstruct user samples from uploaded model updates or parameters.

\textbf{Better optimization objectives and metrics.}
Metric-based optimization is used to improve the quality of the reconstructed samples. However, we only consider the sample-wise mean, variance, and total variation. As discussed in other optimization-based attacks~\cite{jeon2021gradient, Yin2021see}, the adversary can collect more sample information to optimize reconstructed samples better. In addition, since the noise in the background of reconstructed samples can be eliminated by image denoising, the sample-wise indicator reduces the sample's information. Metrics calculated on masked samples can provide more accurate information for optimization.

\textbf{More training sample types.}
Evaluations only show the performance of the proposed attack on image samples. Since the proposed attack does not limit the setting of target models, it can be applied to other data types, such as text data and graph structure data. However, these data correspond to different neural networks and metrics. Due to space limitations, we do not show the effect of the proposed attack on these data. In future work, we will discuss how to reconstruct other types of training samples through the proposed attack.

\end{appendices}

\end{document}